%% file: main.tex
\documentclass[a4paper,UKenglish,cleveref, autoref, thm-restate]{lipics-v2021}
\usepackage[T1]{fontenc}

\usepackage{wrapfig}
\usepackage{graphicx}

\usepackage{hyperref}

\let\llncssubparagraph\subparagraph
\let\subparagraph\paragraph
\let\subparagraph\llncssubparagraph

\usepackage{amsthm,amsmath,amssymb,stmaryrd}
\usepackage{listings}
\usepackage{longtable}
\usepackage{tikz}
\usetikzlibrary{calc}
\usetikzlibrary{positioning}
\usetikzlibrary{automata}
\usetikzlibrary{shapes}
\tikzset{every picture/.style=thick,shorten >=0.5pt,>=stealth,auto,node distance=2.5cm,initial text=}

\let\oldexample\example
\renewcommand{\example}{\oldexample\normalfont}
\let\olddefinition\definition
\renewcommand{\definition}[1][]{%
\ifthenelse{\equal{#1}{}}{\olddefinition\normalfont}
{\olddefinition[\hbox{#1}]\normalfont}}

\input{macros}

\renewcommand\ntnote[1]{}
\renewcommand\hbnote[1]{}

\author{Mohamed H. Bandukara}{Queen Mary University of London, United Kingdom}{m.h.bandukara@qmul.ac.uk}{}{supported by EPSRC DTP EP/R513106/1.}

\author{Nikos Tzevelekos}{Queen Mary University of London, United Kingdom}{nikos.tzevelekos@qmul.ac.uk}{[orcid]}{}

\title{A Logic for Fresh Labelled Transition Systems}

\authorrunning{M. H. Bandukara \and N. Tzevelekos}

\nolinenumbers
\hideLIPIcs

\begin{document}
% \tracingall
\maketitle
% \institute{Queen Mary University of London, United Kingdom}
%
% \maketitle              % typeset the header of the contribution
%
\begin{abstract}
We introduce a Hennessy-Milner logic with recursion for Fresh Labelled Transition Systems (FLTSs).
These are nominal labelled transition systems which keep track of the history, i.e. of data values seen so far, and can capture fresh data generation.
In particular, FLTSs generalise the computations of Fresh-Register Automata, which 
in turn are one of the simplest classes of history-dependent automata operating on infinite input alphabets.
Each automaton comes equipped with a finite set of registers where it can store data values and compare them with others from the input. In addition, the automaton can accept an input just if it be fresh: not seen in the computation before.
The logic we introduce can express a variety of properties, such as the existence of an infinite path of distinct data values or the existence of a finite path where some taint property is violated.
We study the model checking problem and its complexity via reduction to parity games and, using nominal sets techniques, provide an exponential upper bound for it.
\end{abstract}

\section{Introduction}
% Story: a logic to capture fresh labelled transition systems
% Also to express bisimulation-invariant properties of fresh-register automata 

\input{intro}

% \section{Background}

\section{Nominal sets and fresh labelled transition systems (FLTSs)}
\input{nominals}

% \input{Recursion}

% \subsection{Finite Semantics \hbp{Maybe?}}
% \section{\mhml$_\S$ (Vectorial) \hbp{Maybe?}}
\section{Fresh HML with recursion (\fhml)}\label{sec:logic}
\input{logic}

\section{Parity games and model checking}\label{sec:modelchecking}
\input{modelchecking}

\section{Conclusion}

We introduced a logic to capture properties of nominal transition systems with history dependence and fresh-name generation. We showed model checking is decidable and has exponential complexity. A next step would be to parameterise this complexity, e.g.\ by bounding the number of registers in the FLTS or the arities of recursion variables. We also feel that some of our calculations are overly generous, for example in using the bound on orbits of product nominal sets to calculate the number of orbits in the bounded parity game. Our overarching aim is to implement the model checking routine and use it in applications.

\bibliography{main}
\newpage
\appendix
\section{Proofs from \cref{sec:logic}}
\input{app-logic}
\section{Proofs from \cref{sec:modelchecking}}
\input{app-modelchecking}

\end{document}

%% file: macros.tex
\newcommand{\var}{\mathit{Var}}
\newcommand{\VAR}{\mathit{VAR}}
\newcommand{\CONF}[1][\hat{\mathcal{S}}]{#1}
\newcommand\boldemph[1]{\emph{\textbf{#1}}}
\newcommand\Lfr[1]{#1^\bullet}
\newcommand\Gfr[1]{#1^\circledast}
\newcommand\modal[1]{{\bigcirc\mspace{-13mu}#1}\ }

\usepackage{graphics} %needed for the NEW quantifier
\newcommand\new[0]{\reflectbox{\ensuremath{\mathsf{N}}}}
\newcommand\fresh[1]{\new #1.}
\newcommand\nombisim{\approx_{\sf nom}}
\newcommand\bisim{\approx}

\newcommand\nt[1]{{\em\color{blue}#1}}

\newcommand\ntnote[1]{\sidenote{\nt{#1}}}
\newcommand\hbnote[1]{\sidenote{\hb{#1}}}

\newcommand\fhml[0]{$\mu$FHML}

\newcommand\fv[0]{\mathsf{fv}}
\newcommand\nomeq[0]{\sim_{\sf nom}}
\newcommand\pos[0]{\mathrm{Pos}}
\newcommand\choice[2]{\left(\begin{matrix}{#1}\\{#2}\end{matrix}\right)}
\newcommand\regindex[0]{r_i}

\newcommand\hb[1]{{\em\color{red}#1}}
\DeclareMathOperator{\dom}{dom}
\DeclareMathOperator{\rng}{rng}
\DeclareMathOperator{\supp}{supp}

\newcommand\diam[1]{\langle#1\rangle}
\newcommand\sq[1]{[#1]}
\newcommand\sem[2][\xi]{\llbracket#2\rrbracket_{#1}}

\newcommand\semp[2][\pi\cdot\xi]{\llbracket\pi\cdot #2\rrbracket_{#1}}

\newcommand\Acal{\mathcal{A}}

\newcommand\Gcal{\mathcal{G}}

\newcommand\Lcal{\mathcal{L}}

\newcommand\Ncal{\mathcal{N}}

\newcommand\Pcal{\mathcal{P}}

\newcommand\Rcal{\mathcal{R}}
\newcommand\Scal{\mathcal{S}}

\newcommand\Ucal{\mathcal{U}}

\newcommand\Xcal{\mathcal{X}}
\newcommand\Ycal{\mathcal{Y}}
\newcommand\Zcal{\mathcal{Z}}

\newcommand\PPfn{\Pcal_{\rm fin}}
\newcommand\Atoms{\mathbb{A}}
\newcommand\gap{&\,\,}
\newcommand\orb[1][]{\mathcal{O}_{#1}}
\newcommand\orbs[1][]{\mathbb{O}_{#1}}

\newcommand\cls[1]{\mathcal{C}\!\ell_{#1}}

\renewcommand\S{N}

\newcommand\fix{\mathsf{fix}}
\newcommand\ar{\mathsf{ar}}
\newcommand\lfp{\mathsf{lfp}}
\newcommand\gfp{\mathsf{gfp}}
\newcommand{\swp}[2]{(#1\ #2)}
\newcommand{\move}[2]{(#1\mapsto#2)}

\newcommand{\F}[4][\xi]{F_{#2,#1,#3,#4}}

\newcommand{\clos}[0]{\mathsf{clos}}

\newcommand\fgame[0]{\Gcal_\circledast}
\newcommand\fstates[0]{\hat{\Scal}}

\newcommand\gorb[0]{{\mathsf{Orbs}}}
\newcommand\Perm[0]{{\mathsf{Perm}}}
\newcommand\lemmapoint[2]{Lemma~\ref{#1}(\ref{#1:#2})}

\newcommand\NO[0]{\mathsf{NO}}

\newcommand\update[0]{\mathsf{UPDATE}}
\newcommand\ndbb[0]{\mathsf{ndb}}

%% file: intro.tex
Nominal Labelled Transition Systems (LTSs) can capture computational scenarios which
require unbounded sets of data values~\cite{automatanominalsets}.
In the most elementary case, data values are just \emph{names}, that is, atomic identifiers that can only be compared for equality. Nominal LTSs can be used,
for example, for modelling mobile processes~\cite{HDA1}
and computation with variables ~\cite{KAMINSKI1994329,GrumbergKS10}, and for verifying XML query languages~\cite{SCHWENTICK2007289}.
A particularly popular class of nominal LTSs are the ones generated by computations of ``regular'' automata over infinite alphabets, such as register automata~\cite{KAMINSKI1994329} and the (largely) equivalent nominal automata~\cite{automatanominalsets}. A register automaton is 
equipped with a finite set of states, and hence the nod to regularity, but also a finite set of registers where it can store data values and compare them with, and possibly update them with, data values from the input. This design allows register automata to capture languages over infinite alphabets even though their specification is finite.

% More recently,
% they have been applied to program and system modelling, leading to program semantics and verification tools [43, 47, 24] and, respectively, to learning algorithms for model learning [11, 1, 41]. As
% the machines that operate over infinite alphabets are so versatile, they remain an area eager to be studied.

Fresh-Register Automata (FRAs)~\cite{fra} expand on this design by having the option of allowing the automaton to accept a name just if it is globally fresh, that is, it has not appeared so far in the input.
FRAs
 can be used to capture computational models that use names and
name generation, e.g.\ (finitary) $\pi$-calculus processes~\cite{fra,DBLP:journals/jsa/BandukaraT23} or ML-like programs with mutable references~\cite{nominalgamesemantics}.
%Algorithmic nominal game semantics, Murawski + Tzevelekos?
 They have been applied to program verification and system modelling, leading to verification tools ~\cite{imj,nominalgamesemantics,runtimeverira} and model learning techniques~\cite{RCDL,LearningRAFreshval}.
Verification of FRAs has mainly focussed on bisimulation equivalence~\cite{fra,BFRA,deqtheory}. In this work we are interested in furthering formal
verification for nominal computation, with an interest on FRAs and more generally history-dependent computation, in the direction of modal logics.

%We extend nominal LTSs to capture histories fresh-register automata by designing a recursive Hennessy-Milner logic (HML) for them.

% Fresh-Register Automata are an extension of the Register Automata model by Kaminski and Francez [7]. Each automaton is equipped with a finite set of registers where it can store data values. Register automata can verify whether or not a given input data value (referred to as names) is stored in one of its registers, and store it in a register overwriting its current value. This design allows register automata to capture languages over infinite alphabets. Fresh-register automata expand on this by having the option of allowing the automaton to accept a name just if it is globally-fresh, that is, it has not appeared so far in the input. Fresh-register automata can be used to capture computational models that use names and
% name-generation, for example (finitary) $\pi$-calculus processes [16, 1]. To date, there are results for bisimulation
% equivalence of fresh-register automata [16, 13, 1]. In this work we are interested in furthering formal
% verification for fresh-register automata by designing a Hennessy-Milner logic (HML) for them.

Hennessy-Milner Logic (HML) is a modal behavioural logic that captures branching-time properties of labelled transition systems. It was initially introduced in 1980 by Hennessy and Milner~\cite{DBLP:conf/icalp/HennessyM80} as an alternative exposition of observational equivalence~\cite{BradfieldS07}. Observational equivalence has a focus on testing whether two systems can simulate each other in a step-by-step manner,
whereas HML logic focuses on the expressiveness of a single system. Under certain assumptions, two
processes are equivalent just if they satisfy the same HML formulas~\cite{JOB}. The extension of HML with
recursion, called modal $\mu$-calculus, was introduced by Kozen~\cite{Kozen83}.
Here, we study a nominal extension thereof, where modalities
are allowed to include names and where quantifications over some/all names and some/all fresh names are catered for.
More specifically, our contributions include:
\begin{itemize}
\item We introduce the notion of \emph{fresh LTS (FLTS)}, which extends nominal LTS with history dependence.
\item We introduce a \emph{fresh HML with recursion}, called \fhml, which allows us to express recursive branching-time properties involving names and name-freshness.
\item We study the model checking problem for \fhml\ formulas on FLTSs. This leads us to developing parity games for \fhml\ and, through a series of reductions and using nominal techniques, devising a decision procedure for them.
\item We provide (exponential) upper bound calculations for \fhml\ parity games and \fhml\ model checking, in the latter case both on  general FLTSs and on FRAs.
\end{itemize}

\paragraph*{Related work}
HML for the $\pi$-calculus was first examined in~\cite{MilnerPW93}. For properties involving paths
of unbounded length, it is natural to examine recursive extensions thereof such as the modal $\mu$-calculus
and variants thereof~\cite{BradfieldS07}. The logic we introduced is based on the $\pi$-$\mu$-calculus of Dam~\cite{Dam03} and is related to
the recent work on nominal $\mu$-calculi of Klin and collaborators~\cite{KlinL17,KlinL19,KlinHD}.
% Our points of focus are expressiveness of the logics with respect to fresh-register automata and previous nominal calculi, but also decidability, complexity and implementations of model checking.
The $\pi$-$\mu$-calculus~\cite{Dam03} specifically targets $\pi$-calculus processes, but the syntax used for recursion is the same as ours.
The study of a history-dependent nominal $\mu$-calculus in~\cite{KlinHD} is the work closest to ours. In a certain sense, our work continues and answers open problems from \emph{loc cit.}
 in particular by operating on history-dependent LTSs (and FRAs) and accommodating a vectorial version of recursion.
We also provide a more algorithmic view on model checking via parity games and produce complexity bounds. 

In the area of nominal techniques, the modal logics for nominal LTSs by Parrow et al.~\cite{NomLogics} can subsume our logic but do not allow for decidable model checking.
Logics for variants of register automata have been widely studied, e.g.~\cite{DBLP:conf/csl/Segoufin06,DBLP:journals/tocl/DemriL09,DBLP:conf/mfcs/FigueiraS09}.
% This point is taken from Bartek's paper. Refs:
% - S. Demri and R. Lazic. LTL with the freeze quantifier and register automata. ACM Trans. Comput. Log., 10(3):16:1–16:30, 2009.
% - D. Figueira and L. Segoufin. Future-looking logics on data words and trees. In Procs. MFCS 2009, volume 5734 of Lecture Notes in Computer Science, pages 331–343, 2009
% - L. Segoufin. Automata and logics for words and trees over an infinite alphabet. In Procs. CSL 2006, volume 4207 of Lecture Notes in Computer Science, pages 41–57, 2006.
The main focus of those works is on decidability, e.g.\ by restricting the number of registers allowed. No such restrictions are imposed on our logic, which is undecidable for satisfiability but our focus is instead on model checking.
There have also been a series of works on logics related to the $\pi$-calculus, e.g.~\cite{DBLP:journals/mscs/NicolaL08,DBLP:journals/cl/KoutavasH12,DBLP:conf/icalp/BergerHY08}.
Those works mainly study expressiveness and characterisation of variants of behavioural equivalence for $\pi$-calculus processes. 
Other related works, such as~\cite{DBLP:journals/iandc/Dam96,BradfieldStevens99}, expand on the modal $\mu$-calculus and study model checking, and show it to be decidable for finitary $\pi$-calculus processes. 
%Similarly, model checking of finitary $\pi$-calculus processes is decidable on \fhml\ formulas by translating to fresh-register automata~\cite{fra,DBLP:journals/jsa/BandukaraT23}.

%%% Local Variables:
%%% mode: latex
%%% TeX-master: "main"
%%% End:

%% file: nominals.tex
%\paragraph*{Nominal sets}

We start by briefly introducing the nominal sets framework, which will be the basis of much of this paper. Intuitively, nominal sets are sets where names are baked in from the outset.

Let us fix $\Atoms$ to be a countably infinite set of \boldemph{names} (or \boldemph{data values}) upon which nominal sets are built. 
    A permutation $\pi$ on $\Atoms$ (i.e.\ a bijection $\pi:\Atoms\to\Atoms$) is considered \emph{finite} if $\supp(\pi) = \{a \in \Atoms\mid \pi(a) \neq a\}$ is finite.
 For example, the identity permutation $id$ is finite ($id(a)=a$ for all $a\in\Atoms$) and, for any $a,b\in\Atoms$, so is the permutation $\swp{a}{b}$ given by:\ 
%     \[
%       \swp{a_1}{a_2}(x) =\begin{cases}
% a_{3-i} & \text{ if }x=a_i \\ x & \text{ if }x\neq a_1,a_2
%                          \end{cases}
%                        \]
$\swp{a}{b}(a) = b$, $\swp{a}{b}(b) = a$, and $\swp{a}{b}(x) = x$ for all $x\neq a,b$.  
 We denote the set of all finite permutations in $\Atoms$ as $\Perm$.
                       
\begin{definition}[\cite{PittsGabbay,Pitts}]\label{def:nomset}
    A \boldemph{nominal set} $\Xcal$ is a set $X$ equipped with an action from $\Perm$, i.e.\ a function $\_ \cdot \_: \Perm \times X \rightarrow X$ such that for all $x \in X$ and  $\pi, \pi' \in \Perm$:
    \[
    \pi \cdot (\pi' \cdot x) = (\pi \circ \pi') \cdot x,\ id \cdot x = x
    \]
    A set of names $\S\subseteq \Atoms$ \emph{supports} an element $x \in X$ if, for all $\pi \in \Perm$, if $\pi$ fixes all elements of $\S$ then $\pi\cdot x=x$. We stipulate that all elements of $X$ be \emph{finitely supported}, i.e.\ for all $x\in X$ there is some finite $\S\subseteq\Atoms$ that supports $x$. We write $\supp(x)$ for \emph{the support of $x$}, i.e.\ the least set supporting $x$. We say that $a$ \emph{is fresh for $x$}, and write $a\# x$, if $a\notin\supp(x)$.
\end{definition}

By abuse of notation, we may identify $\Xcal$ with its carrier set $X$. 
The action on a nominal set extends to products and subsets by $\pi \cdot(x,y) = (\pi \cdot x, \pi \cdot y)$ and, for any $S \subseteq \Xcal$, by $\pi \cdot S = \{\pi \cdot x \mid x \in S\}$. 
Additionally, if $\Xcal, \Ycal$ are nominal sets then so is $\Xcal \times \Ycal$, and the set of subsets of $\Xcal$ with finite support.
Accordingly, given a relation $R\subseteq\Xcal\times\Ycal$ and $\pi\in\Perm$, we have $\pi\cdot R=\{(\pi\cdot x,\pi\cdot y)\mid (x,y)\in R\}$. We call $R$ \emph{equivariant} if $\supp(R)=\emptyset$, i.e.\ for all $\pi$ and $(x,y)\in\Xcal\times\Ycal$, $(x,y)\in R\iff (\pi\cdot x,\pi\cdot y)\in R$. 

A notion in nominal sets we will be making extensive use of is that of an \emph{orbit}.

\begin{definition}
    Let $\Xcal$ be a nominal set. 
    Two elements $x, y\in \Xcal$ are \emph{nominally equivalent} (denoted $x \nomeq y$) if $\pi \cdot x = y$ for some $\pi\in\Perm$. Note $\nomeq$ is an equivalence relation.
    \\
For each $x \in \Xcal$, we define its \boldemph{orbit} $\orbs(x)$ to be its $\nomeq$-equivalence class; whereas for each finite $\S\subseteq\Atoms$ we let $\orbs[\S](x)$ be its \boldemph{$\S$-orbit}, i.e.\ its closure under permutations fixing $\S$:
    \begin{align*}
        \orbs(x) &= \{\pi \cdot x \mid \pi \in \Perm\} &
    \orbs[\S](x) &= \{\pi \cdot x \mid \forall a\in\S.\,\pi(a)=a\}
    \end{align*}
    Observe that $\orbs(x)=\orbs[\emptyset](x)$.
    Moreover, for any $S\subseteq \Xcal$ with finite support we let:
    \begin{itemize}
    \item
the {\em closure} of $S$ be\
        $\cls{}(S) = \{\pi\cdot x\mid x\in S\land\pi\in\Perm\}$;
      \item the \emph{set of orbits} of $S$ be\
   $
        \orb(S) = \{ \orbs[\supp(S)](x) \mid x \in S\}
    $.
      \end{itemize}
We say that a nominal set $\Xcal$ is \boldemph{orbit-finite} if $\orb(\Xcal)$ is finite.
If $\Xcal$ is orbit-finite, we let its \emph{register index} be\
    $
        \regindex(\Xcal) = \max\{|\supp(x)| \mid x \in \Xcal\}
        $.\ntnote{can we find a standard term for this, e.g. from Pitts}
\end{definition}

Note in particular that
$\orbs(x)=\cls{}(\{x\})$, for any $x\in\Xcal$, and
$\orb(\Xcal) = \{ \orbs(x) \mid x \in \Xcal\}$.
Moreover, for any finitely supported $S\subseteq\Xcal$ we can see by inspection that:
\[\bigcup\orb(S)=S\subseteq\cls{}(S)\subseteq\Xcal
  \]
and $\supp(\cls{}(S))=\emptyset$. Some useful albeit more technical properties follow.

\begin{restatable}{lemma}{nomlemma}\label{lem:nomlemma}
	Let $\Xcal,\Ycal$ be nominal sets and $X,Y \subseteq \Xcal$. 
	Then:
	\begin{enumerate}
		\item\label{nomlemma:orbs} $\{\orbs(x) \mid x \in X\} = \orb(\cls{}(X))$ and $|\orb(\cls{}(X))| \le |\orb(X)|$;
		\item\label{nomlemma:cls} if $X$ has empty support then $\cls{}(X \times Y) = X \times \cls{}(Y)$;
		\item\label{nomlemma:prod} if $\Xcal,\Ycal$ are orbit-finite then
		$|\orb(\Xcal \times \Ycal)| \le |\orb(\Xcal)|\cdot|\orb(\Ycal)|\cdot {n_1}^{n_2}\cdot(1 + \epsilon)$ where $\{n_1,n_2\}=\{\regindex(\Xcal), \regindex(\Ycal)\}\neq\{0\}$ and $n_1\geq n_2$,
		and $\epsilon\leq 1$ is a constant that vanishes for $n_2\geq 4$. 
	\end{enumerate}
\end{restatable}
\vspace{1em}

Orbit-finiteness essentially captures the notion of a set being ``finite up to permutation'', and it is central in the development of automata over nominal sets~\cite{automatanominalsets}.

We are going to define transition systems which retain a possibly unbounded history component but are otherwise finite. For that purpose, orbit-finiteness needs to be relaxed to accommodate histories of unbounded size. One part of our solution is to use configurations comprising a state and a history, where only the set of states is orbit-finite. The other part is to stipulate that the names appearing only in the history component of a configuration are in between them indistinguishable. This can be captured by the following notion.

\begin{definition}
  Given a nominal set $\Xcal$ and $x\in\Xcal$, for each $\S'\subseteq\S\subseteq\Atoms$, we say that the pair \emph{$(\S',\S)$ freshly supports $x$} if, for all $\pi\in\Perm$:
  \[ (\forall a\in\S'.\,\pi(a)=a\land\pi\cdot (\S\setminus\S')=\S\setminus\S')
    \implies \pi\cdot x=x.
  \]
\end{definition}

This notion of support refines ordinary support and, to the best of our knowledge, is novel. In particular, if $x$ is freshly supported by $(\S',\S)$ then it is supported by $\S$. Moreover, $x$ is freshly supported by $(\S,\S)$ iff it is supported by $\S$.

% There is also a notion of least finite fresh support.

% \begin{lemma} Given a nominal set $\Xcal$, each 
%  $x\in\Xcal$ is freshly supported by $(\bigcap)$.
% \end{lemma}

We can now move to our labelled transition systems. These are going to be nominal LTSs~\cite{automatanominalsets} (i.e.\ with states and labels being nominal sets) where each state is paired with a \emph{history} that includes the names (in the support) of the state. The history is updated in each transition to include any names of the label of said transition. 

\begin{definition}
A \emph{Fresh Nominal Labelled-Transition System (FLTS)}
is a tuple $\mathcal{L}=\langle \mathcal{S},L,\rightarrow\rangle$, where:
\begin{itemize}
\item $\mathcal{S}$ is a nominal set of \emph{states}
and $L$ is a nominal set of \emph{actions},
\item    ${\rightarrow}\subseteq \CONF\times L\times\CONF$ is an equivariant \emph{transition relation},
  \item 
  where
 $\CONF=\{(s,H)\in\Scal\times\PPfn(\Atoms)\mid \supp(s)\subseteq H\}$ is the nominal set of \emph{configurations};
\end{itemize}
and such that, for all $(s,H) \in \CONF$:
 \begin{itemize}
     \item 
 $\{(\ell, (s', H'))\mid (s, H) \xrightarrow{\ell} (s', H')\}$ is freshly supported by $(\supp(s),H)$,
\item
for all $(s, H) \xrightarrow{\ell} (s', H')$,\
$\supp(s')\subseteq \supp(s)\cup \supp(\ell)$ and 
$H'=H\cup \supp(\ell)$.
 \end{itemize}
 $\mathcal{L}$ is called \emph{fresh-orbit-finite} if $\cal S$ and $L$ are orbit-finite.
\end{definition}

This scenario applies to a range of examples, such as name-carrying processes (e.g.\ in the $\pi$-calculus) and name-manipulating programs (e.g.\ Java programs with objects). At the automata level, a member formalism is that of \emph{fresh-register automata}~\cite{fra}, which are the baseline paradigm in our study. As we are interested in behavioural properties, rather than language acceptance, we disregard initial and final states. %Below, we let $r$ be a natural number and write $[1,r]$ for the set $\{1,\dots,r\}$.
The automata we define next operate on inputs $(t,a)\in\Sigma\times\Atoms$ where $t$ comes from a finite set of \emph{tags} and $a$ is a name.  

\begin{definition}\label{def:FRAs}
  An $r$-\boldemph{Fresh-Register Automaton ($r$-FRA)} is a tuple $\Acal=\langle Q,\mu,\Sigma,\delta\rangle$, where $r\in\mathbb{N}$ is the number of \emph{registers} in $\Acal$ and:
  \begin{itemize}
  \item $Q$ is a finite set of \emph{states}, and $\mu:Q\to\Pcal([1,r])$ an \emph{availability function};
  \item $\Sigma$ is a finite set of {tags}, and
 $\delta:Q\times\Sigma\times\{i,\Lfr{i},\Gfr{i}\mid 1\leq i\leq r\}\times Q$ is the transition relation;
    \end{itemize}
    subject to the following conditions:
    \begin{itemize}
    \item if $(q,t,i,q')\in\delta$ then $i\in\mu(q)$ and $\mu(q')\subseteq\mu(q)$;
    \item
      if $(q,t,\Lfr i,q')\in\delta$ or $(q,t,\Gfr i,q')\in\delta$ then $\mu(q')\subseteq\mu(q)\cup\{i\}$.
    \end{itemize}
    We shall write $q\xrightarrow{t,x}q'$ for $(q,t,x,q')\in\delta$.
  \end{definition}

  We can see that the labels appearing in FRA transitions are of the form $(t,x)$, where $t$ comes from a finite alphabet while $x$ is a variant of a register index $i$.
Assuming the automaton is at state $q$ and given $q\xrightarrow{t,x}q'$, 
depending on $x$ the automaton will perform one of the following actions before moving to state $q'$:
  \begin{itemize}
  \item $x=i$: read the input $(t,a)$, where $a$ is the name currently in register $i$\,---\,note here that $i$ needs to be available at $q$ (i.e.\ $i\in\mu(q)$);
  \item $x=\Lfr{i}$: read any input $(t,a)$ so long $a$ is not currently stored in any register (i.e.\ it is \emph{locally fresh}), and store $a$ in register $i$;
  \item $x=\Gfr{i}$: read any input $(t,a)$ so long $a$ is has not be seen before (i.e.\ it is \emph{globally fresh}), and store $a$ in register $i$.
  \end{itemize}
Formally, the semantics of an FRA can be given in terms of its derived FLTS.
The latter features states of the form $(q,\rho)$ where $q$ is a state and $\rho$ an assignment of values to available registers, i.e.\ the ones in $\mu(q)$.
It is good to bear in mind that $q$ has empty support.

\begin{definition}
  Given an $r$-FRA $\Acal=\langle Q,\mu,\Sigma,\delta\rangle$ we define its FLTS $\Gcal_\Acal=\langle \Scal_\Acal,\Sigma\times\Atoms,\to\rangle$ by setting (and writing configurations simply as $(q,\rho,H)$ instead of $((q,\rho),H)$):
  \begin{itemize}
    \item 
$ \Scal_\Acal = \{(q,\rho)\mid q\in Q\land \rho:\mu(q)\to\Atoms\text{ injective}\}
$;
\item
  $(q,\rho,H)\xrightarrow{t,a}(q',\rho',H')$ whenever there is $q\xrightarrow{t,x}q'$ such that, for some register $i$:
  \begin{itemize}
    \item 
$x=i$ and $a=\rho(i)$, while $\rho'=\rho$; or
\item
$x=\Lfr i$ and $a\notin\rng(\rho)$, while $\rho'=\rho[i\mapsto a]\upharpoonright\mu(q')$; or
\item
$x=\Gfr i$ and $a\notin H$, while $\rho'=\rho[i\mapsto a]\upharpoonright\mu(q')$;
\end{itemize}
where $\rho[i\mapsto a]$ is the update of $\rho$ mapping $i$ to $a$, and the operator $\upharpoonright$ performs domain restriction. In all cases, $H'=H\cup\{a\}$.
\end{itemize}
\end{definition}

It is not difficult to verify that $\Gcal_\Acal$ is indeed an FLTS. First, the sets $\Scal_\Acal$ and $\Sigma\times\Atoms$ are closed under permutation, hence they are nominal sets. The definition of the transition relation does not single out any specific names and is therefore equivariant. What is more important, the names appearing in the history but not in the registers of a configuration are indistinguishable to the automaton, and that makes the fresh-support condition true. Finally, in every transition, names in the final register assignment are sourced from those of the initial assignment and the label, and the history is correctly updated.

We conclude this section with some motivating examples. These include example of properties used in the literature and simple properties one may want to check in a verification scenario.

%A fresh path is a path where all names are distinct.

\begin{example}\label{example:freshPaths}
  Let us assume a unique tag, e.g.\ $\Sigma=\{\bullet\}$, which we suppress for brevity.  
  Consider the following properties (cf.~\cite{KlinHD}):
  \begin{align*}
    &\text{\em All paths, finite and infinite, contain pairwise distinct names.}&\tag{\textsc{\#All}}\\
    &\text{\em There is an infinite path where all names are distinct.}&\tag{\textsc{\#Path}}
\end{align*}
Thus, both properties specify paths of the forms\ $a_1 a_2 \dots a_n$\ or\ $a_1 a_2a_3 \dots$\
where, for each $a_i$, there is no $j < i$ such that $a_i = a_j$. Let us now examine the following basic FRAs.

\noindent
\begin{centering}
\begin{tikzpicture}
  \node[state,initial,minimum size=5mm]   (q0) {$q_0$};
  \node[minimum size=5mm]   (r0) [below=1mm of q0] {$\emptyset$};
  \path[->] (q0) edge [loop right] node [right] {$\Gfr{1}$} (q0);
\end{tikzpicture}
\begin{tikzpicture}
  \node[state,initial,minimum size=5mm]   (q0) {$q_0$};
  \node[minimum size=5mm]   (r0) [below=1mm of q0] {$\emptyset$};
  \node[state,minimum size=5mm]   (q1) [right of=q0] {$q_1$};
  \node[minimum size=5mm]   (r1) [below=1mm of q1] {$\{1\}$};
  \path[->] (q1) edge [loop right] node [right] {$\Lfr{1}$} (q1);
  \path[->] (q0) edge node [above] {$\Lfr{1}$} (q1);
\end{tikzpicture}
\begin{tikzpicture}
  \node[state,initial,minimum size=5mm]   (q0) {$q_0$};
  \node[minimum size=5mm]   (r0) [below=1mm of q0] {$\emptyset$};
  \node[state,minimum size=5mm]   (q1) [right of=q0] {$q_1$};
  \node[minimum size=5mm]   (r1) [below=1mm of q1] {$\{1\}$};
  \path[->] (q1) edge [loop right] node [right] {$\Gfr{1}$} (q1);
  \path[->] (q0) edge node [above] {$\Lfr{1}$} (q1);
\end{tikzpicture}
% \begin{tikzpicture}
%   \node[state,initial,minimum size=5mm]   (q0) {$q_0$};
%   \node[state,accepting,minimum size=5mm] (q1) [right of=q0] {$q_1$};
%   \node[state,minimum size=5mm] (q2) [right of=q1] {$q_2$};

%   \path[->] (q0) edge [loop above] node [above] {$a$} (q0);
%   \path[->] (q0) edge [loop below] node [below] {$b$} (q0);
%   \path[->] (q0) edge node [above] {$b$} (q1);
%   \path[->] (q1) edge node [above] {$a$} (q2);
% \end{tikzpicture}
\end{centering}
\noindent
The sets beneath each state indicate available registers. The first FRA will engage in paths which always have distinct names, so it satisfies both \textsc{\#All} and \textsc{\#Path}. In fact, all three FRAs, from all possible configurations, satisfy \textsc{\#Path} as the locally fresh transitions can be realised by names that are in fact globally fresh. Moreover:
\begin{itemize}
\item the second FRA fails \textsc{\#All} from both states: for any $H,\rho,a,b$ with $a\#\rho, b$ and $i=0,1$, $(q_i,\rho,H)\xrightarrow{a}(q_1,[1\mapsto a],H\cup\{a\})\xrightarrow{b}(q_1,[1\mapsto b],H\cup\{a,b\})\xrightarrow{a}(q_1,[1\mapsto a],H\cup\{a,b\})$;
\item the third FRA satisfies \textsc{\#All} everywhere: even if going from $q_0$ to $q_1$ may involve an old name (already in history), the remaining names are all going to be globally fresh.
\end{itemize}
\end{example}

\ntnote{\scriptsize maybe present properties as above, giving them a name. It would be good to also include a taint example with freshness. Not sure about having FRAs for them as they take some space. If we do include FRAs then better not do Taint}
\begin{example}\label{example:sessions}
Let us assume a set of tags $\Sigma = \{\mathsf{S}, \mathsf{U}, \mathsf{T}\}$ which stand for $\mathsf{start}$, $\mathsf{use}$ and $\mathsf{terminate}$ respectively.
Consider an application in which a new session is created each time the application is launched.
This session can be used an arbitrary number of times, and then the session is terminated once the application is closed. 
We declare the following property for this application:
\[\begin{aligned}
	&\text{\em There is an infinite path that repeatedly starts a fresh session, arbitrarily uses it}\\
	&\text{\em and then terminates it.}
\end{aligned}\tag{\textsc{\#SUT}}\]
Once the session is terminated, it cannot be resumed or reactivated, a new session must be created.
A path for this will look as follows:
\[
	\mathsf{S}(s_1)\ \mathsf{U}(s_1)^*\ \mathsf{T}(s_1)\ \mathsf{S}(s_2)\ \mathsf{U}(s_2)^*\ \mathsf{T}(s_2)\ \mathsf{S}(s_3)\ \dots 
\]
where, for each $s_i$, there is no $j < i$ such that $s_i = s_j$.
Let us now consider this basic FRA:

\begin{centering}
 \begin{tikzpicture}
 	\node[state,initial,minimum size=5mm]   (q0) {$q_0$};
 	\node[minimum size=5mm]   (r0) [below=1mm of q0] {$\emptyset$};
 	\node[state,minimum size=5mm]   (q1) [right of=q0] {$q_1$};
 	\node[minimum size=5mm]   (r1) [below=1mm of q1] {$\{1\}$};
 	\path[->] (q1) edge [loop right] node [right] {$\mathsf{U}, {1}$} (q1);
 	\path[->] (q0) edge node [above] {$\mathsf{S}, \Gfr{1}$} (q1);
 	\path[->] (q1) edge [bend left] node [below] {$\mathsf{T}, 1$} (q0);
 \end{tikzpicture}
\end{centering}

\noindent
The FRA satisfies the property $\textsc{\#SUT}$ from $q_0$: for any $H, s_1,s_2$ with $s_1,s_2\notin H$, $H'=H\cup\{s_1\}$ and $H'' = H' \cup \{s_2\}$:

\begin{centering}
	 \begin{tikzpicture}[node distance=3cm]
		\node[minimum size=5mm]   (q0) {$(q_0, \emptyset, H)$};
		\node[minimum size=5mm]   (q1) [right of=q0] {$(q_1, [1\mapsto s_1], H')$};
		\node[minimum size=5mm]   (q2) [right of=q1] {$(q_0, \emptyset, H')$};
		\node[minimum size=5mm]   (q3) [right of=q2] {$(q_1, [1\mapsto s_2], H'')$};
		\node[minimum size=5mm]   (q4) [right of=q3] {$\dots$};
		\path[->] (q0) edge node [above] {$\mathsf{S}, s_1$} (q1);
		\path[->] (q1) edge [loop above] node [above] {$\mathsf{U}, s_1$} (q1);
		\path[->] (q1) edge node [above] {$\mathsf{T}, s_1$} (q2);
		\path[->] (q2) edge node [above] {$\mathsf{S}, s_2$} (q3);
		\path[->] (q3) edge [loop above] node [above] {$\mathsf{U}, s_2$} (q3);
		\path[->] (q3) edge node [above] {$\mathsf{T}, s_2$} (q4);
	\end{tikzpicture}
\end{centering}
\end{example}
%
%\begin{example}\
%Consider 
%A taint example (with freshness?)    
%\end{example}
%
%\begin{example} Path freshness:
%All finite and infinite paths are fresh paths.
%There exists a path that is fresh. etc.
%\end{example}

%%% Local Variables:
%%% mode: latex
%%% TeX-master: "main"
%%% End:

%% file: logic.tex
% \section{Nominal \titlemu HML}
In this section we present our logic \fhml.
We introduce the syntax and semantics of formulas, and show that the semantics is well defined. 
%For recursion, fixpoint operators are used and so we additionally show that least fixpoints exist. 
% We begin by providing the syntax for nominal \mhml\ formulas.
Recall the sets $\Atoms,\Sigma$ of {names} and tags, where the former is ranged over by $a, b,$ etc.
It will be useful to consider tags as having arbitrary finite {arities}, and not just unary as in FRAs, determined by a map $\ar:\Sigma\to\mathbb{N}$.
% ; the tags used in FRAs in particular have by definition arity 1.
Let us also assume a set $\var$ of \emph{value variables} and a set $\VAR$ of \emph{recursion variables}. These are ranged over respectively by
$x,y$, etc.\ and
$X,Y$, etc.
Recursion variables have also {arities}, given by a map $\ar:\VAR\to\mathbb{N}$ (note function overload). Here and in the sequel, when writing variable sequences $\vec x$ we shall tacitly assume they contain distinct elements.

\begin{definition}\label{def:mhmlr}
The formulas of  fresh HML with recursion (\fhml) are:
\begin{align*} 
  \text{Formulas} \ni\ \phi &::=\ u=u\mid \phi\lor\phi\mid \lnot\phi\mid \bigvee\nolimits_{x}\phi
  \mid \fresh{x}\phi\mid \diam{\ell}\phi\mid (\mu X(\vec x).\phi)(\vec u)\mid X(\vec u)\\
  \text{Values} \ni\ u &::=\ x \mid a\\
  \text{Labels} \ni\ \ell &::=\ (t,\vec u)\qquad (\text{with }\ar(t)=|\vec u|)
\end{align*}
where
in $X(\vec u)$ and $(\mu X(\vec x).\phi)(\vec u)$ we have $|\vec u|=|\vec x|=\ar(X)$.
\end{definition}

A variable $X$ can be free or bound, with binding performed with $\mu X$.
 As usually, we impose that each recursion variable $X$ appears within an even number of negation operations from its binder, implying that there must be at least one appearance of every such $X$.  
Value variables $x$ are bound by $\bigvee_{x}$, $\fresh{x}$ and $\mu X(\vec x)$. 
We say that a formula is \boldemph{firm} if it has no free value variables, and \boldemph{closed} if it has no free recursion variables.

The size of a formula is calculated in such a way that it be no smaller than the size of its support, and moreover, be linearly related to a concise machine representation of the formula (e.g.\ using de Bruijn indices~\cite{DEBRUIJN1972381}).

\begin{definition}\label{def:formSize}
    The size for each formula is inductively defined as:
  \begin{align*}
    |u=v| &= 2 &
    |\bigvee\nolimits_x\phi| = |\fresh{x}\phi| = |\lnot\phi| &= 1+|\phi|\\
    |\diam{t, \vec u}\phi| &= 1 + |\vec u| + |\phi| &
    | X(\vec u)| &= 1+|\vec u|\\
    |\phi\lor\psi| &= 1+|\phi| + |\psi| &
    | (\mu X(\vec x).\phi)(\vec u)|&= 1+|\phi|+2|\vec u|
    \end{align*}
\end{definition}

The semantics of formulas is given with respect to a given FLTS: a formula is mapped to a subset of its configurations. We shall use the following grammars for value and recursion variable substitutions respectively:
\begin{align*}
\gamma\ &::=\ \varepsilon\ \mid\ \gamma,\{\vec a/\vec x\} %\mid \theta[\sigma X(\vec x).\phi/X]
&
\theta\ &::=\ \varepsilon\  \mid\ \theta,\{\sigma X(\vec x).\phi/X\}
\end{align*}
Thus, substitutions are sequences of mappings which can be on sequences of distinct value variables or single recursion variables.
We may write $\dom(\gamma)$ for the union of the domains of all elements in $\gamma$, while in $\gamma,\{\vec a/\vec x\}$ we assume that $|\vec a|=|\vec x|$ and, for each $i$, $x_i\notin\dom(\gamma)$. A formula $\phi$ acted on by a substitution $\gamma$ is written $\phi\{\gamma\}$ and defined recursively by:
\begin{align*}
  \phi\{\varepsilon\} &= \phi &
  \phi\{\gamma,\{\vec a/\vec x\}\} &=  (\phi\{\gamma\})\{\vec a/\vec x\}
\end{align*}
Similar conventions and definitions relate to recursion variable substitutions.
% While value substitutions are maps, recursion variable substitutions are sequences of \emph{maplets}, i.e.\ singleton maps.
In addition, $\phi\{\theta\}$ is only ever going to be considered in cases where no variable capture be possible.

\begin{definition}\label{def:hdsemantics}
Given an FLTS $\mathcal{L}=\langle \mathcal{S},L,\rightarrow\rangle$ with $L\subseteq\{(t,\vec a)\in\Sigma\times\Atoms^*\mid \ar(t)=|\vec a|\}$, let us set\
$\mathcal{U}=\Pcal(\fstates)$\  and, for each $n \in \mathbb{N}$, \ 
    $\Ucal_n = \Atoms^n \to \Ucal$.\\
  A \emph{$\mathcal{U}$-variable assignment} is a finite partial map  $\xi:\VAR\rightharpoonup\bigcup_{n\in\mathbb{N}}\Ucal_n$ such that, for each $X\in\dom(\xi)$, $\xi(X)\in\mathcal{U}_{\ar(X)}$ ($=\Atoms^{\ar(X)} \to \Ucal$). Given a $\mathcal{U}$-variable assignment $\xi$, the semantics
$\sem{\phi}$
  of a firm formula $\phi$ with respect to $\xi$ is an element of $\Ucal$ given inductively by:
\begin{align*}
  \sem{a=b} &= \emptyset &
    \sem{\bigvee\nolimits_{x}\phi} &= \bigcup\nolimits_{a\in\Atoms}\sem{\phi\{a/x\}}\\
  \sem{a=a} &= \fstates &
    \sem{\fresh{x}\phi} &= \bigcup\nolimits_{a\#\phi,\xi} \{(s, H) \in \sem{\phi\{a/x\}} \mid a \notin H\}\\
  \sem{\phi_1\lor\phi_2} &= \sem{\phi_1}\cup\sem{\phi_2} &
  \sem{X(\vec a)} &= \xi(X)(\vec a)\\
  \sem{\lnot\phi} &= \fstates\setminus\sem{\phi} &
    \sem{(\mu X(\vec x).\phi)(\vec a)} &= (\lfp(\lambda f.\lambda\vec b. \sem[{\xi[X\mapsto f]}]{\phi\{\vec b/\vec x\}}))(\vec a)\\
      \sem{\diam{\ell}\phi} &\multicolumn{3}{l}{${}=\{(s, H)\in\fstates\mid \exists (s, H)\xrightarrow{\ell}(s', H').\,(s', H')\in\sem{\phi}\}$}
\end{align*}
When $\xi$ is empty, we may simply write $\sem[]{\phi}$.
\end{definition}

Note that the notation $\lambda\vec b. \sem{\phi\{\vec b/\vec x\}}$ is shorthand for the function $\{(\vec b,\sem{\phi\{\vec b/\vec x\}})\mid \vec b\in\Atoms^n\}$. 
More generally, the semantics is only defined on firm formulas.

For the semantics to be well-defined, we need to ensure that the required least fixpoints exist. From the Knaster-Tarski theorem,
this can be done by showing that $\Ucal_n$ is a complete lattice and that $\lambda f.\lambda \vec b.\sem[{\xi[X \mapsto f]}]{\phi\{\vec b / \vec x\}}$ is a monotone function. In particular, setting $F=\lambda f.\lambda \vec b.\sem[{\xi[X \mapsto f]}]{\phi\{\vec b / \vec x\}}$, we have\
$
\lfp(F) = \bigsqcap \{g\in \Ucal_n\mid 
F(g)\sqsubseteq g\}
$,
where order and least-upper-bounds in $\Ucal_n$ are given as expected:
\begin{align*}
f\sqsubseteq g &\iff \forall\vec a\in\Atoms^n.\,f(\vec a)\subseteq g(\vec a)
&
    \bigsqcap S &= \lambda\vec a.\,\bigcap\{f(\vec a)\mid f\in S \}
\end{align*}
for any $f,g\in\Ucal_n$ and $S\subseteq\Ucal_n$.

\begin{restatable}{lemma}{completelattice}\label{lem:completelattice}
% \begin{lemma}\label{lem:completelattice}
$\Ucal_n$ is a complete lattice and,
for any $\phi, \xi$,
 the function $\lambda f.\lambda \vec b.\sem[{\xi[X\mapsto f]}]{\phi\{\vec b/\vec x\}}$ is monotone.
% \end{lemma}
\end{restatable}

Finally, we show that for any \fhml\ formula $\phi$ and $\Ucal$-variable assignment $\xi$,  $\supp(\sem{\phi}) \subseteq \supp(\phi) \cup \supp(\xi)$. 
To do this, we show that the function $\lambda \phi, \xi.\sem{\phi}$ is equivariant, and then use the equivariance principle.

\begin{restatable}{lemma}{semanticsequivariant}\label{lem:semanticsequivariant}
    The function $\lambda \phi, \xi.\sem{\phi}$ is equivariant.
    Therefore, 
for any formula $\phi$ and $\Ucal$-variable assignment $\xi$, $\supp(\sem{\phi}) \subseteq \supp(\phi) \cup \supp(\xi)$.
\end{restatable}
% \begin{proof}
%     The proof is similar to Lemma~\ref{lem:equivariant}, except for the following additional rule. 
%     We show the new inductive step as follows:
%     \begin{itemize}
%         \item $\phi=\fresh{x}\phi'$:
%         In this case, $\semp{(\fresh{x}\phi')} = \bigcup_{a\in\Atoms}\{(s, H) \in \semp{(\phi'\{a/x\})} \mid a \notin H\}$. 
%         By inductive hypothesis, the latter is $\bigcup_{a\in\Atoms}\{(s, H) \in \pi \cdot \sem{(\phi'\{a/x\})} \mid a \notin H\}
%         % =\bigcup_{a\in\Atoms}\{\pi \cdot (s, H) \in \sem{(\phi'\{a/x\})} \mid \pi(a) \notin H\}
%         =\pi \cdot \bigcup_{a\in\Atoms}\{(s, H) \in \sem{(\phi'\{a/x\})} \mid a \notin H\}$.
%     \end{itemize}
% \end{proof}

% \begin{lemma}
%     For any formula $\phi$ and $\Ucal$-variable assignment $\xi$, $\supp(\sem{\phi}) \subseteq \supp(\phi) \cup \supp(\xi)$.
% \end{lemma}
% \begin{proof}
%     By the equivariance principle.
% \end{proof}

A direct consequence of the latter result is that $\sem{\phi}$ is finitely supported, and so is the function $\lambda f.\lambda \vec b.\sem[{\xi[X \mapsto f]}]{\phi\{\vec b/\vec x\}}$.

\begin{restatable}{theorem}{semanticswelldefined}\label{thm:semanticswelldefined}
The semantics of \cref{def:hdsemantics} is well defined. Its image is within the subset of $\Ucal$ of finitely-supported sets of configurations.
\end{restatable}

\begin{remark}[Derived connectives]\label{rem:derived}
  We can define standard dual connectives for \fhml\ as:
  \begin{align*}
    \phi\land\psi &\equiv \lnot(\lnot\phi\lor\lnot\psi) &
    \sq{\ell}\phi &\equiv \lnot\diam{\ell}\lnot\phi &                        (\nu X(\vec x).\phi)(\vec u) &\equiv \lnot(\mu X(\vec x).\lnot\phi\{\lnot X/X\})(\vec u)
  \end{align*}
  and write $u\neq v$ for $\lnot(u=v)$. The semantics of the above can be derived from \cref{def:hdsemantics} as expected. For instance:\
$
  \sem{(\nu X(\vec x).\phi)(\vec a)} = (\gfp(\lambda f.\lambda\vec b. \sem[{\xi[X\mapsto f]}]{\phi\{\vec b / \vec x\}}))(\vec a)
  $.\\
  In the same vein,
  we show below that fresh selection is its self-dual:
  $\sem{\lnot\fresh{x}\phi}=\sem{\fresh{x}\lnot\phi}$.
%   \begin{align*}
%     (s,H)\in\sem{\lnot\fresh{x}\phi} &
%     \iff (s,H)\notin\sem{\fresh{x}\phi} 
% \iff\forall a\#\phi,\xi.\,(s,H)\notin\sem{\phi\{a/x\}}\lor a\in H\\
% &\iff\forall a\#\phi,\xi,H.\,(s,H)\notin\sem{\phi\{a/x\}} \\
% &\implies\exists a\#\phi,\xi,H.\,(s,H)\notin\sem{\phi\{a/x\}}
% \iff\exists a\#\phi,\xi,H.\,(s,H)\in\sem{\lnot\phi\{a/x\}}\\
% &\implies\exists a\#\phi,\xi.\,(s,H)\in\{(s',H')\in\sem{\lnot\phi\{a/x\}}\mid a\not\in H'\}\\
% &\iff (s,H)\in\sem{\fresh{x}\lnot\phi}\\
% (s,H)\in\sem{\fresh{x}\lnot\phi} &
% \iff \exists a\#\phi,\xi,H.\,(s,H)\in\sem{\lnot\phi\{a/x\}}
% \iff \exists a\#\phi,\xi,H.\,(s,H)\notin\sem{\phi\{a/x\}}\\
% &\implies\exists a\#\phi,\xi,H.\forall b\#\phi,\xi,H.\,(s,H)=\swp{a}{b}\cdot(s,H)\notin\sem{\phi\{b/x\}}\\
% &\implies\forall b\#\phi,\xi,H.\,(s,H)\notin\sem{\phi\{b/x\}}
% \iff(s,H)\in\sem{\lnot\fresh{x}\phi}
%   \end{align*}
\end{remark}

As in~\cite{PittsGabbay,Pitts}, we can show that $\fresh{x}\phi$ can be taken as meaning ``for any fresh $x$, $\phi$ holds'' or ``there is some fresh $x$ such that $\phi$ holds''. 

\begin{lemma}
  For any $\phi,\xi$ and configuration $(s,H)$,
  \begin{align}
    (s,H)\in\sem{\fresh{x}\phi} 
    \iff \exists a\#\phi,\xi,H.\,(s,H)\in\sem{\phi\{a/x\}}\label{NEWone}\\
\iff \forall a\#\phi,\xi,H.\,(s,H)\in\sem{\phi\{a/x\}}\label{NEWtwo}
  \end{align}
Therefore, we have $\sem{\lnot\fresh{x}\phi}=\sem{\fresh{x}\lnot\phi}$.
\end{lemma}
\begin{proof}
  We note that \eqref{NEWone} is a restatement of the definition of the semantics for $\fresh{x}\phi$. For \eqref{NEWtwo}
it suffices to show the left-to-right inclusion, as the other one follows from there always being some $a\#\phi,\xi,H$.
  For any $(s,H)$:
  \begin{align*}
(s,H)\in\sem{\lnot\fresh{x}\phi} &
\iff\exists a\#\phi,\xi,H.\,(s,H)\in\sem{\phi\{a/x\}}\\
&\implies\exists a\#\phi,\xi,H.\forall b\#\phi,\xi,H.\,(s,H)=\swp{a}{b}\cdot(s,H)\in\sem{\phi\{b/x\}}\\
 &\implies \forall b\#\phi,\xi,H.\,(s,H)\in\sem{\phi\{b/x\}}
  \end{align*}
  For the final claim of this Lemma, for each $(s,H)$ we have:
    \begin{align*}
    (s,H)\in\sem{\lnot\fresh{x}\phi} &
    \iff (s,H)\notin\sem{\fresh{x}\phi} 
\iff\forall a\#\phi,\xi,H.\,(s,H)\notin\sem{\phi\{a/x\}}\\
&\iff\forall a\#\phi,\xi,H.\,(s,H)\in\sem{\lnot\phi\{a/x\}}
\iff (s,H)\in\sem{\fresh{x}\lnot\phi}
  \end{align*}
\end{proof}

\begin{remark}[Comparison with logics of Dam~\cite{Dam03} and Klin et al.~\cite{KlinL19,KlinHD}]
  The $\pi$-$\mu$-calculus is a recursive modal logic for $\pi$-calculus, accompanied with a proof system that is sound and complete for finitary processes~\cite{Dam03}.
% The focus in \emph{loc cit.} is less on decidability and more in proof systems.
  Its connectives are very similar to ours and, in fact, our syntax for the recursive connectives is taken from \emph{loc cit}. Being a logic specifically for the $\pi$-calculus, it employs connectives such as free/bound input and output, in addition to modal connectives that specify the communication channel. Bound I/O in particular makes fresh-name quantification redundant. 
  \\
  The $\mu$-calculi with atoms of~\cite{KlinL19,KlinHD} are the probably closest to our logic. While those logics only feature recursion variables of nullary arities, the vectorial version studied in~\cite{KlinL19} can capture scenarios where non-zero arities are needed. The history-dependent logic of~\cite{KlinHD} captures freshness with this connective and corresponding semantics (in our terminology):
  \begin{align*}
\phi &::= \cdots \mid \#v\quad \text{(for value $v$)}
    &
\sem{\#a} &= \{(s,H)\mid a\notin H\}
  \end{align*}      
The analogy is not fully exact as our semantics is over history-dependent LTSs, while in \emph{loc cit.} on nominal LTSs with the history being attached to the semantics function\,---\,but that should be sufficiently close for our discussion here. Using $\#u$, we can define fresh quantification:
  \[
\fresh{x}\phi \equiv \bigvee\nolimits_x\#x\land(x\neq v_1)\land\cdots\land(x\neq v_n)\land\phi
\]
where $v_1,\dots,v_n$ is an enumeration of all (free) values in $\phi$.
But one can be even more expressive. E.g.\ the following formula stipulates that there are exactly two names in the current history:
\[
  \bigvee\nolimits_{x,y}x\neq y\land\lnot(\#x\lor\#y)\land
  \bigwedge\nolimits_{z}z\neq x,y\implies\#z
\]
We find this added expressivity not necessary for specifying FLTSs built e.g.\ from FRAs, as the exact size of the history is in general not accessible. The weaker history-dependence coming from the fresh-quantifier allows for a somewhat simpler notion of bounded parity game in the next section, compared to~\cite{KlinHD}. 
\end{remark}

By reduction from satisfiability for the $\mu$-calculus with atoms~\cite{KlinL19}, we know that the satisfiability problem is undecidable for \fhml. In the next section we shall study instead the model checking problem, show it is decidable and calculate an upper bound for its complexity.
The following definition will be of use then.

\begin{definition}
	Given a \fhml\ formula $\phi$, its bounding depth $\|\phi\|$ is given recursively by:	
	\begin{align*}
		\|u = v\| &= \|X(\vec u)\| = 0  &\|\lnot \phi\| &= \|\diam{\ell}\phi\| = \|\phi\| &\|(\mu X(\vec x).\phi)(\vec u)\| = \|\phi\| + |\vec x|\\
		\|\phi \lor \psi\| &= \max(\|\phi\|, \|\psi\|) & \|\bigvee\nolimits_x \phi\| &= \|\fresh{x} \phi\| = 1 + \|\phi\|   
	\end{align*}
\end{definition}

% <<<<<<< HEAD
%\paragraph*{Examples revisited}
We conclude this section by revisiting \cref{example:freshPaths,example:sessions} and building formulas for them.
% where all paths (finite and infinite) have distinct names.

\begin{example} The property \textsc{\#All}
 can be represented by the formula $\lnot \phi$, where:
% =======
% \paragraph*{Examples revisited}
% We revisit Example~\ref{example:freshPaths}, where all paths (finite and infinite) have distinct names ($\#\textsc{SUT}$).
% This can be represented by the formula $\lnot \phi$, where:
% >>>>>>> refs/remotes/origin/main
\[
	\phi = \mu X.\bigvee\nolimits_x\diam{x}(X\lor\mu Y.\bigvee\nolimits_y \diam{y}(Y\lor x=y))
\]
Note we omit tags as they are not used here.
The subformula $\mu Y.\bigvee\nolimits_y \diam{y}(Y\lor x=y)$ represents all finite paths that eventually accept an $x$, that is, for each call of $Y$ we accept any name $y$ and either  we repeat $Y$ or, in case $x = y$, we stop. 
%
%In this case, $Y$ is a $\mu$-variable as we require the fixpoint stabilises once an action has been repeated.
Similarly, $\phi$ represents all finite paths where eventually, some $a$ is repeated, that is, for each call of $X$ we accept any name $x$ and either we 
repeat $X$, or $x$ is eventually accepted again (using $\mu Y.\bigvee\nolimits_y \diam{y}(Y\lor x=y)$).

Thus, $\lnot \phi$ represents all paths such that no name is ever repeated.
Note that, by negation elimination, $\lnot \phi$ can be written as:
$
\nu X.\bigwedge\nolimits_x\sq{x}( X\land \nu Y.\bigwedge\nolimits_y \sq{y}(Y\land x\neq y))
$.
\end{example}

\begin{example}
Let us consider \cref{example:sessions}, where we start, repeatedly use, and then terminate a session that cannot be resumed or reactivated. 
This can be represented by:
\[
	\psi = \nu X. \fresh{s}\diam{\mathsf{S}, s}.\mu Y.(\diam{\mathsf{U}, s}Y \lor \diam{\mathsf{T}, s}X)
\]
%\begin{align*}
%	\psi &= \fresh{s_0}\diam{\mathsf{S}, s_{0}}(\nu X(s_i).(\diam{\mathsf{U}, s_i}X(s_i) \lor \diam{\mathsf{T}, s_i}\fresh{s_{i+1}}\psi'))(s_0)
%	\\
%	\psi' &= \diam{\mathsf{S}, s_{i+1}}X(s_{i+1}) \land \lnot(\mu Y().\bigvee\nolimits_y \diam{y}(Y \lor s_{i} = y))
%\end{align*}
where $\mathsf{S,U,T}$ stand for start, use and terminate respectively.
% <<<<<<< HEAD
% We begin by selecting a fresh session $s_0$ that has not yet been used, and then starting it. 
% From there, we then enter the recursive variable $X$ with $s_0$ and then:
% \begin{itemize}
% 	\item use $s_0$, represented by accepting $s_0$ with an action tagged \emph{use}, and repeat $X$ with the same session,
% 	\item or terminate $s_0$ and select some new session $s_1$.  
% \end{itemize}
% In the latter case, we enter $\psi'$, where we start $s_1$ and repeat $X$ now with $s_1$, and we also ensure that the now terminated session $s_0$ is not used again in any capacity (as in the previous example).
% The use of the greatest fixpoint variable $X$ ensures that the property is never violated.
% \end{example}
% =======
We begin by entering the recursive $\nu$-variable $X$, where each time we enter $X$, we select a fresh session $s$ and perform the action $\mathsf{start}(s)$. 
From here, we enter the $\mu$-variable $Y$ where each time we enter, we repeatedly perform the action $\mathsf{use}(s)$ and eventually perform the action $\mathsf{terminate}(s)$. 
Once $s$ has been terminated, we repeat the process by re-entering the $\nu$-variable $X$, starting some new session $s'$. 
Note that the use of the $\mu$-variable $Y$ ensures that the session will eventually terminate, and the use of the $\nu$-variable $X$ ensures that the property is never violated.
\end{example}
%We begin by selecting a fresh session $s_0$ that has not yet been used, and then starting it. 
%From there, we then enter the recursive variable $X$ with $s_0$ and then:
%\begin{itemize}
%	\item use $s_0$, represented by accepting $s_0$ with an action tagged \emph{use}, and repeat $X$ with the same session,
%	\item or terminate $s_0$ and select some new session $s_1$.  
%\end{itemize}
%In the latter case, we enter $\psi'$, where we start $s_1$ and repeat $X$ now with $s_1$, and we also ensure that the now terminated session $s_0$ is not used again in any capacity (as in the previous example).
%The use of the greatest fixpoint variable $X$ ensures that the property is never violated.
%>>>>>>> refs/remotes/origin/main
%%% Local Variables:
%%% mode: latex
%%% TeX-master: "main"
%%% End:

%% file: modelchecking.tex
\newcommand\ndb{qweqwe}
\renewcommand\fgame{\Gcal}
\newcommand\wellbound[1]{\triangleleft_{#1}}

Given an FLTS $L$, a configuration $(s_0,H_0)$ in $L$ and a firm closed formula $\phi_0$,
the \emph{model checking} problem asks whether $(s_0,H_0)\in\sem[]{\phi_0}$. To solve this problem, in this section we introduce parity games that capture satisfiability of formulas in given configurations. To establish this connection it will be useful to consider formulas $\phi_0$ which may not be closed.
We start off with a general definition of parity games in nominal sets.

\begin{definition}\label{def:paritygames}
  A \boldemph{(nominal) parity game} is a tuple $\Gcal=\langle V,V_A,V_D,E,\Omega\rangle$ where $\langle V,E\rangle$ is a directed graph, with vertices called \emph{positions} and edges called \emph{moves}, $V=V_A\uplus V_D$ is a partitioning of positions into \emph{Attacker} and \emph{Defender} ones respectively, and $\Omega:V\to\{0,\dots,d\}$ is a \emph{ranking function} for some maximum rank $d\in\mathbb{N}$.\\
  In addition, $V$ is a finitely supported subset of some nominal set $\hat V$ and, for all $P_1,P_2\in V$ and $P_1',P_2'\in\hat V$, if $(P_1,P_1')\nomeq (P_2,P_2')$ then:
  \begin{itemize}
  \item if $P_1\in V_\chi$ then $P_2\in V_\chi$ (for $\chi\in\{A,D\}$), and $\Omega(P_1)=\Omega(P_2)$;
  \item if $(P_1,P_1')\in E$ then $(P_2,P_2')\in E$.
  \end{itemize}
\end{definition}

Given a game $\Gcal$, we may write $\pos(\Gcal)$ for its set of positions $V$.
Intuitively, the nominal conditions in \cref{def:paritygames} say that the game seen from two nominally equivalent positions is the same up to permutation.
Note that $\Gcal$ is not equivariant in general, and its support is the support of $V$.
%Given a game $\Gcal$ as above, its \emph{nominal closure} is the game $\Cl(\Gcal)=\langle \Cl(V),\Cl(V_A),\Cl(V_D),\Cl(E),\Cl(\Omega)\rangle$ and is equivariant.

Starting from a root position $P_0\in V$, a \emph{play} is a path within $\Gcal$, i.e.\ a finite or infinite sequence $\Pi=P_0,P_1,\dots$ such that $(P_i,P_{i+1})\in E$ for each $i$. A play is \emph{complete} if it ends in a leaf. Defender (Attacker) wins a complete play if the last position in it belongs to Attacker (resp.~Defender). Defender (Attacker) wins an infinite play if the highest position rank in it is even (resp.~odd). We say that  $P_0$ is a \emph{winning position} for Defender (or Defender \emph{wins from} $P_0$) if there is a partial function $\Theta: V_D\rightharpoonup E$ (called a \emph{strategy}) such that Defender wins  every complete or infinite play $P_0,P_1,\dots$ such that for each $P_i\in V_D$ we have $\Theta(P_i)=(P_i,P_{i+1})$.

The development of parity games for \fhml\ follows
parity games for the $\mu$-calculus, cf.~\cite{EmersonJutla91}. We start by
considering the negation-free fragment of our logic. In the sequel, we shall consider the derived connectives of \cref{rem:derived} as primitives (and exclude negation).

\begin{definition}
  A formula is called \emph{negation-free}
  if it contains no negations. For each formula $\phi$, we
write $!\phi$ for its negation-free variant. This can be constructed by induction using standard duality rules for pushing negation inside boolean, modal and recursion constructors, along with the rules: \
$ !(\lnot\lnot\phi)={!\phi},\
        !(\lnot \fresh{x} \phi) = \fresh{x} !(\lnot \phi) 
$.
\end{definition}

\begin{restatable}{lemma}{negfreesize}\label{lem:negfreesize}
    For any formula $\phi$, $|!(\phi)|$ is bounded by $|\phi| - n$, where $n$ is the number of negations in $\phi$. 
\end{restatable}

For the remainder of this section, and unless specified otherwise, any \fhml\ formulas used are going to be assumed to be firm and negation-free.

% In our definitions it will be useful to use substitutions of recursion formulas for recursion variables. As the latter have arities, such substitutions are parametric on the arguments variables are applied on. To denote substitutions of this form, we use the following notation:
% \[
%     \{\sigma X(\vec x).\psi / X \} = \{(\sigma X(\vec x).\psi)(\vec a) / X(\vec a) \mid \vec a \in \Atoms^{\ar(X)}\}
% \]
% with $\sigma\in\{\mu,\nu\}$.
We shall make use of the (recursion) \emph{alternation depth} of formulas. At this point it is important that the formulas we examine have a unique specification for each of their bound variables. More specifically, for any formula $\phi$ we stipulate that:
\begin{itemize}
\item the bound and free recursion variables of $\phi$ be disjoint;
\item if $\sigma_1 X(\vec x_1).\phi_1,\sigma_2 X(\vec x_2).\phi_2$ are subformulas of $\phi$ then $(\sigma_1,\vec x_1,\phi_1)=(\sigma_2,\vec x_2,\phi_2)$.
\end{itemize}
Note that these conditions are not unusual (cf.~\cite{KlinL17}). They are not restricting our logic either as any formula has an $\alpha$-equivalent one in the form above. At the same time, they allow us to match bound recursion variables with their binding subformulas.

For value variables, on the other hand, we simply disallow a binder on a variable $x$ to be within the scope of another binder on $x$.

\begin{definition}[\cite{BradfieldW18, EmersonJS01}]\label{def:altDepthNom}
    Let $\phi$ be a \fhml\ formula. 
    We define the \emph{dependency order} on bound variables of $\phi$ as the smallest partial order $\leq_\phi$ such that\
    $X \leq_\phi Y$\
    if $X(\vec u)$ occurs free in some $\phi'$ such that $\sigma Y(\vec x) .\phi'$ appears in $\phi$ (for some $\vec u$). 
    \\
    The \boldemph{alternation depth} of a 
    bound variable $X$ in $\phi$ is defined as the maximal length of a chain $X=X_1 \leq_\phi \dots \leq_\phi X_n$ such that:
    \begin{itemize}
        \item if $X$ is a $\mu$-variable, each $X_i$ is a $\mu$-variable if $i$ is odd, and a $\nu$-variable otherwise; 
    \item if $X$ is a $\nu$-variable,  each $X_i$ is a $\nu$-variables if $i$ is odd, and a $\mu$-variable otherwise. 
    \end{itemize}
    The {alternation depth} of a formula $\phi$ (denoted \emph{adepth($\phi$)}) is the maximum alternation depth of variables bound in $\phi$, or zero if there are no fixpoints.
\end{definition}

Parity games for \fhml\ have positions of the form $(s, H, \xi, \psi)$, with:
\begin{itemize}
    \item $(s, H) \in \fstates$ being the current configuration that is examined;
    \item $\xi$ being a $\Ucal$-variable assignment;
    \item $\psi$ being the current formula. 
\end{itemize}
We define games with respect to a \emph{root position} corresponding to the model-checking problem we aim to decide. The positions and moves are given inductively using \emph{game rules}.

\begin{definition}\label{def:parityG}
    Let $\phi_0$ be a firm negation-free \fhml\
    formula,
    $\xi$ a $\Ucal$-variable assignment containing the free recursion variables of $\phi_0$, $\Lcal$ an FLTS and $(s_0,H_0)$ a configuration in $\Lcal$.
    The \boldemph{game} $\fgame(\Lcal, \phi_0, \xi, s_0, H_0)$ has root position $(s_0, H_0, \xi, \phi_0)$ and game rules:
    \begin{itemize}
        \item $(s,H,\xi, a = a)$ belongs to Attacker, whereas
 $(s,H,\xi, a = b)$ to Defender;
        \item $(s,H, a \neq a)$ belongs to Defender, whereas 
 $(s,H, a \neq b)$  to Attacker;
        \item $(s,H,\xi, X(\vec a))$ belongs to Attacker if $(s, H)\in \xi(X)(\vec a)$, and to Defender otherwise;
        \item $(s,H,\xi, \modal{\ell}\phi) \rightarrow (s',H',\xi, \phi)$, if $(s, H) \xrightarrow{\ell} (s', H')$, belongs to Defender if $\modal{\ell}=\diam{\ell}$, and to Attacker if $\modal{\ell}=[\ell]$;
        \item $(s,H,\xi, \phi_1 \odot \phi_2) \rightarrow (s,H,\xi,\phi_i)$, $i \in \{1,2\}$, belongs to Defender if $\odot=\lor$, and to Attacker if $\odot=\land$;
        % \item $(s,H,\xi, \phi_1 \wedge \phi_2) \rightarrow (s,H,\xi, \phi_i)$, $i \in \{1,2\}$, played by Attacker,
        \item $(s,H,\xi, \bigodot\nolimits_x \phi) \rightarrow (s,H,\xi,\phi\{a/x\})$, $a \in \Atoms$, belongs to Defender if $\bigodot=\bigvee$, and to Attacker if $\bigodot=\bigwedge$;
        % \item $(s,H,\xi, \bigwedge\nolimits_x \phi) \rightarrow (s,H,\xi,\phi\{a/x\})$, $a \in \Atoms$, played by Attacker,
        % \item $(s,H,\xi, [\ell]\phi) \rightarrow (s',H',\xi, \phi)$, $\forall (s', H'). (s, H) \xrightarrow{\ell} (s', H')$, played by Attacker,
        \item $(s,H,\xi,\fresh{x}\phi) \rightarrow (s, H, \xi, \phi\{a/x\}),\ a \# \phi,\xi,H$, belongs to either;
        \item $(s,H,\xi, (\sigma X(\vec x).\phi)(\vec a)) \rightarrow (s, H, \xi,(\phi\{\sigma X(\vec x).\phi / X\})\{\vec a / \vec x\})$, $\sigma\in\{\mu,\nu\}$, belongs to either.
        % \item $(s,H,\xi, (\nu X(\vec x).\phi')(\vec a)) \rightarrow (s, H, \xi, (\phi'\{\nu X(\vec x).\phi' / X\})\{\vec a / \vec x\})$, played by either.
    \end{itemize}
    % In the special cases where $\phi$ is a closed formula (i.e., with no free recursive variables $X$), we suppress $\xi$ everywhere as it is not used and simply write $\fgame(\Lcal, \phi, s, H)$,
  % For winning conditions we stipulate that, in a finite play, a player wins if the other player is required but unable to make a move.
  % In an infinite play, the player who owns the highest infinitely recurring rank wins, where Defender owns even-ranked positions, and Attacker owns odd-ranked ones.
      The rank of a position with formula $\phi$ is given by $\Omega(\phi)$:
    \[        \Omega(\phi) =
      \begin{cases} 2 \cdot \lfloor adepth(X)/2 \rfloor & \text{if $\phi$ is of the form }(\nu X(\vec x).\phi')(\vec a)\\
        2 \cdot \lfloor adepth(X)/2 \rfloor + 1 & \text{if $\phi$ is of the form }(\mu X(\vec x).\phi')(\vec a)\\
         0 &\text{otherwise}
      \end{cases}
    \]
In the case when $\phi_0$ is closed, we denote the game simply by $\fgame(\Lcal, \phi_0, s_0, H_0)$.    
\end{definition}

Note that the formulas appearing inside positions are firm and negation-free (as is $\phi_0$ by assumption). 

%
%  \begin{example}
%One or two examples here
%  \end{example}

Later in this section it will be important to count the orbits of a history-bounded version of the parity game of \cref{def:parityG}. To that effect, it is useful to the formulas appearing in the game, in terms of subformulas of $\phi_0$ and substitutions. This is achieved by the following notion of formula closure. 
  
\begin{definition}\label{def:closFresh}
    Given a \fhml\ formula $\phi_0$, we define its closure $\clos(\phi_0)$ to be the smallest set $\chi$ of triples $(\phi,\gamma,\theta)$ such that $(\phi_0,\varepsilon,\varepsilon) \in \chi$ and:
    \begin{itemize}
        \item if $(\phi_1 \odot \phi_2, \gamma,\theta) \in \chi$, with $\odot\in\{\lor,\land\}$, then $(\phi_1, \gamma,\theta), (\phi_2, \gamma,\theta) \in \chi$;
        \item if $(\modal{\ell}\phi, \gamma,\theta) \in \chi$, with $\modal{l}\in\{\diam{\ell},\sq{\ell}\}$, then $(\phi, \gamma,\theta) \in \chi$;
        \item if $(\bigodot_x\phi, \gamma,\theta) \in \chi$, with $\bigodot\in\{\bigvee,\bigwedge\}$, %or $(\fresh{x} \phi', \theta) \in \chi$
          then $(\phi,(\gamma,\{a/x\}),\theta) \in \chi$, for all $a\in\Atoms$;
        \item if $(\fresh{x}\phi,\gamma, \theta) \in \chi$
          then $(\phi,(\gamma,\{a/x\}), \theta) \in \chi$, for all $a\#\phi\{\theta\}\{\gamma\}$;
        \item if $((\sigma X(\vec x).\phi)(\vec u),\gamma, \theta) \in \chi$, with $\sigma\in\{\mu,\nu\}$, then $(\phi,(\gamma,\{\vec c/\vec x\}), (\theta,\{\sigma X(\vec x).\phi / X\}))\in\chi$, for all $\vec c \in \Atoms^{\ar(X)}$.
        % \item if $((\nu X(\vec x).\phi')(\vec a), \theta) \in \chi$ then $\{(\phi'\{\vec c/\vec x\}, \theta[(\nu X(\vec x).\phi')\{\theta\} / X]) \mid \vec c \in \Atoms^{\ar(X)}\} \subseteq \chi$. 
    \end{itemize}
    We write $\clos^*(\phi_0) = \{\phi\{\theta\}\{\gamma\}\mid (\phi,\gamma, \theta) \in \clos(\phi_0)\}$.
\end{definition}

The following lemma shows that the formulas appearing in the game $\fgame(\Lcal,\phi_0, s_0,H_0)$ can be traced back to the closure. It moreover provides some bounds on their orbits and supports. 

\begin{restatable}{lemma}{freshClosProps}\label{lem:freshClosProps}
Given a closed formula $\phi_0$, an FLTS $\Lcal$ and configuration $(s_0,H_0)$:
    \begin{enumerate}
        % \item\label{lem:freshClosProps:1} for all $(\phi,\gamma,\theta)\in\clos(\phi_0)$, the free variables of $\phi$ are included in $\dom(\theta)\uplus \cal Z$ and all formulas in $\rng(\theta)$ are fixpoints and have free variables in $\dom(\theta)\uplus\Zcal$;
        \item\label{lem:freshClosProps:2} for all positions $(s,H,\phi)$ in $\fgame(\Lcal,\phi_0, s_0,H_0)$ we have that $\phi\in\clos^*(\phi_0)$; 
        \item\label{lem:freshClosProps:3} $|\orb(\clos^*(\phi_0))|$ is bounded by $|\phi_0|\cdot \frac{(|\supp(\phi_0)| + \|\phi_0\|)!}{|\supp(\phi_0)|!}$;
\item\label{lem:freshClosProps:4} for all $(\phi,\gamma,\theta)\in\clos(\phi)$, $|\supp(\phi,\gamma,\theta)| \le |\supp(\phi_0)| + \|\phi_0\|$. 
%\ntnote{this needs to be proven still}
    \end{enumerate}
\end{restatable}

We proceed to showing the correspondence between the semantics of an \fhml\ formula and its corresponding parity game.
\begin{restatable}{theorem}{FnomSatPar}\label{thm:FnomSatPar} %THEOREM 6 - BRADFIELDW18
    Given a \fhml\ formula $\phi_0$,
    $\Ucal$-variable assignment $\xi$,
    FLTS $\Lcal$ and $(s_0, H_0)$ from $\Lcal$, we have $(s_0, H_0) \in \sem{\phi_0}$ iff Defender wins from $(s_0, H_0,\xi, \phi_0)$ in $\fgame(\Lcal, \phi_0, \xi, s_0,H_0)$.
\end{restatable}

We have thus reduced model checking of formulas to winning in parity games. 
At this stage, we can restrict our attention to closed (firm negation-free) formulas $\phi_0$.
The next step is deciding
winning regions of our parity games. The latter requires the games to be finite, whereas ours are not so. We therefore reduce our parity games to equivalent finite ones.
For this to be possible, we restrict our attention to FLTSs that are fresh-orbit-finite.

For brevity, in the remainder of this section we shall say that $(\Lcal,\phi_0,s_0,H_0)$ is a \boldemph{setup of grade $N$} if:
\begin{itemize}
\item $\Lcal$ is a fresh-orbit-finite FLTS with set of states $\Scal$, and $(s_0,H_0)\in\hat \Scal$;
\item $\phi_0$ is a closed (firm negation-free) \fhml\ formula;
\item $N = |\supp(\phi_0)| + \|\phi_0\| + \regindex(\Scal)$. 
\end{itemize}
Given a setup $(\Lcal,\phi_0,s_0,H_0)$ of grade $N$, we first present a version of $\fgame(\Lcal,\phi_0,s_0,H_0)$ with bounded histories.\ntnote{$N$ is used for sets of names earlier, can be confusing.} 
The idea is to restrict the size of the history to the maximum number of names that are available to the formula and the states of the FLTS, and one extra to represent a fresh name. This amounts to a bound of $N+1$ names. 
When a new name is added and the size of the history is greater than $N$, then names that are in the history but no longer present in the formula are forgotten. 
Each position in this game is a tuple $(s, H, \phi)$ with $|H|\leq N+1$.
% \begin{itemize}
%     \item $(s, H) \in \fstates$ being the current configuration being examined and $\phi$ a \fhml\ formula;
%     \item $|H| \le \S + 1$ and $\supp(s) \subseteq H$.
% \end{itemize}
\begin{definition}
Let $(\Lcal,\phi_0,s_0,H_0)$ be a setup of grade $N$. 
The \boldemph{history-bounded game} $\Gcal_\S(\Lcal,\phi_0, s_0, H_0)$ has root
 position $(s_0, H_0', \phi_0)$, for some $H_0'\wellbound{N}(s_0,\phi_0,H_0)$, and game rules as in \cref{def:parityG}, apart from the rule for modal connectives:
    \begin{itemize}
 %        \item $(s,H, a = a)$ belongs to Attacker, whereas 
 % $(s,H, a = b)$ belongs to Defender;
 %        \item $(s,H, a \neq a)$ belongs to Defender, whereas 
 % $(s,H, a \neq b)$ belongs to Attacker;
        \item $(s,H, \modal{\ell}\phi) \rightarrow (s',H', \phi)$, if $\ell=(t, \vec a)$
        and $(s,H) \xrightarrow{t, \vec a} (s',H\cup\{\vec a\})$ with $H'\wellbound{N}(s',\phi,H\cup\{\vec a\})$,
        played by Defender if $\modal{\ell}=\diam{\ell}$, and by Attacker if $\modal{\ell}=\sq{\ell}$;
        \end{itemize}
% \item $(s,H, \phi_1 \odot \phi_2) \rightarrow (s,H,\phi_i)$, $i \in \{1,2\}$, belongs to Defender if $\odot=\lor$, and by Attacker if $\odot=\land$;
%         \item $(s,H, \bigodot\nolimits_x \phi) \rightarrow (s,H,\phi\{a/x\})$, $a \in \Atoms$, played by Defender if $\bigodot=\bigvee$, and by Attacker if $\bigodot=\bigwedge$;
%         \item $(s,H,\fresh{x}\phi) \rightarrow (s, H, \phi\{a/x\}),\ a \#\phi,\xi,H$, played by either;
%         \item $(s,H, (\sigma X(\vec x).\phi)(\vec a)) \rightarrow (s, H,\phi\{\sigma X(\vec x).\phi / X\}\{\vec a / \vec x\})$, $\sigma\in\{\mu,\nu\}$, played by either.
while the ranking function is the same as in \cref{def:parityG}.
The \emph{well-bounding} relation $\triangleleft$ is given by $H'\wellbound{N}(s,\phi,H)$ if $H'\subseteq H$ is such that:
\begin{itemize}
            \item if $|H| \le N$ then $H' = H$,
            \item else $H \cap \supp(s,\phi) \subseteq H' \subseteq H$ and $|H'| = N + 1$.
        \end{itemize}
\end{definition}

We note that, though the positions in the unbounded game have their histories trimmed at modal operators, they remain FLTS configurations, i.e.\ for each $(s,H,\phi)$ we have $\supp(s)\subseteq H$. 
We next show that we can decide winning a parity game $\Gcal$ by examining its equivalent history-bounded game $\Gcal_N$.
To do this, we are going to build a bisimulation-like relation $\Rcal$ between $\pos(\Gcal)$ and $\pos(\Gcal_N)$. Since the positions in the two games can have name-mismatches, but essentially do ``the same'' transitions, it is useful to introduce a notion of {bisimilation} up to renaming. Note use of relation composition: $R_1;R_2=\{(x,z)\mid \exists y.\,(x,y)\in R_1\land (y,z)\in R_2\}$.

\begin{definition}
Given parity games $\Gcal_1, \Gcal_2$, we call a relation $R \subseteq \pos(\Gcal_1) \times \pos(\Gcal_2)$ a \emph{bisimulation} if for all $(P_1, P_2) \in R$:
    \begin{itemize}
        \item $P_1$ has the same rank and belongs to the same player as $P_2$;
        \item for each $P_1 \rightarrow P_1'$, there exist some $P_2 \rightarrow P_2'$  such that $(P_1',P_2')\in R$;
        \item for each $P_2 \rightarrow P_2'$, there exist some $P_1 \rightarrow P_1'$ such that $(P_1', P_2') \in R$.
    \end{itemize}
    We say that $P_1$ and $P_2$ are \emph{bisimilar} (denoted ${P_1}\bisim{P_2}$)
    if there is a bisimulation $R$ such that $(P_1, P_2) \in R$. \\
    A \emph{bisimulation up to nominal equivalence (nom-bisimulation)} is defined to be a relation $R$ with the same properties as above, apart from the requirement in the last two bullet points being instead that $(P_1',P_2')\in {\nomeq};R;{\nomeq}$.
    We say that $P_1$ and $P_2$ are \emph{nom-bisimilar} (denoted ${P_1}\nombisim{P_2}$)
    if there is a nom-bisimulation $R$ such that $(P_1, P_2) \in R$. 
\end{definition}

In the following two lemmas, we let $P_1,P_2$ be positions in games $\Gcal_1,\Gcal_2$ respectively.

\begin{lemma}\label{lem:nombisimbisim}
If $P_1\nombisim P_2$ then $P_1\bisim P_2$. %(seen as positions in $\Cl(\Gcal_1),\Cl(\Gcal_2)$ respectively). 
\end{lemma}
\begin{proof}
  Let $R\subseteq\pos(\Gcal_1)\times\pos(\Gcal_2)$ be 
  a nom-bisimulation. We show that $R'={\nomeq;R;\nomeq}\subseteq\pos(\Gcal_1)\times\pos(\Gcal_2)$ is a bisimulation.
Suppose $(P_1,P_2)\in R'$, i.e.\ $\pi_i\cdot P_i=Q_i$ ($i=1,2$) and $(Q_1,Q_2)\in R$, and let $P_1\to P_1'$. We have $P_1,Q_1\in\pos(\Gcal_1)$ and, taking $Q_1'=\pi_1\cdot P_1'$, $(P_1,P_1')\nomeq(Q_1,Q_1')$, hence $Q_1\to Q_1'$ (and $Q_1'\in\pos(\Gcal_1)$). Then, by nom-bisimulation, $Q_2\to Q_2'$ and $(Q_1',Q_2')\in R'$. As before, setting $P_2'=\pi_2^{-1}\cdot Q_2'$, we obtain $P_2\to P_2'$ (and $P_2'\in\pos(\Gcal_2)$). We observe that $(P_1',P_2')\in{\nomeq;R';\nomeq}=R'$.
% Note that nominal closure may take us outside the positions of $\Gcal_i$ (for $i=1,2$), hence the need to consider $\Cl(\Gcal_i)$.
\end{proof}

\begin{restatable}{lemma}{nombisimequiv}\label{lem:nombisimequiv}
If $P_1\bisim P_2$ then Defender wins from $P_1$ iff they win from $P_2$.
\end{restatable}

We can relate games with their history-bounded counterparts by nom-bisimulations.

\begin{definition}\label{def:theFreshRelation}
Given a setup $(\Lcal,\phi_0,s_0,H_0)$ of grade $N$, consider $\fgame(\Lcal, \phi_0,s_0, H_0)$  and $\Gcal_\S(\Lcal, \phi_0, s_0, H_0)$.
We define the \emph{history-bounding relation} $\Rcal_{(\Lcal, \phi_0, s_0, H_0)} \subseteq \pos(\fgame) \times \pos(\Gcal_\S)$, by setting $((s_1, H_1, \phi_1), (s_2, H_2, \phi_2)) \in \Rcal$ if $s_1 = s_2$, $\phi_1 = \phi_2$ and $H_2\wellbound{N}(s_1,\phi_1,H_1)$.
    % \begin{itemize}
    %     \item if $|H_1| \le \S$ then $H_1 = H_2$,
    %     \item if $|H_1| > \S$ then $\supp(\phi_1) \cap H_1 \subseteq H_2 \subseteq H_1$ and $|H_2| = \S + 1$.
    % \end{itemize}
\end{definition}

Thus, a position $(s,H,\phi)$ in the unbounded game is related to itself if its history $H$ is small, i.e.\ it has no more than $N$ names. If $|H|$ is larger, then we relate $(s,H,\phi)$ to each $(s,H',\phi)$ where $H'\subseteq H$ has $N+1$ names and in particular contains $\supp(s)$ (as $(s,H')$ is a configuration) and all names from $\supp(\phi)$ that appear in $H$.

\begin{restatable}{lemma}{boundingisnombisim}\label{lem:boundingisnombisim}
    % Let $\hat\phi$ be a \fhml\ formula, $\Lcal$ be an F-LTS, $(\hat s,\hat H)$ a state in $\Lcal$ and $\S = |\supp(\hat \phi)| + \|\hat \phi\| + \regindex(\fstates)$. Pick $\hat H'$ such that 
    % \begin{itemize}
    %   \item if $|\hat H| \le \S$ then $\hat H' = \hat H$,
    %     \item if $|\hat H| > \S$ then $\supp(\hat\phi, \hat s) \cap \hat H \subseteq \hat H' \subseteq \hat H$ and $|\hat H'| = \S + 1$.
    %         \end{itemize}
    %     Let $\fgame = \fgame(\Lcal, \hat \phi, \hat s, \hat H)$ and $\Gcal_\S = \Gcal_\S(\Lcal, \hat \phi, \hat s, \hat H')$, and let $\Rcal_{(\Lcal, \hat \phi, \hat s, \hat H)} \subseteq \pos(\fgame) \times \pos(\Gcal_\S)$ be a history bounded relation.
    % Then, $((\hat s, \hat H, \hat \phi), (\hat s, \hat H', \hat \phi)) \in \Rcal_{(\Lcal, \hat \phi, \hat s, \hat H)}$ and 
    % $\Rcal_{(\Lcal, \hat \phi, \hat s, \hat H)}$ is a bisimulation up to nominal equivalence.
$\Rcal_{(\Lcal, \phi_0, s_0, H_0)}$ is a nom-bisimulation.
\end{restatable}

Combining \cref{thm:FnomSatPar} with \cref{lem:boundingisnombisim,lem:nombisimbisim,lem:nombisimequiv}, we can reduce \fhml\ model-checking to deciding winning a history-bounded parity game. However,
although the size of histories, and therefore of positions, is now bounded, the number of positions can still be infinite due to the names appearing in them. We can resolve this by grouping together positions that are ``equivalent'' up to permuting their names. Formally, this is done by considering games in which positions are orbits of ordinary positions.

% If we have a parity game $\fgame$, with a corresponding history bounded game $\Gcal_\S$, 
% then for the history bounded relation $\Rcal_{(\Lcal, \phi, s, H)} \subseteq \pos(\fgame) \times \pos(\Gcal_\S)$, we show that $\Rcal' = \{(\orbs(P_1), \orbs(P_2)) \mid (P_1, P_2) \in \Rcal_{(\Lcal, \phi, s, H)}\}$ is a bisimulation.

\begin{definition} 
    Given a parity game $\Gcal=\langle V,V_A,V_D,E,\Omega\rangle$, 
    we construct the \boldemph{orbits game} $\gorb(\Gcal)=\langle V',V_A',V_D',E',\Omega'\rangle$ as follows:
    \begin{itemize}
    \item $V' = \{\orbs(P) \mid P \in V\}$ and $V_\chi'=\{\orbs(P) \mid P \in V_\chi\}$ (for $\chi\in\{A,D\}$);
      \item $E'=\{(\orbs(P),\orbs(P'))\mid (P,P')\in E\}$;
        \item  for each $P\in V$, $\Omega'(\orbs(P)) = \Omega(P)$. 
    \end{itemize}
\end{definition}

We can see that $\gorb(\Gcal)$ is well defined. In particular, $\Omega'$ is well defined due to the fact that $\Omega(P_1)=\Omega(P_2)$ whenever $P_1\nomeq P_2$.  
Moreover, we can show that the relation\ $R
=        \{(P, \orbs(P)) \mid P \in \pos(\Gcal)\}$\ is a bisimulation.

\begin{lemma}\label{lem:parityBisim}
  Given $\Gcal$ as above and $P\in\pos(\Gcal)$, $P\bisim\orbs(P)$.
\end{lemma}

Taking stock of the game transformations defined above and the bisimulation relations between them, we can decide winning in a setup of finite grade.
%For the remainder of this section, we fix a firm closed \fhml\ formula $\phi_0$, a fresh-orbit-finite FLTS $\Lcal$ and a configuration $(s_0,H_0)$ of $\Lcal$.

\begin{corollary}\label{cor:winningNom}
Given a setup $(\Lcal,\phi_0,s_0,H_0)$ of grade $N$, Defender wins in $\Gcal(\Lcal,\phi_0,s_0,H_0)$ (from $(s_0,H_0,\phi_0)$) iff they win in $\gorb(\Gcal_N(\Lcal,\phi_0,s_0,H_0))$ from $\orbs(s_0,H_0',\phi_0)$. Hence, it is decidable whether $(s_0,H_0)\in\sem[]{\phi_0}$.
\end{corollary}
\begin{proof}
The first part follows from \cref{thm:FnomSatPar,lem:boundingisnombisim,lem:nombisimbisim,lem:nombisimequiv,lem:parityBisim}. The second part follows from \cref{thm:FnomSatPar} and the fact that 
$\gorb(\Gcal_N(\Lcal,\phi_0,s_0,H_0))$ is a finite game.
\end{proof}

We next calculate upper bounds for the size of $\gorb(\Gcal_N(\Lcal,\phi_0,s_0,H_0))$ and the time needed to construct it. This allows us to give a complexity bound for model checking.
Recall that we set $N=|\supp(\phi_0)|+\|\phi_0\|+\regindex(\Scal)$. In the calculations below we prefer to distinguish between complexity due to $\phi_0$ and $\Lcal$, so instead of $N$ we shall use the measure: $M(\phi_0)=|\supp(\phi_0)|+\|\phi_0\|$. This can be seen as the \emph{nominal potential} of $\phi_0$.%\ntnote{bring in the definition of $\Phi$}

In constructing the orbits game, we need to be able to check if two positions of $\Gcal_N$ are nominally equivalent. The way this is performed is described in \cref{app:orbitGameSize}: given $(s_1,H_1,\phi_1),(s_2,H_2,\phi_2)$, we try to build a permutation $\pi$ such that $\pi\cdot\phi_1=\phi_2$ and, if successful, we extend $\pi$ to some $\pi'$ so that, in addition, $\pi'\cdot s_1=s_2$. The latter part requires some oracle that is able to answer such questions for the given nominal set of states $\Scal$.
We let a \emph{permutation oracle} $\Psi_\Lcal$ to be a function that,
given  $s, s' \in \Scal$
	 and  $\vec a, \vec b \in \Atoms^k$ (for some $k$, and each containing distinct names), the oracle returns 
	a permutation $\pi'$ such that $\pi' \cdot s =s'$ and $\pi'\cdot \vec a = \vec b$,
	or $\NO$ if no such permutation exists. 
	We stipulate that the time complexity of $\Psi_\Lcal(\vec a, \vec b, s, s')$ is bounded by some $\Phi_\Lcal(k)$.
%\end{definition}

% \begin{restatable}{lemma}{sizebound}\label{lem:sizebound}
% 	For every \fhml\ formula $\phi$ such that for every subformula $(\mu X(\vec x).\phi')(\vec u)$ of $\phi$ we have $X$ is free in $\phi'$, then $|\supp(\phi)|+2\|\phi\|+\zeta(\phi)\leq|\phi|$ where $\zeta(\phi)$ is the number of free value variables (of the form $x, y, \dots$) in $\phi$.
% \end{restatable}

\begin{restatable}{lemma}{orbitGameSize}\label{lem:orbitGameSize}\label{lem:fnomparsize}
  Given a setup $(\Lcal,\phi_0,s_0,H_0)$ of grade $N$ and $\Gcal_N=\Gcal_N(\Lcal, \phi_0, s_0, H_0)$:
  \begin{enumerate}
    \item the size of $\pos(\gorb(\Gcal_\S))$ is bounded by
    $
        |\orb(\Scal)| \cdot|\phi_0|\cdot \frac{M({\phi_0})!}{|\supp(\phi_0)|!}\cdot (M(\phi_0)+\regindex(\Scal) + 1)^{M(\phi_0) + 1}\cdot(1+o(1))
    $.
    \item the time complexity of constructing $\gorb(\Gcal_N)$ is in
$ O(   (|\phi_0| + \Phi_\Lcal(M(\phi_0)) +\regindex(\Scal))
    \cdot
    \max(M(\phi_0)+\regindex(\Scal) + 2, |\Scal|
    )  \cdot 
    (|\orb(\Scal)| \cdot|\phi_0|\cdot \frac{M(\phi_0)!}{|\supp(\phi_0)|!}\cdot (M(\phi_0)+\regindex(\Scal) + 1)^{M(\phi_0) + 1})
    ^2)$.
\end{enumerate}
\end{restatable}
\begin{proof}  We show~1, while~2 is proven in \cref{app:complexityboundedgame}.
    By definition and \cref{lem:freshClosProps}(1), we know that $\pos(\Gcal_N) \subseteq \fstates \times \clos^*(\phi_0)$.
        Then, by \cref{lem:nomlemma}(1,2) we have that $\pos(\gorb(\Gcal_\S)) \subseteq \orb(\cls{}(\fstates \times \clos^*(\phi_0))) = \orb(\fstates \times \cls{}(\clos^*(\phi_0)))$.
    For every $(s, H, \phi) \in \Gcal_\S$, we have $\supp(s) \subseteq H$ and $|H|\leq N+1$, and so we have $|\orb(\fstates)| \le |\orb(\Scal)| \cdot (\S + 1)$ (what matters of $H$ is the size of $|H\setminus\supp(\phi)|$).
    Using \cref{lem:nomlemma}(3) and the fact that the largest support sizes in $\fstates$ and $\clos^*(\phi_0)$ are bounded by $N+1$ and $|\supp(\phi_0)| + \|\phi_0\|$ (by \cref{lem:freshClosProps}(3)) respectively:
    \begin{align*}
        |\orb(\cls{}(\fstates \times \clos^*(\phi_0))| &\le |\orb(\Scal)| \cdot (\S + 1) \cdot|\orb(\cls{}(\clos^*(\phi_0)))|\cdot (\S + 1)^{|\supp(\phi_0)| + \|\phi_0\|}\cdot(1 + \epsilon) 
        \\ &\le |\orb(\Scal)| \cdot|\orb(\clos^*(\phi_0))|\cdot (\S + 1)^{|\supp(\phi_0)| + \|\phi_0\| + 1}\cdot(1+\epsilon)
\\
        &\le |\orb(\Scal)| \cdot|\phi_0|\cdot \frac{(|\supp(\phi_0)|+\|\phi_0\|)!}{|\supp(\phi_0)|!}\cdot (\S + 1)^{|\supp(\phi_0)| + \|\phi_0\| + 1}\cdot(1+\epsilon)
        % \\
%        &\le |\orb(\Scal)| \cdot|\phi_0|\cdot \frac{(|\phi_0|)!}{|\supp(\phi_0)|!}\cdot (\S + 1)^{|\phi_0| + 1}\cdot(1+\epsilon)
    \end{align*}
    with the last two lines using \cref{lem:nomlemma}(1, second part) and \cref{lem:freshClosProps}(2).
%     and \cref{lem:sizebound}.
\end{proof}

% \begin{restatable}{lemma}{FRAGameSize}\label{app:FRAGameSize}
% 	Let $\Acal$ be an $r$-FRA with fresh-orbit-finite F-LTS $\Lcal$.
% 	Given a setup $(\Lcal,\phi_0,s_0,H_0)$ of grade $N$ and $\Gcal=\Gcal_N(\Lcal, \phi_0, s_0, H_0)$,
% 	the time complexity of constructing $\gorb(\Gcal)$ is in
% 	\begin{align*}
% 		&O(
% 		(2r + |\phi_0| + |\supp(\phi_0)| + \|\phi_0\|)
% 		\cdot
% 		\max(|\supp(\phi_0)| + \|\phi_0\| + r + 2, |Q|
% 		)
% 		\\
% 		&\ \cdot 
% 		(|Q| \cdot|\phi_0|\cdot \frac{(|\supp(\phi_0)|+\|\phi_0\|)!}{|\supp(\phi_0)|!}\cdot (\S + 1)^{|\supp(\phi_0)| + \|\phi_0\| + 1}\cdot(1+o(1)))
% 		^2)
% 	\end{align*}
% \end{restatable}

% \begin{restatable}{lemma}{FRAGameSize}\label{lem:FRAGameSize}
% 	Let $\Acal$ be an $r$-FRA with fresh-orbit-finite F-LTS $\Lcal$.
% 	Given a setup $(\Lcal,\phi_0,(q_0, \rho_0),H_0)$ of grade $N$ and $\Gcal_N=\Gcal_N(\Lcal, \phi_0, (q_0, \rho_0), H_0)$
% 	 the time complexity of constructing $\gorb(\Gcal_N)$ is in 
% 	 $
% 	 O(
% 	 (2r + 2|\phi_0|)
% 	 \cdot
% 	 \max(|\phi_0| + r + 2, |\fstates|
% 	 )
% 	 $
% \end{restatable}
\begin{restatable}{theorem}{MCinFHML}\label{thm:MCinFHML}
  Let $\Lcal$ be a fresh-orbit-finite FLTS with $n$ orbits and set of states $\Scal$, with register index $r$,
% and $(s_0,H_0)\in\hat\Scal$, 
  and let $\phi_0$ be a \fhml\ closed formula with maximum alternation depth $d$. 
    Then, the time complexity of model checking $\phi_0$ on some configuration in $\Lcal$ is in
	\begin{align*}
		&O(\ (n \cdot|\phi_0|\cdot \frac{M(\phi_0)!}{|\supp(\phi_0)|!}\cdot (M(\phi_0)+r + 1)^{M(\phi_0) + 1})^{\log(d) + 6}
		\\ &\ + 
		(|\phi_0| + \Phi_\Lcal(M(\phi_0)) +r)
		\cdot
		\max(M(\phi_0)+r + 2, n
		)\\&\quad \cdot 
		(n \cdot|\phi_0|\cdot \frac{M(\phi_0)!}{|\supp(\phi_0)|!}\cdot (M(\phi_0)+r + 1)^{M(\phi_0) + 1})
		^2\
		)
	\end{align*}
\end{restatable}

\begin{remark}
In the case of model checking formulas on an $r$-FRA $\Acal=\langle Q,\mu,\Sigma,\delta\rangle$, the calculations above are adapted by setting:
\begin{itemize}
\item $|\orb(\Scal)|=|Q|=n$, as two state-assignment pairs are equal up to permutation iff they share the same state;
\item $\regindex(\Lcal)\leq r$, which is equality if there is a state using all $r$  registers;
  \item we can take $\Phi_\Lcal(k)=k+r$, as it suffices to find a matching between the two register assignments (of $s_1$ and $s_2$) and check it complies with mapping $\vec a$ to $\vec b$.
  \end{itemize}
\end{remark}

%% file: app-logic.tex
%\subsection{Proof of \cref{lem:completelattice}}
% \begin{lemma}[Restatement of Lemma~\ref{lem:completelattice}]
% 	$\Ucal_n$ is a complete lattice and,
% 	for any $\phi, \xi$,
% 	the function $\lambda f.\lambda \vec b.\sem[{\xi[X\mapsto f]}]{\phi\{\vec b / \vec x\}}$ is monotone.
% \end{lemma}
\nomlemma*

\begin{proof}
	\begin{enumerate}
	\item For the former, we have that 
	$$\orb(\cls{}(X)) = \{\orbs(\pi\cdot x) \mid x \in X, \pi \in \Perm\} = \{\orbs(x) \mid x \in X\}$$
	For the latter, we define the function $f: \orb(\cls{}(X)) \to \orb[\supp(X)](X)$ as follows:
	\[
	f(O) = \orbs[\supp(X)](y) \text{ for some } y \in O \cap X
	\]
	Let us choose $O_1, O_2 \in \orb(\cls{}(X))$ such that $O_1 \neq O_2$ and $\hat x_1 \in O_1 \cap X, \hat x_2 \in O_2 \cap X$.
	Suppose $f(O_1) = \orbs[\supp(X)](\hat x_1) = \orbs[\supp(X)](\hat x_2) = f(O_2)$. 
	This means that $\hat x_1, \hat x_2 \in \orbs[\supp(X)](\hat x_1)$,
	and so there is some $\pi \in \fix(\supp(X))$ such that $\pi \cdot \hat x_1 = \hat x_2$. 
	As $\fix(\supp(X)) \subseteq \Perm$, this implies $\hat x_2 \in O_1$, which is a contradiction. 
	Thus, $f$ is a 1-1 function and $|\orb(\cls{}(X))| \le |\orb[\supp(X)](X)|$ as required.	
	
	\item We have that 
	\begin{align*}
		\cls{}(X \times Y) &= \{\pi \cdot (x, y) \mid \pi \in \Perm, (x,y) \in X \times Y\}
		\\ &= \{(\pi \cdot x, \pi \cdot y) \mid \pi \in \Perm, (x, y) \in X \times Y\}
	\end{align*} 
	As $\supp(X) = \emptyset$, this means that $X$ is closed under permutation, and so we obtain:
	\begin{align*}
		\{(\pi \cdot x, \pi \cdot y) \mid \pi \in \Perm, (x, y) \in X \times Y\} &= \{(x, \pi \cdot y) \mid \pi \in \Perm, (x, y) \in X \times Y\}\\
		&= X \times \{\pi \cdot y \mid \pi \in \Perm, y \in Y\} \\ &= X\times \clos(Y)
	\end{align*}
%	3
	\item Let $x_{O}$ range over orbit representatives in $\orb(\Xcal), \orb(\Ycal)$ such that for any $O_\Xcal \in \orb(\Xcal)$, $x_{O_\Xcal} \in O_\Xcal$, and similarly for $\Ycal$.
	We define the following function:
	\[
	F(O_\Xcal, O_\Ycal) = \bigcup_{\pi' \in \Perm}\{\pi \cdot (x_{O_\Xcal}, \pi'\cdot x_{O_\Ycal}) \mid \pi \in \Perm\}
	\]
	Where $O_\Xcal \in \orb(\Xcal)$ and $O_\Ycal \in \orb(\Ycal)$. 
	Without loss of generality, let us assume that $\supp(x_{O_\Xcal}) \cap \supp(x_{O_\Ycal}) = \emptyset$. 
	Then, we can say that $\{\pi \circ \pi' \mid \pi, \pi' \in \Perm\} = \{\pi \circ \pi' \mid \pi \in \Perm, \pi' \in \fix(\Atoms \setminus \supp(x_{O_\Xcal}, x_{O_\Ycal})\}$ as both sets are closed under permutation.  
	Let us fix $\vec b_0$ to be some sequencing of $\supp(x_{O_\Ycal})$ and let us write $\vec b' \lesssim \vec b_0$ to indicate $\vec b'$ is a sub-sequence of $\vec b_0$.
	From this, we can say:
	\begin{align*}
		F(O_\Xcal, O_\Ycal) &= \bigcup_{\pi' \in \fix(\Atoms \setminus \supp(x_{O_\Xcal}, x_{O_\Ycal}))}\{\pi \cdot (x_{O_\Xcal}, \pi'\cdot x_{O_\Ycal}) \mid \pi \in \Perm\} \\
		&= \bigcup_{\substack{\vec b_0 \lesssim \vec b,\\ \{\vec a\} \subseteq \supp(x_{O_\Xcal})}} \{\pi \cdot (x_{O_\Xcal}, \move{\vec b}{\vec a}\cdot x_{O_\Ycal}) \mid \pi \in \Perm\} \\
		&= \bigcup_{\substack{\vec b_0 \lesssim \vec b,\\ \{\vec a\} \subseteq \supp(x_{O_\Xcal})}} \orbs((x_{O_\Xcal}, \move{\vec b}{\vec a}\cdot x_{O_\Ycal}))
	\end{align*}
	Next we can write $\orb(\Xcal \times \Ycal)$ in terms of $F$ as follows:
	\begin{align*}
		\orb(\Xcal\times \Ycal) &= \{\{(\pi \cdot x_{O_\Xcal},\pi'\cdot x_{O_\Ycal}) \mid \pi, \pi' \in \Perm\} \mid O_\Xcal \in \orb(\Xcal), O_\Ycal \in \orb(\Ycal)\}\\
		&= \{\{\pi\cdot (x_{O_\Xcal},\pi^{-1}\cdot \pi'\cdot x_{O_\Ycal}) \mid \pi, \pi' \in \Perm\} \mid O_\Xcal \in \orb(\Xcal), O_\Ycal \in \orb(\Ycal)\}\\
		&= \{\{\pi\cdot (x_{O_\Xcal},\pi''\cdot x_{O_\Ycal}) \mid \pi, \pi'' \in \Perm\} \mid O_\Xcal \in \orb(\Xcal), O_\Ycal \in \orb(\Ycal)\}\\
		&= \bigcup_{\substack{O_\Xcal \in \orb(\Xcal),\\ O_\Ycal \in \orb(\Ycal)}}F(O_\Xcal, O_\Ycal)\\
		\therefore\ |\orb(\Xcal \times \Ycal)| &\le \sum_{\substack{O_\Xcal \in \orb(\Xcal),\\ O_\Ycal \in \orb(\Ycal)}}|F(O_\Xcal, O_\Ycal)| 
	\end{align*}
	It remains to calculate $|F(O_\Xcal, O_\Ycal)|$ for any $O_\Xcal \in \orb(\Xcal), O_\Ycal \in \orb(\Ycal)$. 
	Let us write $m_1 = |\supp(x_{O_\Xcal})|$ and $m_2 = |\supp(x_{O_\Ycal})|$,  assuming WLOG that $m_1 \ge m_2$.
	% for any $X_A \in O_A,X_B \in O_B$ and $\pi' \in \fix(\Atoms \setminus \supp(X_A, X_B))$ we have:
	\begin{align*}
		|F(O_\Xcal, O_\Ycal)| &= |\bigcup_{\substack{\vec b \lesssim \vec b_0,\\ \{\vec a\} \subseteq \supp(x_{O_\Xcal})}}\orbs((x_{O_\Xcal}, \move{\vec b}{\vec a}\cdot x_{O_\Ycal}))|\\
		&=
		|\{\move{\vec b}{\vec a} \mid \{\vec a\} \subseteq \supp(x_{O_\Xcal}), \vec b \lesssim \vec b_0\}|
		\\&= |\{\sigma: \supp(x_{O_\Ycal}) \rightharpoonup \supp(x_{O_\Xcal}) \mid \sigma\text{ is a 1-1 function}\}|
		\\ &= m_1(m_1 - 1)\dots(m_1 - m_2 + 1) 
		\\&+ (m_2)(m_1)(m_1 - 1)\dots(m_1 - m_2 + 2) + \dots + 1
		\\ &=\sum^{m_2}_{i = 0} \choice{m_2}{i} m_1(m_1 - 1)\dots(m_1 - i + 1)  
		% \\&\le \sum^{n_\Ycal}_{i = 0} 
		% \choice{n_\Ycal}{i}
		% (n_\Xcal-1)^i \cdot \frac{n_\Xcal}{n_\Xcal - 1}
		% \\ &\le {n_\Xcal}^{n_\Ycal} \cdot \frac{n_\Xcal}{n_\Xcal-1} \le {n_\Xcal}^{n_\Ycal}(1 + \epsilon) 
	\end{align*}
        Let us call the above expression $A(m_1,m_2)$.
        Suppose $\{n_1, n_2\} = \{\regindex(\Xcal), \regindex(\Ycal)\}$ such that $n_1 \ge n_2$.
        Then we have:
        \[
        |\orb(\Xcal \times \Ycal)| \le |\orb(\Xcal)|\cdot|\orb(\Ycal)|\cdot A(n_1,n_2)
      \]
      Moreover:
      \begin{align*}
A(n_1,n_2) &\le \sum^{n_2}_{i = 0} 
		\choice{n_2}{i}
		(n_1-1)^i \cdot \frac{n_1}{n_1 - 1}
        \\ &
=  \frac{n_1}{n_1 - 1}\cdot\sum^{n_2}_{i = 0} 
		\choice{n_2}{i}
		(n_1-1)^i 
=  \frac{n_1}{n_1 - 1}\cdot {n_1}^{n_2}=  {n_1}^{n_2}\cdot (1 + \epsilon) 
      \end{align*}
      with $\epsilon=1/(n_1-1)$. We finally show that a bound without $\epsilon$ can be obtained when $n_2\ge 4$. Let us assume that $n_2\ge 4$ and set
      $B(n_1,n_2)=\sum^{n_2}_{i = 0} \choice{n_2}{i} (n_1-1)^i$. We break the sums of $A$ and $B$ as follows:
      \begin{align*}
        A(n_1,n_2) = A'(n_1,n_2)+A''(n_1,n_2) &=
\sum^{4}_{i = 0} \choice{n_2}{i} n_1(n_1 - 1)\dots(n_1 - i + 1)\\ & +
\sum^{n_2}_{i = 5} \choice{n_2}{i} n_1(n_1 - 1)\dots(n_1 - i + 1)\\
B(n_1,n_2) = B'(n_1,n_2)+B''(n_1,n_2) &=
\sum^{4}_{i = 0} \choice{n_2}{i} (n_1 - 1)^i+
\sum^{n_2}_{i = 5} \choice{n_2}{i} (n_1 - 1)^i
\end{align*}                                                
Note first that $A''(n_1,n_2)\leq B''(n_1,n_2)$ as:
\begin{align*}
& n_1(n_1-1)(n_1-2)(n_1-3)(n_1-4)\leq (n_1-1)^5\\
&\iff n_1(n_1-2)(n_1-3)(n_1-4)\leq (n_1-1)^4 \\
&\iff n_1^4-9n_1^3+26n_1^2-24n_1\leq n_1^4-4n_1^3+6n_1^2-4n_1+1\\
&\iff -5n_1^3+20n_1^2-20n_1-1\leq 0
\end{align*}
Next, writing $\kappa_i=\choice{n_2}{i}$, we have:
\begin{align*}
  A'(n_1,n_2) &= 1+n_2n_1+\kappa_2n_1(n_1-1)+\kappa_3n_1(n_1-1)(n_1-2)+\kappa_4n_1(n_1-1)(n_1-2)(n_1-3)\\
 B'(n_1,n_2) &= 1+n_2(n_1-1)+\kappa_2(n_1-1)^2+\kappa_3(n_1-1)^3+\kappa_4(n_1-1)^4
\end{align*}
and, hence:
\begin{align*}
  A'(n_1,n_2)-B'(n_1,n_2) &= n_2+\kappa_2(n_1-1)+\kappa_3(-n_1+1)+\kappa_4(-2n_1^3+5n_1^2-2n_1-1)\\
  &= -2\kappa_4n_1^3+5\kappa_4n_1^2+(\kappa_2-\kappa_3-2\kappa_4)n_1+(n_2-\kappa_2+\kappa_3-\kappa_4)
\end{align*}
so, writing $\kappa'_i=\choice{n_2}{i}/n_2$, we have $A'(n_1,n_2)-B'(n_1,n_2)\leq 0$ iff
\begin{align*}
D(n_1,n_2) = -2\kappa_4'n_1^3+5\kappa_4'n_1^2+(\kappa_2'-\kappa_3'-2\kappa_4')n_1+(1-\kappa_2'+\kappa_3'-\kappa_4')\leq0
\end{align*}
Differentiating $D(n_1,n_2)$ for $n_1$ we get:
\begin{align*}
\frac{\partial D}{\partial n_1}= -6\kappa_4'n_1^2+10\kappa_4'n_1+(\kappa_2'-\kappa_3'-2\kappa_4')
\end{align*}
By some elementary analysis we can see that the above is less or equal to 0 when $n_1,n_2\geq 4$. Given that $n_1\geq n_2$, it suffices to show that $D(n_2,n_2)\leq 0$ for $n_2\geq 4$. For the latter, we note that $D(4,4)< 0$ and that $D(n_2,n_2)$ is decreasing for $n\leq4$.
\end{enumerate}
\end{proof}

\completelattice*
\begin{proof}
	First, we show that $\Ucal_n$ is a complete lattice. 
	It suffices to show that for any $S\subseteq \Ucal_n$, $\bigsqcup S$ is the least lower bound of $S$, and $\bigsqcap S$ is the greatest upper bound of $S$. In order to show that $\bigsqcup S$ is the least upper bound of $S$, we must show that:
 \begin{enumerate}
		\item $\forall f\in S$. $f\sqsubseteq \bigsqcup S$
		\item $\forall g\in \Ucal_n.\ (\forall f\in S. f\sqsubseteq g)\implies\bigsqcup S \sqsubseteq g$
	\end{enumerate}
 	For case 1 with $f\in S$, we have that $\forall \vec a\in\Atoms. f(\vec a) \subseteq (\bigsqcup S)(\vec a)$
	and hence $f\sqsubseteq \bigsqcup S$.\\
  	For case 2, let us pick such a $g$. It follows that for all $f\in S$ and $\vec a$ we have that $f(\vec a) \subseteq g(\vec a)$,
	meaning we have $g(\vec a) \supseteq \bigcup \{f(\vec a) \mid f \in S\}$.
	Therefore, $\bigsqcup S \sqsubseteq g$.
	To show that $\bigsqcap S$ is the greatest lower bound, we must show that:
	\begin{enumerate}
		\item $\forall f\in S$. $f\sqsupseteq \bigsqcap S$
		\item $\forall g\in \Ucal_n.\ (\forall f\in S. f\sqsupseteq g)\implies\bigsqcap S \sqsupseteq g$
	\end{enumerate}
	For case 1 with $f\in S$, we have that $\forall \vec a\in\Atoms. f(\vec a)\supseteq (\bigsqcap S)(\vec a)$ and hence $f\sqsupseteq \bigsqcap S$.\\
	For case 2, let us pick such a $g$. It follows that for all $f \in S$ and $\vec a$ we have that $f(\vec a)\supseteq g(\vec a)$, meaning we have $g(\vec a) \subseteq \bigcap\{(f(\vec a)\mid f \in S\}$. Therefore, $\bigsqcap S \sqsupseteq g$.     
	With a least upper bound (the join) and a greatest lower bound (the meet), we conclude that $\Ucal_n$ is a complete lattice. 
	
	Next, we show that the function $\lambda f.\lambda \vec x.\sem[{\xi[X\mapsto f]}]{\phi}$ is monotone.
	For any $\phi',\xi', \vec y, Y$, we define:
	\[\F[\xi']{\phi'}{\vec y}{Y} = \lambda f.\lambda \vec y.\sem[{\xi'[Y\mapsto f]}]{\phi}\]
	 We need to show that for all $\vec c \in \Atoms^{|\vec x|}$
	and $f, f'\in \Atoms^{|\vec x|} \rightarrow \Ucal$ such that $f \sqsubseteq f'$: 
	\begin{align*}
		\F{\phi}{\vec x}{X}(f)(\vec c) \subseteq \F{\phi}{\vec x}{X}(f')(\vec c) &\text{ when $X$ has an even number of negations}\\
		\F{\phi}{\vec x}{X}(f)(\vec c) \supseteq \F{\phi}{\vec x}{X}(f')(\vec c) &\text{ when $X$ has an odd number of negations}
	\end{align*}
	We perform an induction on $|\phi|$
	and state our inductive hypotheses (denoted \textit{IH}) for every $\phi'$ with $|\phi'| < |\phi|$ as follows:
	\begin{enumerate}
		\item $\F{\phi'}{\vec x}{X}(f) \sqsubseteq \F{\phi'}{\vec x}{X}(f')$ ($\phi$ has an even
		number of negations)
		\item $\F{\phi'}{\vec x}{X}(f) \sqsupseteq \F{\phi'}{\vec x}{X}(f')$ ($\phi$ has a odd number of negations)
	\end{enumerate}
	First, we examine the case where $X$ is not present in $\phi$. For any $\vec c \in \Atoms^{|\vec x|}$:
	\[
	\F{\phi}{\vec x}{X}(f)(\vec c) = \F{\phi}{\vec x}{X}(f')(\vec c) = \sem{\phi\{\vec c / \vec x\}}
	\]
	We now examine all cases where $X$ is present in $\phi$, with focus on the cases where every occurrence of $X$ is within an even number of negations. 
	The base case is when $\phi = X(\vec a)$:
	\begin{align*}
		\F{\phi}{\vec x}{X}(f)(\vec a) &= \xi(X)(\vec a) = f(\vec a)\\
		\F{\phi}{\vec x}{X}(f')(\vec a) &= \xi(X)(\vec a) = f'(\vec a)
	\end{align*}
	As $f\sqsubseteq f'$, we have that $f(\vec a) \subseteq f'(\vec a)$ for all $\vec a \in \Atoms^{|\vec x|}$, 
	meaning $\F{\phi}{\vec x}{X}(f) \sqsubseteq \F{\phi}{\vec x}{X}(f')$.\\
	Next we examine inductive steps for all $\vec c \in \Atoms^{|\vec x|}$:
	\begin{itemize}
		\item $\phi = \lnot \phi'$: for $\hat f \in \{f, f'\}$, we have
		\begin{align*}
			&\F{\phi}{\vec x}{X}(\hat  f)(\vec c)  = \fstates \setminus \sem[{\xi[X \mapsto \hat f]}]{\phi'(\vec c)} =  \fstates \setminus \F{\phi'}{\vec x}{X}(\hat f)(\vec c)
		\end{align*}
		By the inductive hypothesis, we have that $\F{\phi'}{\vec x}{X}(f) \sqsupseteq \F{\phi'}{\vec x}{X}(f')$, meaning for all $\vec c \in \Atoms^{|\vec x|}$, we have that
		\begin{align*}
			\gap\F{\phi'}{\vec x}{X}(f)(\vec c) \supseteq \F{\phi'}{\vec x}{X}(f')(\vec c) \\
			\therefore\gap\F{\phi}{\vec x}{X}(f)(\vec c) \subseteq \F{\phi}{\vec x}{X}(f')(\vec c)
		\end{align*}
		as required.
		Note we use the second inductive hypothesis as $\lnot \phi'$ has one more negation than $\phi'$.
		\item $\phi = \phi_1 \lor \phi_2$: for $\hat f \in \{f, f'\}$, we have
		\[
		\F{\phi}{\vec x}{X}(\hat f)(\vec c) = \sem[{\xi[X\mapsto \hat f]}]{\phi_1\{\vec c / \vec x\})} \cup \sem[{\xi[X\mapsto \hat f]}]{\phi_2\{\vec c / \vec x\}} = \F{\phi_1}{\vec x}{X}(\hat f)(\vec c) \cup \F{\phi_2}{\vec x}{X}(\hat f)(\vec c)
		\]
		By the inductive hypothesis, we have that $\F{\phi_i}{\vec x}{X}(f)(\vec c) \subseteq \F{\phi_i}{\vec x}{X}(f')(\vec c)$ for $i \in \{1,2\}$,
		therefore
		$\F{\phi}{\vec x}{X}(f)(\vec c)\subseteq \F{\phi}{\vec x}{X}(f')(\vec c)$.
		\item $\phi = \bigvee_{y} \phi'$: for $\hat f \in \{f,f'\}$, we have
		\[
		\F{\phi}{\vec x}{X}(\hat f)(\vec c) = \bigcup_{a\in \Atoms}\sem[{\xi[X\mapsto \hat f]}]{\phi'\{a/y\}\{\vec c / \vec x\}} = \bigcup_{a \in \Atoms}\F{\phi'\{a/y\}}{\vec x}{X}(\hat f)(\vec c)
		\]
		As for all $a$, $|\phi| > |\phi'\{a/y\}|$, we can use the inductive hypothesis, therefore:    \[
		\bigcup_{a \in \Atoms}\F{\phi'\{a/y\}}{\vec x}{X}(f)(\vec c) \subseteq  \bigcup_{a \in \Atoms}\F{\phi'\{a/y\}}{\vec x}{X}(f')(\vec c)
		\]
		as required.
	    \item $\phi = \fresh{y} \phi'$: for $\hat f \in \{f,f'\}$, we have
		\begin{align*}
			\F{\phi}{\vec x}{X}(\hat f)(\vec c) &= \bigcup_{a\in \Atoms}\{(s, H) \in \sem[{\xi[X\mapsto \hat f]}]{\phi'\{a/y\}\{\vec c/\vec x\}} \mid a \notin H\} 
			\\ &= \bigcup_{a \in \Atoms}\{(s, H) \in \F{\phi'\{a/y\}}{\vec x}{X}(\hat f)(\vec c) \mid a \notin H\}
		\end{align*}
		As for all $a$, $|\phi| > |\phi'\{a/y\}|$, we can use the inductive hypothesis.
		Using the inductive hypothesis, if $(s, H) \in \F{\phi'\{a/y\}}{\vec x}{X}(f)(\vec c)$ then $(s, H) \in \F{\phi'\{a/y\}}{\vec x}{X}(f')(\vec c)$, thus we can conclude.
		\item $\phi = \diam{\ell} \phi'$: for $\hat f \in \{f,f'\}$, we have:
		\begin{align*}
		&\F{\phi}{\vec x}{X}(\hat f)(\vec c) = \{(s, H) \in \fstates \mid \exists s\xrightarrow{\ell}(s', H').(s', H')\in \sem[{\xi[X\mapsto \hat f]}]{\phi'\{\vec c / \vec x\}} \} 
		\\=\ & \{(s', H') \in \fstates \mid \exists s\xrightarrow{\ell}(s', H').(s', H') \in \F{\phi'}{\vec x}{X}(\hat f)(\vec c)\}
		\end{align*}
		Using the inductive hypothesis, if $(s', H')\in \F{\phi'}{\vec x}{X}(f)(\vec c)$ then $(s', H')\in \F{\phi'}{\vec x}{X}(f')(\vec c)$,
		thus, we can conclude.
		\item $\phi = (\mu Y(\vec y).\phi')(\vec a)$: for $\hat f \in \{f,f'\}$, we have:
    	\begin{align*}
			\F{\phi}{\vec x}{X}(\hat f)(\vec c) =\gap (\lfp(\lambda g. \lambda \vec d.\sem[{\xi[X\mapsto \hat f][Y\mapsto g]}]{\phi'(\vec c)\{\vec d/\vec y\}}))(\vec a) \\
			% \lfp(\F[{\xi[X\mapsto \hat f]}]{\phi'(\vec c)}{\vec y}{Y})=
			=\gap \bigsqcap \{g\in \Ucal_n\mid 
			\lambda \vec d. \sem[{\xi[X\mapsto \hat f][Y\mapsto g]}]{\phi'(\vec c)\{\vec d/\vec y\}}\sqsubseteq g\}(\vec a)\\
			=\gap\bigsqcap \{g\in \Ucal_n\mid 
			\lambda \vec d. \sem[{\xi[X\mapsto \hat f][Y\mapsto g]}]{\phi'\{\vec d/\vec y\}(\vec c)}\sqsubseteq g\}(\vec a)\\
			=\gap \bigsqcap\{g \in \Ucal_n\mid \lambda \vec d.\F[{\xi[Y\mapsto g]}]{\phi'\{\vec d/\vec y\}}{\vec x}{X}(\hat f)(\vec c) \sqsubseteq g\}(\vec a)\\
			=\gap \bigsqcap G_{\hat f, \vec c}(\vec a) &[\text{for economy}]
		\end{align*}
		Let us take some $ g' \in G_{f', \vec c}$.
		We have that $\lambda \vec y.\F[{\xi[Y\mapsto g']}]{\phi'}{\vec x}{X}(f')(\vec c) \sqsubseteq g'$ by definition.
		Furthermore, by the inductive hypothesis,
		we have that 
		$$\lambda \vec y.\F[{\xi[Y\mapsto g']}]{\phi'}{\vec x}{X}(f)(c)\sqsubseteq\lambda \vec y.\F[{\xi[Y\mapsto g']}]{\phi'}{\vec x}{X}(f')(c) \sqsubseteq g'$$
		Hence, $g' \in G_{f, \vec c}$ and so: 
		\[
			G_{f', \vec c} \subseteq G_{f, \vec c} \text{ and therefore } \bigsqcap G_{f, \vec c} (\vec a) \subseteq \bigsqcap G_{f', \vec c} (\vec a)
		\]
		as required.
	\end{itemize}
%	\begin{itemize}

%	\end{itemize}
\end{proof}
	
\subsection{Proof of \cref{lem:semanticsequivariant}}
\semanticsequivariant*
%	The function $\lambda \phi, \xi.\sem{\phi}$ is equivariant.
%	Therefore, 
%	for any formula $\phi$ and $\Ucal$-variable assignment $\xi$, $\supp(\sem{\phi}) \subseteq \supp(\phi) \cup \supp(\xi)$.
%\end{lemma}
\begin{proof}
	We perform proof by induction on $\phi$.
	Let $\pi$ be some permutation, we show that
	$\semp{\phi} = \pi\cdot(\sem{\phi})$.
	The bases cases are as follows:
	\begin{itemize}
		\item $a=b: \semp{(a=b)} = \sem[\pi\cdot \xi]{\pi(a) = \pi(b)} = \emptyset = \pi\cdot \emptyset = \pi\cdot\sem{a=b}$.
		\item $a=a: \semp{(a=a)} = \sem[\pi\cdot \xi]{\pi(a) = \pi(a)} = \fstates = \pi\cdot \fstates = \pi\cdot\sem{a=a}$.
		\item $X(\vec a): \semp{(X(\vec a))} = \sem[\pi\cdot\xi]{X(\pi\cdot \vec a)} = (\pi\cdot\xi)(X)(\pi\cdot \vec a)$.
		We have that $\pi\cdot\xi = \{\pi\cdot (Y, U)\mid (Y,U)\in \xi\} = \{(Y, \pi \cdot U)\mid (Y,U)\in \xi\}$. 
		Let us say $\xi(X) = U_x$, then $(\pi\cdot\xi)(X) = \pi\cdot U_x$. 
		From this we can say that $\pi\cdot U_x = \{\pi\cdot (\vec b, S)\mid (\vec b, S) \in U_x\}=\{(\pi\cdot \vec b, \pi\cdot S)\mid (\vec b, S) \in U_x\}$.
		It follows that $(\pi\cdot U_x)(\pi\cdot \vec a) = \pi\cdot S$
		which is equal to 
		$\pi\cdot (\xi(X)(\vec a)) = \pi\cdot \sem{X(\vec a)}$.
	\end{itemize}
	We now examine inductive steps.\
	\begin{itemize}
		\item $\phi_1 \vee \phi_2$:
		we have that $\sem[\pi\cdot \xi]{\pi\cdot \phi_1 \vee \pi\cdot \phi_2} = \semp{\phi_1} \cup \semp{\phi_2}$, by the inductive hypothesis, the latter is $\pi\cdot \sem{\phi_1} \cup \pi\cdot \sem{\phi_2}=\pi\cdot (\sem{\phi_1}\cup\sem{\phi_2})$ as required. 
		\item $\lnot \phi_1$: 
		By definition, $\semp{(\lnot\phi_1)} = \fstates\setminus\semp{\phi_1}$, which by applying the inductive hypothesis gives $\fstates\setminus\pi\cdot (\sem{\phi_1})$.
		As $\fstates = \pi\cdot \fstates$, this is equal to $\pi\cdot \fstates\setminus\pi\cdot (\sem{\phi_1})=\pi\cdot(\fstates\setminus(\sem{\phi_1}))=\pi\cdot(\sem{\lnot\phi_1})$. 
		\item $\phi=\bigvee_{x}\phi_1$:
		In this case, $\semp{(\bigvee_{x}\phi_1)} = \bigcup_{a\in\Atoms}\semp{(\phi_1\{a/x\})}$. By inductive hypothesis, the latter is $\bigcup_{a\in\Atoms}\pi\cdot (\sem{(\phi_1\{a/x\})})=\pi\cdot (\bigcup_{a\in\Atoms}\sem{(\phi_1\{a/x\}})=\pi\cdot (\sem{\bigvee_{x}\phi_1})$.
		\item $\phi=\fresh{x}\phi'$:
		In this case, $\semp{(\fresh{x}\phi')} = \bigcup_{a\in\Atoms}\{(s, H) \in \semp{(\phi'\{a/x\})} \mid a \notin H\}$. 
		By inductive hypothesis, the latter is $\bigcup_{a\in\Atoms}\{(s, H) \in \pi \cdot \sem{(\phi'\{a/x\})} \mid a \notin H\}
		% =\bigcup_{a\in\Atoms}\{\pi \cdot (s, H) \in \sem{(\phi'\{a/x\})} \mid \pi(a) \notin H\}
		=\pi \cdot \bigcup_{a\in\Atoms}\{(s, H) \in \sem{(\phi'\{a/x\})} \mid a \notin H\}$.
		\item $\diam{\ell}\phi_1:$ we have that
		\begin{align*}
			\semp{(\diam{\ell}\phi_1)} &= \{s\in\fstates\mid\exists s\xrightarrow{\pi\cdot\ell} s' . s' \in \semp{\phi_1}\}
			= \{s\in\fstates\mid\exists s\xrightarrow{\pi\cdot\ell} s' . s' \in \pi\cdot(\sem{\phi_1})\}\\
			&= \{s\in\fstates\mid\exists s\xrightarrow{\pi\cdot\ell} \pi\cdot s' . \pi\cdot s' \in \pi\cdot(\sem{\phi_1})\}\ (\text{Rewrite}\ [s' = \pi\cdot s'])\\ 
			&= \{s\in\fstates\mid\exists s\xrightarrow{\pi\cdot\ell} \pi\cdot s' . s' \in \sem{\phi_1}\}
			= \{s\in\fstates\mid\exists \pi^{-1}\cdot s\xrightarrow{\ell} s' . s' \in \sem{\phi_1}\}\\
			&= \{\pi\cdot s\in\fstates\mid\exists s\xrightarrow{\ell} s' . s' \in \sem{\phi_1}\}
			= \pi\cdot\{s\in\fstates\mid\exists s\xrightarrow{\ell} s' . s' \in \sem{\phi_1}\}\\
			&= \pi\cdot\sem{(\diam{\ell}\phi_1)}
		\end{align*}
		\item $(\mu X(\vec x).\phi)(\vec a):$ we have that
    \begin{align*}
	\pi\cdot \sem{(\mu X(\vec x).\phi)(\vec a)} &= \pi\cdot(\lfp(\lambda f.\lambda\vec b. \sem[{\xi[X\mapsto f]}]{\phi\{\vec b/\vec x\}}))(\vec a)\\
	&= (\lfp(\lambda f.\lambda\vec b. \pi\cdot\sem[{\xi[X\mapsto (\pi^{-1}\cdot f)]}]{\phi\{\pi^{-1} \cdot \vec b/\vec x\}}))(\pi\cdot \vec a)\\
	&= (\lfp(\lambda f.\lambda\vec b. \semp[\pi\cdot(\xi{[X\mapsto \pi^{-1}\cdot f]})]{\phi\{\pi^{-1}\cdot \vec b / \vec x\}}))(\pi\cdot \vec a)\ (IH)\\
	&= (\lfp(\lambda f.\lambda\vec b. \sem[(\pi\cdot\xi){[X\mapsto f]})]{(\pi \cdot\phi)\{\vec b / \vec x\}}))(\pi\cdot \vec a)\\
	&= \sem[\pi\cdot\xi]{(\mu X(\vec x).(\pi\cdot\phi)(\vec a)}\\
	&= \semp{(\mu X(\vec x).\phi(\vec a)}
\end{align*}
		
	\end{itemize}
\end{proof}
%The proof is similar to Lemma~\ref{lem:equivariant}, except for the following additional rule. 
%We show the new inductive step as follows:
%\begin{itemize}
%
%\end{itemize}
\subsection{Proof of \cref{thm:semanticswelldefined}}
\semanticswelldefined*
%	The semantics of \cref{def:hdsemantics} is well defined.
%	 Its image is within the subset of $\Ucal$ of finitely-supported sets of configurations.
%\end{theorem}
\begin{proof}
For the former, we must show that least fixpoints exist for any $\lfp(\lambda f.\lambda \vec b.\sem[{\xi[X \mapsto f]}]{\phi\{\vec b/\vec x\}})$. 
By definition, the function $F = \lambda f.\lambda \vec b.\sem[{\xi[X \mapsto f]}]{\phi\{\vec b/\vec x\}}$ is in $\Ucal_{\ar(X)}$ and by Lemma~\ref{lem:completelattice}, $\Ucal_{\ar(X)}$ is a complete lattice and $F$ is monotone. 
Thus, by the Knaster-Tarski theorem~\cite{Tarski}, the least fixpoint of any such $F$ exists.
\\
The proof of the latter follows from \cref{lem:semanticsequivariant}.
%\ntnote{only the second part follows from Lemma 13, not the well-defined part}
\end{proof}
%%% Local Variables:
%%% mode: latex
%%% TeX-master: "main"
%%% End:

%% file: app-modelchecking.tex
We shall be using the following generalised notion of orbits for a subset of a nominal set.

\begin{definition}
  Let $\Xcal$ be a nominal set and $S\subset\Xcal$ have finite support. For each finite $N\subseteq\Atoms$ we set:\
$
\orb[N](S) = \{ \orbs[N](x)\mid x\in S\} 
$.
\end{definition}

\subsection{Proof of \cref{lem:negfreesize}}
%\begin{lemma}\label{lem:negfreesize}
\negfreesize*
%	For any formula $\phi$, $|!(\phi)|$ is bounded by $|\phi| - n$, where $n$ is the number of negations in $\phi$. 
%\end{lemma}

\begin{proof}
	We perform an induction on the measure $(m, n)$ of $\phi$, where $m$ is the number of non-negation connectives and constructs that appear in $\phi$, 
	with the following base cases:
 	\begin{itemize}
		\item $|!(u = v)| = |u = v| = 2$,
		\item $|!(u \neq v)| = |u \neq v| = 2$,
		\item $X(\vec u) = {!(X(\vec u))}$, hence they have the same size.
	\end{itemize}
	Next, we perform inductive steps.
	The cases where $\phi$ is not of the form $\lnot \psi$ follow from the inductive hypothesis. 
	For instance, if $\phi = \phi_1 \lor \phi_2$, then 
	\[
		|\phi| = |\phi_1| + |\phi_2| + 1 \overset{\text{IH}}{\le} |!(\phi_1)| + |!(\phi_2)| + 1 = |!(\phi)|
	\]
	Thus, we only examine the more interesting cases where $\phi$ is of the form $\lnot \psi$. 
	We perform a case analysis on $\phi$:
	    \begin{itemize}
		\item $\lnot (u = v)$: this formula has $n = 1$, and $|\lnot (u = v)| = |\{u,v\}|$. 
		Then, we have
		$|!(\lnot(u = v))| = |u \neq v| = |\lnot (u = v)| - n = |\{u,v\}|$.
		Similar can be shown for $\lnot (u \neq v)$.
		\item $\lnot (\phi_1 \lor \phi_2)$: this formula has size $2 + |\phi_1| + |\phi_2|$. 
		By definition, we have that $|!(\lnot (\phi_1 \lor \phi_2))| = |!(\lnot \phi_1) \land !(\lnot \phi_2)| = |!(\lnot \phi_1)| + |!(\lnot \phi_2)| + 1$. 
		Let us write $n_1$ for the number of negations in $\phi_1$, and $n_2$ for the number of negations in $\phi_2$. 
		Thus, the number of negations in $\phi$ would be $n = n_1 + n_2 + 1$.
		The inductive hypothesis can be used as each $\lnot \phi_i$ will have one less connective than $\phi$, that being the $\lor$, meaning
		we have that $|!(\lnot \phi_i)| \le |\lnot \phi_i| - n_i - 1$ for $i \in \{1,2\}$.
		Hence, we have:
		\begin{align*}
			|!(\lnot \phi_1)| +| !(\lnot \phi_2)| + 1 & \le (|\lnot\phi_1| -n_1-1) + (|\lnot\phi_2| -n_2- 1) + 1\\
			&= |\phi_1| + |\phi_2| -(n_1+n_2+1)+ 2\\
			&= (|\phi_1| + |\phi_2| + 2) - n
		\end{align*}
		Similar can be shown for $\lnot (\phi_1 \land \phi_2)$.
		\item $\lnot \bigvee\nolimits_x \phi'$: this formula has size $2 + |\phi'|$. 
		By definition, we have that $!(\lnot \bigvee\nolimits_x \phi') = \bigwedge\nolimits_x !(\lnot \phi')$. 
		The inductive hypothesis can be used as $\lnot \phi'$ has one less connective than $\phi$, that being $\bigvee\nolimits_x$, meaning
		we have that $|!(\lnot \phi')| \le |\lnot \phi'| - n$.
		Hence, we have:
		\begin{align*}
			|!(\lnot \phi')| + 1 &\le (|\lnot \phi'| - n) + 1= |\phi'| - n + 2
		\end{align*}    
		A similar argument can be shown for $\lnot\bigwedge\nolimits_x \phi'$.
        \item $\lnot \fresh{x}\phi'$ this formula has size $2 + |\phi'|$. 
		By definition, we have that $!(\lnot \fresh{x} \phi') = \fresh{x}!(\lnot \phi')$. 
		The inductive hypothesis can be used as $\lnot \phi'$ has one less connective than $\phi$, that being $\fresh{x}$, meaning
		we have that $|!(\lnot \phi')| \le |\lnot \phi'| - n$.
		Hence, we have:
		\begin{align*}
			|!(\lnot \phi')| + 1 &\le (|\lnot \phi'| - n) + 1= |\phi'| - n + 2
		\end{align*}
		\item $\lnot \diam{t, \vec u} \phi'$: 
		this formula has size $2 + |\vec u| + |\phi'|$. 
		By definition, we have that $!(\lnot\diam{t, \vec u}\phi') = \sq{t, \vec u}!(\lnot \phi')$.
		The inductive hypothesis can be used as $\lnot \phi'$ has one less connective than $\phi$, that being $\diam{t, \vec u}$, meaning we have that $|!(\lnot \phi')| \le |\lnot \phi'| - n$. 
		Hence, we have:
		\[
		|!(\lnot \phi')| + 1 + |\vec u| \le (|\lnot \phi'| - n) + 1 + |\vec u| = 2 + |\phi'| + |\vec u| - n 
		\]
		A similar argument can be shown for $\lnot\sq{t, \vec u}\phi'$.
		\item $\lnot (\mu X(\vec x).\phi')(\vec u)$: this formula has size $2 + |\phi'| + |\vec u|$. 
		By definition, we have that $!(\lnot (\mu X(\vec x).\phi')(\vec u)) = (\nu X(\vec x).!(\lnot \phi'_X))(\vec u)$.
		The inductive hypothesis can be used as $\lnot \phi'_X$ has one less construct than $\phi$, that being $\mu X$, meaning
		we have that $|!(\lnot\phi'_X)| \leq |\lnot \phi'_X| - n'$ where $n'$ is the number of negations in $\lnot \phi'_X$.
		This means that $|\lnot \phi'_X| - n' = |\phi'| - n$ as both have the number of negations removed,
		hence, we have:
		\begin{align*}
			1 + |!(\lnot \phi'_X)| + |\vec u| &\le 1 + |\lnot \phi'_X| - n' + |\vec u| \\
			&= 1 + |\phi'| - n + |\vec u|\\
			&\le 2 + |\phi'| + |\vec u| - n
		\end{align*}
		A similar argument can be shown for $\lnot (\nu X(\vec x).\phi')(\vec u)$.
		\item $\lnot \lnot \phi'$: this formula has size $2 + |\phi'|$. 
		By definition, we have that $!(\lnot \lnot \phi') = !(\phi')$. 
		The inductive hypothesis can be used as the number of connectives for $\lnot \lnot \phi'$ and $\phi'$ are the same, however, there are two less negations. 
		Hence, by the inductive hypothesis we have $|!(\phi')| \le |\phi'| - (n - 2)$, 
		meaning $|!(\phi)| = |!(\phi')| \le |\phi'| + 2 - n $.
%		:
%		\begin{align*}
%			|!(\phi)| = |!(\phi')| &\le |\phi'| + 2 - n 
%		\end{align*}
	\end{itemize}
%    \begin{itemize}
%    
%    \end{itemize}
\end{proof}

\subsection{Proof of  \cref{lem:freshClosProps}}

\begin{lemma}\label{lem:unionOrb2}
	Given two sets $X_1,X_2\subseteq \Xcal$ such that $\supp(X_i) \subseteq \S \subseteq \Atoms$ for $i \in \{1, 2\}$,
	then $\orb[\S](X_1 \cup X_2) = \orb[\S](X_1) \cup \orb[\S](X_2)$.
\end{lemma}

\begin{proof}
	By definition, we have that:
	\begin{align*}
		\orb[\S](X_1 \cup X_2) &= \{\orbs[\S](x) \mid x \in X_1 \cup X_2\} 
		\\ &= \{\orbs[\S](x) \mid x \in X_1\} \cup \{\orbs[\S](x) \mid x \in X_2\} 
		\\ &= \orb[\S](X_1) \cup \orb[\S](X_2)    
	\end{align*}
\end{proof}

\begin{lemma}\label{lem:unionOrb3} 
	Let $\S \subseteq \Atoms$ such that $\S$ is finite and $\Xcal$ be a nominal set.
	Then for any $X \subseteq \Xcal$ and $b \notin \S$ such that $\supp(X) \subseteq \S \cup \{b\}$, 
	there exists a 1-1 function 
	$f: \orb[\S](\bigcup_{a \notin \S} (a\ b)\cdot X) \rightarrow \orb[\S \cup \{b\}](X)$.
\end{lemma}
\begin{proof}
	We define the function $f$ as follows: for each $O \in \orb[\S](\bigcup_{a \notin \S} (a\ b)\cdot X)$, such that $O = \orbs[\S]((a\ b) \cdot x)$ for some $x \in X$ and some $a \notin \S$, then $f(O) = \orbs[\S \cup \{b\}](x)$.
	The fact that $f$ is a function follows from the fact that the orbits in the codomain are taken over a larger set than the orbits in the domain, it remains to show $f$ is $1-1$. 
	Let us choose some $x_1, x_2 \in X$, so $f(\orbs[\S]((a\ b) \cdot x_i)) = \orbs[\S \cup \{b\}](x_i)$. 
	Assume $\orbs[\S]((a\ b) \cdot x_1) \neq \orbs[\S]((a\ b)\cdot x_2)$, meaning there is no $\pi \in \fix(\S)$ such that $\pi \cdot ((a\ b)\cdot x_1) = (a\ b)\cdot x_2$.     
	Let us also assume $f(\orbs[\S]((a\ b)\cdot x_1)) = f(\orbs[\S]((a\ b)\cdot x_2)) = \orbs[\S \cup \{b\}](x_1)$. 
	This means there must be some $\pi' \in \fix(\S \cup \{b\})$ such that $\pi' \cdot x_1 = x_2$, and so we also have $((a\ b)\circ \pi') \cdot x_1 = (a\ b) \cdot x_2$.
	Let us write $\hat \pi = (a\ b)\circ \pi' \circ (a\ b)$, so $((a\ b)\circ \pi') \cdot x_1 = \hat \pi \cdot ((a\ b)\cdot x_1) = (a\ b)\cdot x_2$. We have that $\hat \pi \in \fix(\S)$, however, this is a contradiction as this would mean $\orbs[\S]((a\ b) \cdot x_1) = \orbs[\S]((a\ b)\cdot x_2)$. 
	Hence, $f$ is $1-1$.
\end{proof}

\begin{lemma}\label{lem:freshClosProps1}
		Let $\Lcal$ be a nominal LTS.
	Given a negation-free formula $\phi_0$ with set of free recursion variables $\cal Z$ and a configuration $(s_0,H_0)$, thenfor all $(\phi,\gamma,\theta)\in\clos(\phi_0)$, the free variables of $\phi$ are included in $\dom(\theta)\uplus \cal Z$ and all formulas in $\rng(\theta)$ are fixpoints and have free variables in $\dom(\theta)\uplus\Zcal$.
\end{lemma}

\begin{proof}
	We denote $\fv(\psi)$ for the free variables of a formula $\psi$, 
	and we do induction on the derivation of $(\phi,\gamma,\theta)$, starting from $(\phi_0,\{\}, \{\})$. 
	The base cases are clear; let $(\phi',\gamma',\theta')$ be the parent of $(\phi, \gamma, \theta)$. 
	We do a case analysis on $\phi'$ as follows:
	%\ntnote{unclear what the case analysis below is on}
	\begin{itemize}
		\item $\phi_1 \lor \phi_2$: By the inductive hypothesis,  
		we have that $\fv(\phi_1), \fv(\phi_2) \subseteq \fv(\phi_1 \lor \phi_2) \subseteq \dom(\theta)\uplus \Zcal$.
		A similar argument can be shown for $(\phi_1 \land \phi_2)$.
		\item $\diam{\ell}\phi$: By the inductive hypothesis, 
		we have that $\fv(\phi') = \fv(\phi) \subseteq \dom(\theta) \uplus \Zcal$. 
		A similar argument can be shown for $\sq{\ell}\phi, \bigvee_x \phi'',  \bigwedge_x \phi''$ and $\fresh{x}\phi''$.
		\item $\fresh{x}\phi''$: there exists an $a \in \Atoms$ such that $\phi = \phi''\{a/x\}$. By the inductive hypothesis, we have that $\fv(\phi) = \fv(\phi'') = \fv(\phi') \subseteq \dom(\theta) \uplus \Zcal$.			
		\item $(\mu X(\vec x).\phi'')(\vec a)$: $X$ will become free in $\phi=\phi''\{\vec c/\vec x\}$, for some $\vec c$, and $\theta=\theta'[(\mu X(\vec x).\phi'')/X]$.
		By definition, $X\notin\Zcal$.
		By the inductive hypothesis, we have that 
		$\fv(\phi) =\fv(\phi'') \cup \{X\} \subseteq (\dom(\theta') \uplus \Zcal)\cup \{X\}=\dom(\theta) \uplus \Zcal$. 
		By the inductive hypothesis, the new elements in the range of $\theta$ have free variables in $\Zcal$
		and are fixpoints.
		A similar argument can be shown for $(\nu X(\vec x).\phi'')(\vec a)$.
	\end{itemize}
\end{proof}

\freshClosProps*

\begin{proof}
	\begin{enumerate}
		\item 
		For $\phi$ to be in $\clos^*(\phi_0)$, it suffices to show that there exists some $(\phi', \gamma', \theta') \in \clos(\phi_0)$ such that $\phi = \phi'\{\theta'\}\{\gamma'\}$. 
		We perform induction on the path leading to $(\phi', \gamma', \theta')$ starting from $(\phi_0, \{\}, \{\})$ with the following base case:
		\begin{itemize}
			\item $(s, H, \xi, \phi) \in \Gcal(\Lcal, \phi_0,\xi_0, s_0)$ implies $\phi \in \clos(\phi_0)$ and so $(\phi', \gamma', \theta') \in \clos(\phi_0)$.
		\end{itemize}
		Next, we perform inductive steps. Suppose the parent position of $(s, H,\xi,\phi)$ has formula:
		\begin{itemize}
			\item $\phi_1 \lor \phi_2$: by the inductive hypothesis, there is a triple $(\psi,\gamma,\theta)\in\clos(\phi_0)$ such that $\phi_1\lor\phi_2=\psi\{\theta\}\{\gamma\}$. From Lemma~\ref{lem:freshClosProps1}, there is $\psi=\psi_1\lor\psi_2$ and therefore $\phi=\psi_1\{\theta\}\{\gamma\}\lor\psi_2\{\theta\}\{\gamma\}$, and $(\psi_1,\gamma,\theta),(\psi_2,\gamma,\theta)\in\clos(\phi)$, as required. The case for  $\phi_1 \land \phi_2$ is similar.
			\item $\diam{\ell}\phi''$: 
			by the inductive hypothesis, there is a triple $(\psi,\gamma, \theta)\in\clos(\phi_0)$ such that $\diam{\ell}\phi''=\psi\{\theta\}\{\gamma\}$.
			From Lemma~\ref{lem:freshClosProps1}, there is $\psi=\diam{\ell}\psi''$, and $(\psi'', \gamma,\theta)\in \clos(\phi_0)$, as required. 
			The case for $\sq{\ell}\phi''$ is similar.
			\item $\fresh{x}\phi''$:
			by the inductive hypothesis, there is a triple $(\psi, \gamma, \theta) \in \clos(\phi_0)$ such that $\fresh{x}\phi''=\psi\{\theta\}\{\gamma\}$.
			From Lemma~\ref{lem:freshClosProps1}, there is $\psi=\fresh{x}\psi''$, and $(\psi'', (\gamma, \{a/x\}), \theta) \in \clos(\phi_0)$ as required.
			The cases for $\bigvee_x \phi''$ and $\bigwedge_x \phi''$ are similar. 
			\item $(\mu X(\vec x).\phi'')(\vec a)$: by the inductive hypothesis, there is a triple $(\psi, \gamma,\theta) \in \clos(\phi_0)$ such that $\psi\{\theta\}\{\gamma\} = (\mu X(\vec x).\phi'')(\vec a)$. 
			If $\psi$ is a variable, then $\psi = X(\vec a)$ and $(\mu X(\vec x).\phi'')(\vec a) = \theta(X(\vec a))\{\gamma\}$. 
			By Definition~\ref{def:closFresh}, there must be (for some $\vec b$) $((\mu X(\vec x).\phi')(\vec b),\gamma, \theta) \in \clos(\phi_0)$, meaning $(\phi'',(\gamma, \{\vec a / \vec x\}), \theta) \in \clos(\phi_0)$ and $\phi''\{\theta\}\{\gamma\}\{\vec a / \vec x\} = (\phi''\{\mu X(\vec x) / X\})\{\gamma\}\{\vec a / \vec x\}$ as required.
%			 \hbnote{check this again}
			\\
			If $\psi$ is not a variable, then $\psi = (\mu X(\vec x).\hat \phi)(\vec a))$ such that $ (\mu X(\vec x).\hat \phi)(\vec a))\{\theta\}\{\gamma\} =  (\mu X(\vec x).\phi'')(\vec a))$. 
			Let $\theta'' = \theta[\mu X(\vec x).\phi' / X]$.Then, from $(\psi,\gamma, \theta) \in \clos(\phi_0)$ we have $(\phi''\{\vec a/ \vec x\}, \gamma,\theta') \in \clos(\phi_0)$ and $(\phi'\{\vec a/ \vec x\})\{\theta'\}\{\gamma\} = (\phi'\{\theta'\})\{\gamma\}\{\vec a/ \vec x\} = (\phi'\{\mu X(\vec x).\phi' / X\})\{\gamma\}\{\vec a / \vec x\}$ as required.
%			\hbnote{check this again}
			
			The case for $(\nu X(\vec x).\phi')(\vec a)$ is similar.
		\end{itemize}		
		\item 
		% Let us write $\orb[\phi',\theta']$ to be the orbit that fixes the support of all formulas that appear in the path from $\clos(\phi, [])$ to $\clos(\phi', \theta')$.
		%\ntnote{update this with $\psi,\theta$. Actually, this does not work bc later we need a version of Lemma~11 that does not hold. Instead, try taking orbits wrt set $N$ which is the support of all formulas from the root $(\phi,[])$ to $(\psi,\theta)$}
		%\fixednote{better write $\orb[\phi'] = \orb[\supp(\phi,\phi')]$ as $\phi$ may contain names that do not appear in $\phi'$ anymore, but may appear in $\theta$}.
%		By definition, we have that $|\orb[\supp(\phi)](\clos(\phi))| = |\orb[\supp(\phi)](\clos(\phi))|$,
		%\fixednote{orbits needed}, 
%		so it remains to show that $|\orb[\supp(\phi)](\clos(\phi))| \le |\phi|\cdot m(\phi)$.
		%\fixednote{update}
		Let $(\psi, \gamma, \theta) \in \clos(\phi)$. 
		We write $\clos_\phi(\psi, \gamma, \theta)$ for the subset of $\clos(\phi)$ derived from and including $(\psi,\gamma, \theta)$, and let $\Ncal$ be the support of all formulas and name substitutions in the path from $(\phi,\{\}, \{\})$ to $(\psi,\gamma, \theta)$,
		therefore $\supp(\phi, \psi, \gamma, \theta) \subseteq \Ncal$.
		%\fixednote{maybe add that therefore $\supp(\phi,\psi,\theta)\subseteq\Ncal$ (this answers a comment later on as well)}
		Without loss of generality, for the fixpoint case we consider a more general $\clos$ rule for this case, which goes from $((\mu X(\vec x).\psi')(\vec a),\gamma, \theta)$ to $(\bigvee\nolimits_{x_1}\dots\bigvee\nolimits_{x_n}\psi',\gamma',\theta')$ where $\theta' = \theta\{\mu X(\vec x).\psi' / X\}$ and $\gamma' = (\gamma, \{\vec a / \vec x\})$. 
		This change allows for each $x \in \vec x$ to become free sequentially, in contrast to making all $x \in \vec x$ free at once. 
		The impact on result is unaffected as this will add more intermediate formulas (i.e., each $\bigvee_{x_i}$).
		% \fixednote{expland a bit more on "goes to", also, you need $\theta'$}
		Next we define the following function:
		\[
		m'(\phi', \gamma', \theta') = \frac{(\|\phi'\| + |\Ncal|)!}{|\Ncal|!}
		\]
		noting that $m(\phi) = m'(\phi, [])$ (where it does not cause ambiguity, the $\theta'$ and $\gamma'$ parameters in $m'$ may be dropped for economy).
		We show that $|\orb[\Ncal](\clos_\phi(\psi, \gamma, \theta))| \le |\psi|\cdot m'(\psi, \gamma, \theta)$
		%\fixednote{orbit subscript} 
		by performing induction on $\psi$:
		%\ntnote{the m needs $\theta$ as a second argument; maybe you can say we drop it for economy when not used}
		\begin{itemize}
			\item $|\orb[\Ncal](\clos_\phi(a = a, \gamma, \theta))|= 1 \le 2 = |\psi|\cdot m'(\psi)$.
			Similar can be shown for $a \neq a$.
			\item $|\orb[\Ncal](\clos_\phi(a = b, \gamma, \theta))| = 1 \le 2 = |\psi|\cdot m'(\psi)$.
			Similar can be shown for $a \neq b$.
			\item $|\orb[\Ncal](\clos(X(\vec a), \gamma, \theta))| = 1 \le 1 + |\vec a| = |\psi|\cdot m'(\psi)$.
		\end{itemize}
		Next we perform inductive steps:
		\begin{itemize}
			\item $\psi_1 \lor \psi_2$: by the inductive hypothesis, $|\orb[\Ncal](\clos_\phi(\psi_i,\gamma, \theta))| \le |\psi_i|\cdot m'(\psi_i)$ for $i \in \{1, 2\}$.
			By Lemma~\ref{lem:unionOrb2},
			%\fixednote{I think the lemma (\ref{lem:unionOrb}) cannot be applied in this direction?} 
			we have that $|\orb[\Ncal](\clos_\phi(\psi_1 \lor \psi_2,\gamma, \theta))| \le |\orb[\Ncal](\clos_\phi(\psi_1, \gamma, \theta))| + |\orb[\Ncal](\clos_\phi(\psi_2,\gamma, \theta))|$.
			As, for $i \in \{1,2\}$, $\|\psi_i\| \le \|\psi\|$ we have $m'(\psi_i) \le m'(\psi)$. 
			Thus we have:
			\begin{align*}
				|\orb[\Ncal](\clos_\phi(\psi_1 \lor \psi_2,\gamma, \theta))| &\le |\orb[\Ncal](\clos_\phi(\psi_1, \gamma, \theta))| + |\orb[\Ncal](\clos_\phi(\psi_2,\gamma, \theta))| 
				\\&\le (|\psi_1|\cdot m'(\psi_1)) + (|\psi_2|\cdot m'(\psi_2)) \le |\psi|\cdot m(\psi)
			\end{align*}%\fixednote{a m missing at end}
			A similar argument can be shown for $\psi_1 \land \psi_2$.
			\item $\diam{\ell}\psi'$: by the inductive hypothesis, $|\orb[\Ncal](\clos_\phi(\psi', \gamma, \theta))| \le |\psi'|\cdot m'(\psi')$.
			By definition
			we have that $|\orb[\Ncal](\clos_\phi(\psi, \gamma, \theta))| \le |\orb[\Ncal](\clos_\phi(\psi',\gamma,\theta))| + 1$ and $m'(\psi') = m'(\psi)$.
			%\fixednote{the ineq is the other way round than used below. Maybe find an expression with equality instead}
			Thus we have:
			\begin{align*}
				|\orb[\Ncal](\clos_\phi(\psi,\gamma,\theta))| &\le |\orb[\Ncal](\clos_\phi(\psi',\gamma,\theta))| + 1
				\\&\le |\psi'|\cdot m'(\psi') + 1 
				\\&\le |\psi|\cdot m'(\psi)
			\end{align*}
			A similar argument can be shown for $\sq{\ell}\psi'$.
			\item $\bigvee_x \psi'$: by the inductive hypothesis, $|\orb[\Ncal \cup \{a\}](\clos_\phi(\psi',(\gamma, \{a/x\}), \theta))| \le |\psi'|\cdot m'(\psi',(\gamma, \{a/x\}),\theta)$ for all $a \in \Atoms$.
			By definition, we have that 
			%\ntnote{here we should be counting names in $\Ncal$ and not in $\Ncal$ (bc of the orbits mentioning $\Ncal$); therefore, for the induction we should probably use a larger m, namely $m'(\phi',\theta)=(\|\phi'\|+|\Ncal|)!/|\Ncal|!$}
			\begin{align*}
				\clos_\phi(\psi,\gamma, \theta) &= (\bigcup\nolimits_{a \in \Ncal} \clos_\phi(\psi',(\gamma, \{a/x\}), \theta))\\ &\cup (\bigcup\nolimits_{a \notin \Ncal} \clos_\phi(\psi',(\gamma, \{a/x\}), \theta)) \cup \{(\psi,\gamma, \theta)\}
			\end{align*}
			we know, for any $b \notin \Ncal$:
			%\fixednote{here you need $b\notin N$}
			\[
			\orb[\Ncal](\bigcup\nolimits_{a \notin \Ncal} \clos_\phi(\psi',(\gamma, \{a/x\}), \theta) = \orb[\Ncal](\bigcup\nolimits_{a \notin \Ncal} (a\ b) \cdot \clos_\phi(\psi',(\gamma, \{\vec a / \vec x\}), \theta))
			\]
			and then by Lemma~\ref{lem:unionOrb3} we obtain 
			\[
			|\orb[\Ncal](\bigcup\nolimits_{a \notin \Ncal} (a\ b) \cdot \clos_\phi(\psi',(\gamma, \{\vec a / \vec x\}), \theta))| \le |\orb[\Ncal \cup \{b\}](\clos_\phi(\psi',(\gamma, \{\vec a / \vec x\}), \theta))|
			\]
			As each such, $\clos_\phi(\psi',(\gamma, \{a/x\}), \theta)$ and $\bigcup\nolimits_{a \notin \Ncal} (a\ b) \cdot\clos_\phi(\psi',(\gamma, \{ a /  x\}), \theta)$ are closed under permutations fixing $\supp(\psi,\theta)$, by Lemma~\ref{lem:unionOrb2} for $b \notin \Ncal$ we have that 
			%\fixednote{$b\notin N$}
			\begin{align*}
				\orb[\Ncal](\clos_\phi(\psi, \theta)) 
				&= \bigcup\nolimits_{a \in \Ncal} \orb[\Ncal](\clos_\phi(\psi',(\gamma, \{a/x\}), \theta)) 
				\\
				%\cup \orb[\psi](\clos_\phi(\psi'\{b/x\}, \theta))\\
				&\quad\cup \orb[\Ncal]\left(\bigcup\nolimits_{a \notin \Ncal} (a\ b) \cdot\clos_\phi(\psi',(\gamma, \{a/x\}), \theta)\right)
				\cup\{(\psi, \gamma,\theta)\} 
			\end{align*}
			and hence
			\begin{align*}
				|\orb[\Ncal](\clos_\phi(\psi, \gamma,\theta)) |
				% &= \orb[\psi](\bigcup\nolimits_{a \in \supp(\psi)} \clos_\phi(\psi'\{a/x\}, \theta) \cup \clos_\phi(\psi'\{b/x\}, \theta) \cup (\psi, \theta))\\
				%&= \orb[\psi](\bigcup\nolimits_{a \in \supp(\psi,\theta)} \clos_\phi(\psi'\{a/x\}, \theta)) \cup \orb[\psi](\clos_\phi(\psi'\{b/x\}, \theta)) \cup \orb[\psi](\{\psi, \theta\}) 
				%\\
				&\leq \sum\nolimits_{a \in \Ncal} |\orb[\Ncal](\clos_\phi(\psi',(\gamma, \{a/x\}), \theta))| \\
				&\quad + |\orb[\Ncal \cup \{b\}](\clos_\phi(\psi',(\gamma, \{ a /  x\}), \theta))|+1\\
				&\overset{\text{IH}}{\leq}
				\sum\nolimits_{a \in \Ncal} |\psi'|\cdot m'(\psi',(\gamma, \{a/x\}),\theta) \\
				&\quad + |\psi'|\cdot m'(\psi',(\gamma, \{ a /  x\})\theta)+1
			\end{align*}
			For $a \in \Ncal$, then $m'(\psi',(\gamma, \{a/x\}),\theta) = (|\Ncal| + \|\psi\| - 1)!/|\Ncal|! = m'(\psi,\gamma,\theta)/(|\Ncal|+\|\psi\|)$. 
			For $b \notin \supp(\psi',\gamma,\theta)$, then $m'(\psi',(\gamma, \{\vec a / \vec x\}),\theta) = (|\Ncal| + \|\psi\|) / (|\Ncal| + 1)! = m'(\psi)/(|\Ncal| + 1)$.
			Thus, we obtain:
			\begin{align*}
				|\orb[\Ncal](\clos_\phi(\psi,\gamma, \theta)| 
				&\le 
				\sum\nolimits_{a \in \Ncal} |\psi'|\cdot m'(\psi,(\gamma, \{a/x\}),\theta)/(|\Ncal|+1)\\
				&\quad + |\psi'|\cdot m'(\psi,\gamma,\theta)/(|\Ncal|+1)+1
				\\
				&=|\psi'|\cdot m'(\psi,\gamma,\theta)/(|\Ncal|+1)\cdot (|\Ncal|+1)+1 
				\\
				&=|\psi'|\cdot m'(\psi,\gamma,\theta)+1
				\leq |\psi|\cdot m'(\psi,\gamma,\theta)
			\end{align*}
			A similar argument can be shown for $\bigwedge\nolimits_x \psi'$ and $\fresh{x} \psi'$.
			\item $(\mu X(\vec x).\psi')(\vec a)$: 
			Let $\theta' = \theta[\mu X(\vec x).\psi' / X]$.
			%\fixednote{technically speaking, we cannot exactly apply the IH here as $\hat\psi$ is not in the closure and therefore it does not have a $\Ncal$. W can say at the beginning of the proof for part 3 that we consider a more general clos rule for this case, which takes us from $(\phi,\theta)$ to $(\hat\psi,\theta')$}
			%
			Let us write $\hat \psi = \bigvee\nolimits_{x_1}\dots\bigvee\nolimits_{x_n}\psi'$ where $x_i$ is the $i$th element of $\vec x$.
			By definition, we have that
			\[
			\clos_\phi((\mu X(\vec x).\psi')(\vec a), \gamma, \theta) = \{(\mu X(\vec x).\psi')(\vec a), \gamma, \theta)\} \cup \clos_\phi(\hat \psi, (\gamma, \{\vec a / \vec x\}), \theta')
			\]
			and by Lemma~\ref{lem:unionOrb2}:
			% noting $\supp(\psi,\theta) \supseteq\supp(\hat \psi,\theta')$:
			\[
			\orb[\Ncal](\clos_\phi((\mu X(\vec x).\psi')(\vec a),\gamma, \theta)) \subseteq \{((\mu X(\vec x).\psi')(\vec a),\gamma, \theta)\} \cup \orb[\Ncal](\clos_\phi(\hat \psi,(\gamma, \{\vec a/\vec x\}), \theta'))
			\]
			By the inductive hypothesis, we have:
			\[
			|\orb[\Ncal](\clos_\phi(\hat \psi,(\gamma, \{\vec a/\vec x\}) \theta'))| \le |\hat \psi|\cdot m'   (\hat \psi,(\gamma, \{\vec a/\vec x\}),\theta')
			\]
			We note that
			$m'(\hat \psi,(\gamma, \{\vec a/\vec x\}),\theta') = m'(\psi,\gamma,\theta)$, thus obtaining:
			%\ntnote{I think Lemma 11 is applied the wrong way round here?}
			\begin{align*}
				|\orb[\Ncal](\clos_\phi(\psi,\gamma, \theta))| &\le |\orb[\Ncal](\clos_\phi(\hat \psi,(\gamma, \{\vec a/\vec x\})\theta'))| + 1
				\\ &\overset{\text{IH}}{\le} |\hat \psi|\cdot m'(\hat \psi) + 1
				\\ &\leq (|\vec x| + |\psi'|)\cdot m'(\psi,\gamma,\theta) + 1
				\leq |\psi|\cdot m'(\psi,\gamma,\theta)
			\end{align*}
			A similar argument can be shown for $(\nu X(\vec x).\psi')(\vec a)$.
		\end{itemize}
		\item    We show that for any $(\phi',\gamma, \theta) \in \clos(\phi)$, $|\supp(\phi')| + \|\phi'\| \le |\supp(\phi)| + \|\phi\|$.
		% We show that for any $(\phi', \theta) \in \clos(\phi)$ with parent $(\phi'', \theta)$, then $|\supp(\phi')| \le |\supp(\phi'')| + \|\phi''\|$.
		Let $(\phi'',\gamma'\theta')$ be the parent node of $(\phi',\gamma' \theta)$.
		We do induction on the derivation of $(\phi',\gamma, \theta)$, starting from $(\phi,\{\}, \{\})$; the base cases are clear.
		We perform a case analysis on $\phi''$ as follows:
		\begin{itemize}
			\item $\phi_1 \lor \phi_2$: 
			By the inductive hypothesis, $|\supp(\phi_1 \lor \phi_2)| + \|\phi_1 \lor \phi_2\| \le |\supp(\phi)| + \|\phi\|$. 
			We have that $\supp(\phi_1,\phi_2) \subseteq \supp(\phi_1 \lor \phi_2)$ and $\|\phi_i\| \le \|\phi_1 \lor \phi_2\|$ for $i \in \{1,2\}$, and so we obtain:
			\[
			|\supp(\phi_i)| \le |\supp(\phi_1 \lor \phi_2)| \le  |\supp(\phi)| + \|\phi\|
			\]
			Similar can be shown for $\phi_1 \land \phi_2$.
			\item $\diam{\ell} \phi'$: By the inductive hypothesis, $|\supp(\diam{\ell}\phi')| + \|\phi'\| \le |\supp(\phi)| + \|\phi\|$, and by definition, $\supp(\phi')\subseteq \supp(\diam{\ell}\phi')$ and $\|\phi'\| = \|\diam{\ell}\phi'\|$, therefore:
			\[
			|\supp(\phi')| \le |\supp(\diam{\ell}\phi')| \le |\supp(\phi)| + \|\phi\|
			\]
			Similar can be shown for $\sq{\ell} \phi'$.
			\item $\fresh{x} \phi'$: By the inductive hypothesis, we have that $|\supp(\fresh{x} \phi')| + \|\fresh{x} \phi'\| \le |\supp(\phi)| + \|\phi\|$. 
			We have that $\supp(\phi') \subseteq \supp(\fresh{x} \phi')$ and $\|\phi'\| + 1= \|\fresh{x}\phi' \|$, and so:
			\[
				|\supp(\phi')| + \|\phi'\| \le |\supp(\fresh{x}\phi')| + \|\fresh{x}\phi'\|
			\]
			Similar can be shown for $\bigvee_x \hat \phi$ and $\bigwedge_x \hat \phi$.
			\item $(\mu X(\vec x).\phi')(\vec u)$: By the inductive hypothesis, we have that $|\supp((\mu X(\vec x).\phi')(\vec u))| + \|(\mu X(\vec x).\phi')(\vec u)\| \le |\supp(\phi)| + \|\phi\|$. 
			We know that $\supp(\phi') \subseteq \supp(\phi)$ and $\|\phi'\| + |\vec x| = \|\phi\|$, and so:
			\[
				|\supp(\phi')| + \|\phi'\| \le |\supp(\phi)| + \|\phi\|
			\]
%			We know that $\phi' = \hat \phi\{\vec a/\vec x\}$ for some $\vec a \in \Atoms^{\ar(X)}$, so $\supp(\phi') \subseteq \supp((\mu X(\vec x).\hat \phi)(\vec u)) \cup \{\vec a\}$.
%			We also have that $\|(\mu X(\vec x).\hat \phi)(\vec u)\| = \|\phi'\| + |\vec a|$.
%			Therefore:
%			\[
%			|\supp(\phi')| \|\phi'\| \le (|\supp((\mu X(\vec x).\hat \phi)(\vec u))| + |\vec a|) + (\|(\mu X(\vec x).\hat \phi)(\vec u)\| - |\vec a|) \le |\supp(\phi)| + \|\phi\|
%			\]
			Similar can be shown for $(\nu X(\vec x).\hat \phi)(\vec u)$.
		\end{itemize}
	\end{enumerate}
\end{proof}

\subsection{Proof of \cref{thm:FnomSatPar}}

\begin{lemma}\label{lem:parityNom}
	In a game $\Gcal(\Lcal, \psi, \xi, s, H)$, if defender follows a winning strategy from some position $(s, \xi, (\mu X(\vec x).\phi)(\vec a))$,
	then there must be finitely many unfoldings of the form $(\mu X(\vec x).\phi)(\vec k)$ for $\vec k \in \Atoms^{\ar(X)}$. 
\end{lemma}
\begin{proof}
	If defender has a winning strategy from $(s, \xi, (\mu X(\vec x).\phi)(\vec a))$, then this means for every play one of the following possibilities must be true:
	\begin{enumerate}
		\item the play is finite,
		\item the play is infinite with some position with a higher even rank.
	\end{enumerate}
	For point 1, it follows that in a finite play, there will be finitely many unfoldings of $(\mu X(\vec x).\phi)(\vec a)$ as required.
	For point 2, suppose that there is an infinite play that defender wins.
	Let us say that variable $Y$ has the maximum infinitely recurring rank $d$, and as it is an even rank, $d$ is owned by defender.
	By the previous lemma, if a $(\mu X(\vec x).\phi)(\vec a)$-position reaches a $(\nu Y(\vec y).\phi')(\vec b)$-position that then reaches some $(\mu X(\vec x).\phi)(\vec k)$-position (for any $\vec k \in \Atoms^{\ar(X)}$), then $X \le_\psi Y$ and so $X$ has a higher rank than $d$.
	If this rank recurs infinitely often, then this means the play is losing for defender, and so $\mu X$ must unfold finitely many times.
\end{proof}

\begin{lemma}\label{lem:parityNom2}
	In a game $\Gcal(\Lcal, \psi, \xi, \hat s, \hat H)$, 
	if we have a $(\sigma X(\vec x).\phi')(\vec a)$-position  that reaches a $(\sigma' Y(\vec y).\phi'')(\vec b)$-position which then reaches a $(\sigma X(\vec x).\phi')(\vec c)$-position, then $X \le_\psi Y$.  
\end{lemma}%\fixednote{Update as in Lemma 81?}
\begin{proof}
	We let $\theta = \{\sigma X(\vec x).\phi'/ X\}$ and $\gamma = \{\vec a/ \vec x\}$ be substitutions and $\phi=(\sigma X(\vec x).\phi')(\vec a)$.
	In the game, the position, say, $(s, H, \xi, (\sigma X(\vec x).\phi')(\vec a))$ can make a move to $(s, H, \xi, \phi'\{\theta\}\{\gamma\})$, in particular this will include $(X(\vec c))\{\theta\}\{\gamma\} = (\sigma X(\vec x).\phi')(\vec c)$. 
	From this point, a (first) position containing $(\sigma' Y(\vec y).\phi'')(\vec b)$ can be reached.
	Then, there will be substitutions $\theta', \gamma'$ and subformula $\hat\phi''$ of $\phi$ such that $\theta \subseteq \theta'$, $\gamma \subseteq \gamma'$,
	$(\sigma'Y(\vec y).\hat \phi'')(\vec b)\{\theta'\}\{\gamma'\} = (\sigma' Y(\vec y).\phi'')(\vec b)$
	and 
	$((\sigma' Y(\vec y).\hat\phi'')(\vec b), \gamma', \theta') \in \clos(\phi)$
	(cf.~\lemmapoint{lem:freshClosProps}{2}).
	From this position, a position containing some $(\sigma X(\vec x).\phi')(\vec c)$ can be reached.
	As before,
	there must exist some $\theta' \subseteq \theta''$, $\gamma' \subseteq \gamma''$ 
	such that $\hat \phi\{\theta''\}\{\gamma''\} = (\sigma X(\vec x).\phi')(\vec c)$ and $(\hat \phi, \gamma'', \theta'') \in \clos(\phi)$. 
	Since $\hat \phi$ must be a subformula of $(\sigma X(\vec x).\phi')(\vec a)$, we can only have that $\hat \phi = X(\vec c)$, 
	meaning $X(\vec c)$ is free in $(\sigma Y(\vec y).\phi'')(\vec b)$ and so $X \le_\psi Y$.
\end{proof}

\FnomSatPar*
%	For every negation-free formula $\phi_0$,
%	$\Ucal$-variable assignment $\xi$,
%	fresh-orbit-finite F-LTS $\Lcal$ and $(s_0, H_0) \in \fstates$, then $(s_0, H_0) \in \sem{\phi_0}$ iff Defender wins from position $(s_0, H_0,\xi, \phi_0)$ in $\fgame(\Lcal, \phi_0, \xi, s_0,H_0)$.
%\end{theorem}

\begin{proof}
	Let us choose some \fhml\ formula $\phi$ with a $\Ucal$-variable assignment $\xi$ and fresh-orbit-finite F-LTS $\Lcal$
	We perform 
	induction on $\phi$ with the following base cases:
	\begin{itemize}
		\item $(s, H, \xi, a = a)$: $(s, H)\in\sem{\phi} = \fstates$, defender wins (similarly for $a \neq b$),
		\item $(s, H, \xi, a = b)$: $(s, H)\in\sem{\phi} = \emptyset$, attacker wins (similarly for $a \neq a$),
		\item $(s, H, \xi, X(\vec a))$: if $(s, H)\in \xi(X)(\vec a)$, defender wins and $(s, H) \in \sem{\phi}$, else attacker wins and $(s, H)\notin \sem{\phi}$
	\end{itemize}
	Next we perform inductive steps :
	\begin{itemize}
		\item $(s, H, \xi, \phi_1 \lor \phi_2)$:
		Let us assume $(s, H) \in \sem{\phi_1 \lor \phi_2}$.
		By the inductive hypothesis, if $(s, H) \in \sem{\phi_i}$, then defender has a winning strategy from $(s, H, \phi_i)$ for some $i \in \{1, 2\}$. 
		For any such $i$, if $s \in \sem{\phi_i}$ then $(s, H)\in \sem{\phi_1 \lor \phi_2}$, and defender has a winning strategy by the play $
		(s, H, \xi, \phi_1 \lor \phi_2) \rightarrow (s, H, \xi, \phi_i)$.
		
		Next, let us assume that defender has a winning strategy from $(s, H, \xi, \phi_1 \lor \phi_2)$. 
		By the inductive hypothesis, if defender wins from $(s, H, \xi, \phi_i)$ then $(s, H)\in \sem{\phi_i}$ for some $i \in \{1, 2\}$.
		As defender wins from $(s, H, \xi, \phi_1 \lor \phi_2)$, defender wins from $(s, H, \xi, \phi_i)$ for some $i \in \{1, 2\}$, then $(s, H) \in \sem{\phi_i} \subseteq \sem{\phi_1 \lor \phi_2}$ as required.
		
		The proof for $\phi_1 \land \phi_2$ is similar.
		\item $(s, \xi, \bigvee_x \phi')$: 
		By the inductive hypothesis, if $(s, H)\in \sem{\phi'\{a/x\}}$, then defender has a winning strategy from $(s, H, \xi, \phi\{a/x\})$ for some $a \in \Atoms$.
		For any such $a$, if $(s, H)\in \sem{\phi'\{a/x\}}$ then $(s, H) \in \sem{\bigvee_x \phi'}$, and defender has a winning strategy by the move $(s, H, \xi, \bigvee_x \phi') \rightarrow (s, H, \xi, \phi'\{a/x\})$.
		
		Next, let us assume that defender has a winning strategy from $(s, H, \xi, \bigvee_x \phi')$. 
		By the inductive hypothesis, if defender wins from $(s, H, \xi, \phi'\{a/x\})$ then $(s, H) \in \sem{\phi'\{a/x\}}$ for some $a \in \Atoms$.
		As defender wins from $(s, H, \xi, \bigvee_x \phi')$, defender wins from $(s, H, \xi, \phi\{a/x\})$ for some $a \in \Atoms$, then $(s, H) \in \sem{\phi\{a/x\}} \subseteq \sem{\bigvee_x \phi'}$ as required.
		
		The proof for $\bigwedge_x \phi'$ and $\fresh{x} \phi'$ is similar.
		\item $(s, H,\diam{\ell} \phi')$: If $(s, H)\in \sem{\diam{\ell} \phi'}$, then there exists some $(s', H')$ such that $(s, H) \xrightarrow{\ell} (s', H')$ and $(s', H')\in \sem{\phi'}$. 
		By the inductive hypothesis, if $(s', H')\in \sem{\phi'}$ then defender has a winning strategy from $(s', H', \xi, \phi')$.
		As this position is played by defender, it suffices to choose $(s', H', \xi, \phi')$, meaning defender has a winning strategy.  
		
		Next, let us assume that defender has a winning strategy from $(s, H, \xi, \diam{\ell} \phi')$. 
		This means there is some move $(s, H, \xi, \diam{\ell} \phi') \to (s', H', \xi, \phi')$ with $(s, H) \xrightarrow{\ell} (s', H')$. 
		By the inductive hypothesis, if defender wins from $(s', H', \xi, \phi')$ then $(s', H') \in \sem{\phi'}$, thus:
		\[
			(s, H) \in \{(s_1, H_1) \in \fstates \mid \exists (s_1, H_1) \xrightarrow{\ell} (s_2, H_2).(s_2, H_2) \in \sem{\phi'}\} = \sem{\phi}
		\]	
%		which implies that $(s, H) \in \sem{\diam{\ell}\phi'}$ as required.
% \hbnote{@NT: changed this as I believe this makes more sense, please check if you can!}
%		By the inductive hypothesis, if defender wins from $(s, H, \xi, \phi')$ then $(s, H) \in \sem{\phi'}$.
%		As defender wins from $(s, H, \xi, \diam{\ell} \phi')$, defender wins from $(s, H, \xi, \phi')$, then $(s, H) \in \sem{\phi'} \subseteq \sem{\diam{\ell}\phi'}$ as required.
%		
		The proof for $\sq{\ell} \phi'$ is similar.
		
	\end{itemize}
	The more interesting cases are the fixpoint operators.
	Starting with $\mu$, we want to show that, for any $(s', H')$:
	\[
	(s', H')\in \sem{(\mu X(\vec x).\phi')(\vec a)} \text{ iff defender wins from } (s', H', \xi, (\mu X(\vec x).\phi')(\vec a)) \text{ in } \Gcal(\Lcal, \phi, \xi, s, H)
	\]
	We begin with left-to-right implication, assuming that $(s', H') \in \sem{(\mu X(\vec x).\phi')(\vec a)}$. 
	We define the following function:%\fixednote{typo: $\phi'$?}
	\[
	F = \lambda f.\lambda \vec b.\sem[{\xi[X \mapsto f]}]{\phi'\{\vec b / \vec x\}}
	\]
	We denote that for any $k$ that $\xi^k = \xi[X \mapsto F^k(\emptyset)]$.
	Furthermore, we denote $\Gcal_{k}$ to be the game $\Gcal(\Lcal, \phi, \xi^k, s, H)$ for all $k$, and denote simply $\Gcal$ for $\Gcal(\Lcal, \phi, \xi, s, H)$.
	By definition, we have that:
	\[
	\sem{(\mu X(\vec x).\phi')(\vec a)} = \lfp(F)(\vec a) = (\bigcup_{k \in \omega} F^k(\emptyset))(\vec a)
	\]
	We show that for all $k$ and $\vec c$, if defender wins from position $(s', H', \xi^k, \phi'\{\vec c/\vec x\})$ in $\Gcal_k$,
	then defender wins from $(s', \xi, (\mu X(\vec x).\phi')(\vec c))$.
	Note that the latter has unique next position $(s', H', \xi, (\phi'\{\mu X(\vec x).\phi' / X\}\{\vec c / \vec x\}))$. 
	We perform induction on $k$. 
	If $k = 0$, then moves from $(s', \xi^0, \phi'\{\vec c / \vec x\})$ in $\Gcal_0$ never reach a position $(s'',H'', \xi^0, X(\vec d))$ for some $\vec d$. 
	This is because $F^0(\emptyset)(\vec d) = \emptyset$, and so it is not possible to win from $(s'', H'', \xi^0, X(\vec d))$.
	Thus, the same strategy used to win from $(s', H', \xi^0, \phi'\{\vec c / \vec x\})$ can be used to win from $(s', H', \xi, (\mu X(\vec x).\phi')(\vec c))$ in $\Gcal$.
	Suppose that defender wins from $(s', H', \xi^{k + 1}, \phi'\{\vec c / \vec x\})$ in $\Gcal_{k + 1}$, then we obtain the following strategy for $(s', H', \xi, (\phi'\{\mu X(\vec x).\phi' / X\})\{\vec c / \vec x\})$:
	play as normal in $\Gcal$ from $(s', H', \xi, (\phi'\{\mu X(\vec x).\phi' / X\})\{\vec c / \vec x\})$.
	If some $(s'', H'', \xi, (\mu X(\vec x).\phi')(\vec d))$ is reached, then $(s'', H'', \xi^{k + 1}, X(\vec d))$ is reached in $\Gcal_{k + 1}$ from $(s', H', \xi^{k+1}, \phi'\{\vec c / \vec x\})$. 
	By definition, defender is winning from $(s'', H'', \xi^{k + 1}, X(\vec d))$ in $\Gcal_{k+1}$ iff $(s'', H'')\in \sem[\xi^{k + 1}]{X(\vec d)}$.
	Then $(s'', H'') \in \xi^{k + 1}(X)(\vec d) = (F^{k + 1}(\emptyset))(\vec d) = \sem[\xi^k]{\phi'\{\vec d / \vec x\}}$.
	Hence, using the outer induction hypothesis we deduce that defender wins from $(s'', H'',\xi^k, \phi'\{\vec d / \vec x\})$ in $\Gcal_k$, and so by the inner inductive hypothesis, defender has a winning strategy from $(s'', H'', \xi, (\mu X(\vec x).\phi')(\vec d))$ in $\Gcal$, thus this completes the proof of the inner claim.
	\\
	As $(s', H') \in \sem{(\mu X(\vec x).\phi')(\vec a)}$, then $(s', H') \in \bigcup_{j \in \omega} (F^j(\emptyset))(\vec a)$ so $(s', H') \in (F^k(\emptyset))(\vec a)$ for some $k \in \omega$. 
	Therefore, we have that $k > 0$ and $(s', H') \in (F(F^{k - 1}(\emptyset)))(\vec a) = \sem[\xi^k]{\phi'\{\vec a / \vec x\}}$.
	Thus, using the outer inductive hypothesis, defender wins from $(s, H, \xi^k, \phi'\{\vec a / \vec x\})$ in $\Gcal_k$ and so defender wins from $(s', H', \xi, (\mu X(\vec x).\phi')(\vec a))$ in $\Gcal$.
	
	Next, suppose defender wins from $(s', H', \xi, (\mu X(\vec x).\phi')(\vec a))$ in $\Gcal$. 
	By Lemma~\ref{lem:parityNom}, and as defender is winning, every play in this game must have finitely many unfoldings of the form $(\mu X(\vec x).\phi')(\vec c)$ for $\vec c \in \Atoms^{\ar(X)}$.
	Let us consider the winning strategy for defender from $(s', H', \xi, (\mu X(\vec x).\phi')(\vec a))$.
	This can be interpreted as a sub-graph of the full game where, in defender positions, there is a unique successor position respecting the winning strategy (in attacker positions, all successor positions are present). 
	We call positions of the form $(s'', H'', \xi, (\mu X(\vec x).\phi')(\vec c))$ \emph{critical}. 
	By the previous argument, each critical position in the graph, in any play, can each reach finitely many other critical positions. 
	We define the height of every critical position $c$ (denoted $h(c)$) as the maximum number of critical positions reached during any play from $c$. 
	We show that for each $c = (s'', H'', \xi, (\mu X(\vec x).\phi')(\vec c))$, then $(s'', H'')\in \sem[\xi^{h-1}]{\phi'\{\vec c / \vec x\}}$ with $h = h(c)$. 
	This then implies that $(s', H') \in \sem[\xi^{h - 1}]{\{\phi'\{\vec c / \vec x\}} \subseteq \sem{(\mu X(\vec x).\phi')(\vec a)}$.
	We perform induction on $h$ with $h \ge 1$ (a critical position can always reach itself). 
	For $h = 1$, the defender wins from $(s'',H'', \xi, (\mu X(\vec x).\phi')(\vec c))$ without reaching any other critical position, therefore defender also wins from $(s'', \xi^1, \phi'\{\vec c / \vec x\})$ in $\Gcal_1$.
	Thus by the outer inductive hypothesis, $(s'', H'') \in \sem{\phi'\{\vec c / \vec x\}}$.
	Suppose that $c$ has height $h + 1$. 
	Then the defender wins from $(s'', H'', \xi^h, \phi\{\vec c / \vec x\})$ in $\Gcal_h$ as they can play as for $(s'',H'', \xi, (\mu X(\vec x).\phi')(\vec c))$ in $\Gcal$ and if they reach some $(s''', H''', \xi^h, X(\vec d))$ then, in the game with a winning strategy $\Gcal$, this reaches the winning position $(s''', H''', \xi, (\mu X(\vec x).\phi')(\vec d))$ with height less than or equal to $h$. 
	By the inner inductive hypothesis, we have that $(s''', H''')\in \sem[\xi^{h-1}]{\phi'\{\vec d / \vec x\}} = (F^h(\emptyset))(\vec d) = \xi^h(X)(\vec d)$, thus defender wins from $(s''', \xi^h, X(\vec d))$ in $\Gcal_h$ and so defender wins from $(s', H', \xi^h, \phi'\{\vec c / \vec x\})$. 
	By applying the outer inductive hypothesis we obtain $(s'', H'')\in \sem[\xi^h]{\phi'\{\vec c / \vec x\}}$.  

	Next, we examine the case of $\nu$. 
	Similarly to $\mu$, we want to show that:
	% for any $\vec a$:
	\[
	(s', H')\in \sem{(\nu X(\vec x).\phi')(\vec a)} \text{ iff defender wins from }(s', H', \xi, (\nu X(\vec x).\phi')(\vec a))\text{ in }\Gcal(\Lcal, \phi, \xi, s, H)
	\]
	The latter position has the following unique successor position
	$(s', H', \xi, (\phi'\{\nu X(\vec x).\phi'/X\})\{\vec a/\vec x\})$.
	We begin with left-to-right implication, assuming that $(s', H')\in \sem{(\nu X(\vec x).\phi')(\vec a)}$.
	We note that:
	\begin{align*}
		&\sem{(\nu X(\vec x).\phi')(\vec a)} = \sem[\xi']{\phi'\{\vec a / \vec x\}}\\
		&\text{ with } \xi' = \xi[X \mapsto \sem{(\nu X(\vec x).\phi')}]
	\end{align*}
	By the inductive hypothesis, $(s', H') \in \sem[\xi']{\phi'\{\vec a / \vec x\}}$ implies that defender wins from $(s', H', \xi', \phi'\{\vec a / \vec x\})$.
	We build a winning strategy 
	from $(s', H',\xi,(\nu X(\vec x).\phi')(\vec a))$
	for the defender as follows: the defender follows the strategy for $(s', H', \xi', \phi'\{\vec a / \vec x\})$.
	If some position $(s'', H'', \xi, (\nu X(\vec x).\phi')(\vec b))$ is reached, this would correspond to reaching $(s'', H'', \xi', X(\vec b))$ in the game the defender wins. 
	As the defender is winning, this means $(s'', H'')\in \xi'(X)(\vec b) = \sem{(\nu X(\vec x).\phi')(\vec b)}$, hence $(s'', H'')\in \sem[\xi']{\phi'\{\vec b / \vec x\}} = \sem{(\nu X(\vec x).\phi')(\vec b)}$.
	Then, by the inductive hypothesis, defender wins from $(s'', H'', \xi, \phi'\{\vec b / \vec x\})$. 
	We continue in our game to $(s'', H'',\xi, (\phi'\{\nu X(\vec x).\phi'/X\})\{\vec b / \vec x\})$ following the winning strategy for $(s'', H'', \xi, \phi\{\vec b / \vec x\})$, and so on.
	This strategy is winning as:
	\begin{enumerate}
		\item if there are finitely many unfoldings of the form $X(\vec k)$, then the last $(s'', H'', \xi, \phi\{\vec k / \vec x\})$ is winning for the defender,
		\item otherwise, it must be that 
		some $X(\vec k)$ has infinitely many unfoldings, so defender wins as $\Omega((\nu X(\vec x).\phi')(\vec k))$ is even and must be the highest infinitely repeating rank (c.f Lemma~\ref{lem:parityNom2}).
	\end{enumerate}
	Lastly, we assume defender has a winning strategy from $(s', H', \xi, (\nu X(\vec x).\phi')(\vec a))$. 
	We redefine for any $j$ that $\xi^j = \xi[X \mapsto F^j(\fstates)]$.
%	As before, we denote $\Gcal_{k}$ to be the game $\Gcal(\Lcal, \phi, \xi^k, s, H)$ for all $k$, and denote simply $\Gcal$ for $\Gcal(\Lcal, \phi, \xi, s, H)$. 
	By definition we have that:
	\[
		\sem{(\nu X(\vec x).\phi')(\vec a)} = \gfp(F)(\vec a) = (\bigcap_{j \in \omega} F^j(\Scal))(\vec a)
	\]
	We claim that the defender wins from $(s', H', \xi^j,\phi\{\vec a / \vec x\})$ for all $j$. 
	For $j = 0$, we just play as in $(s', H', \xi, (\phi\{\nu X(\vec x).\phi' / X\})\{\vec a / \vec x\})$, and if we reach some $(s'', H', \xi^0, X(\vec c))$, then defender wins. 
	This is because $(s'', H'') \in \xi^0(X)(\vec c) = (F^0(\Scal))(\vec c) = \Scal$. 
	If some $(s'', H'', \xi^{j + 1}, \phi\{\vec c / \vec x\})$ is reached, then we need to show that $(s'', H'')\in \xi^{j+1}(X)(\vec c) = (F^{j + 1}(\fstates))(\vec c)$. 
	In the game from $(s', H', \xi, (\phi'\{\nu X(\vec x) / X\})\{\vec a / \vec x\})$, we would reach $(s'', H'', \xi, (\nu X(\vec x).\phi')(\vec c))$. 
	By the inductive hypothesis, defender wins from $(s'', H'',\xi, \phi\{\vec c / \vec x\})$, and by the outer inductive hypothesis, $(s'', H'')\in \sem[\xi^j]{\phi\{\vec c / \vec x\}} = \xi^{j + 1}(X)(\vec c)$ as required.
	Thus, by the outer inductive hypothesis, $(s', H')\in \sem[\xi^j]{\phi\{\vec c / \vec x\}}$ for all j, and therefore 
	\[
	(s', H')\in (\bigcap_{j \in \omega}F^j(\Scal))(\vec a) = \sem{(\nu X(\vec x).\phi')(\vec a)}
	\]
\end{proof}

%\beyondthisline
%\begin{proof}
%    The proof is symmetric to Theorem~\ref{thm:nomSatPar}. We show the inductive step for the additional $\fresh{x}$ connective:
%    \begin{itemize}
%        \item $(s, H, \xi, \fresh{x}\phi')$: 
%        By the inductive hypothesis, if $(s, H)\in \sem{\phi'\{a/x\}}$, then defender has a winning strategy from $(s, H, \xi, \phi\{a/x\})$ for some $(a, H) \in \Atoms \setminus H$.
%        For any such $a$, if $(s, H)\in \sem{\phi'\{a/x\}}$ then $(s, H) \in \sem{\fresh{x} \phi'}$, and defender has a winning strategy by the move $(s, H, \xi, \fresh{x} \phi') \rightarrow (s, H, \xi, \phi'\{a/x\})$.
%
%        Next, let us assume that defender has a winning strategy from $(s, H, \xi, \fresh{x} \phi')$. 
%        By the inductive hypothesis, if defender wins from $(s, H, \xi, \phi'\{a/x\})$ then $(s, H) \in \sem{\phi'\{a/x\}}$ for some $a \in \Atoms \setminus H$.
%        As defender wins from $(s, H, \xi, \fresh{X} \phi')$, defender also wins from $(s, H, \xi, \phi\{a/x\})$ for some $a \in \Atoms\setminus H$, then $(s, H) \in \sem{\phi\{a/x\}} \subseteq \sem{\fresh{x} \phi'}$ as required.
%    \end{itemize}
%  \end{proof}

\subsection{Proof of \cref{lem:boundingisnombisim}}

\boundingisnombisim
\begin{proof}
  Let us write $\Rcal$ for $\Rcal_{(\Lcal, \phi_0, s_0, H_0)}$.
Let us choose some $((s, H_1, \phi),(s, H_2, \phi)) \in \Rcal_{(\Lcal, \phi_0, s_0, H_0)}$. By definition, the rank and polarity of $(s, H_1, \phi)$ and $(s, H_2, \phi)$ are the same.
    %We how that each move from $(s, H_1, \phi) \rightarrow (s', H_1', \phi')$ has a matching move $(s, H_2, \phi) \rightarrow (s', H_2', \phi')$ such that $((s', H_1', \phi'),(s', H_2', \phi')) \in \Rcal_{(\Lcal, \hat \phi, \hat s, \hat H)}$. 
    %For every $(P_1, P_2) \in \Rcal_{(\Lcal, \hat \phi, \hat s, \hat H)}$, as $P_1$ has the same formula and rank as $P_2$, 
Let us introduce the following notation for economy (where $ H \subseteq \Atoms$) :\
        $ H_{(-b+a)} = ( H \setminus \{b\}) \cup \{a\}$.\\
        We next look at the possible moves.
We perform a case analysis on $\phi$:
    \begin{itemize}
        \item $u = v$: neither $(s, H_1, \phi)$ nor $(s, H_2, \phi)$ can make a move.
        %, therefore, $\rel{(s, H_1, \phi)}{(s, H_2, \phi)}$.
        %
        \item $\phi_1 \lor \phi_2$: we have that $(s, H_1, \phi) \rightarrow (s, H_1, \phi_i)$ and $(s, H_2, \phi) \rightarrow (s, H_2, \phi_i)$ for $i \in \{1, 2\}$. 
        For any such $i$, as the history does not change, if $|H_2| \le \S $ then $((s, H_1, \phi_i), (s, H_2, \phi_i)) \in \Rcal$ for $i \in \{1,2\}$. 
        Now suppose $|H_2| = \S + 1$. 
        As $\supp(\phi_i) \subseteq \supp(\phi)$, it follows that $\supp(\phi_i, s) \cap H_1 \subseteq \supp(\phi, s) \cap H_1 \subseteq H_2 \subseteq H_1$, meaning $((s, H_1, \phi_i),(s, H_2, \phi_i)) \in \Rcal$ for $i \in \{1,2\}$ as required.
        Similar can be shown for $\phi_1 \land \phi_2$.
        \item $\fresh{x} \phi'$ the possible moves in $\fgame$ are $(s, H_1, \phi)\rightarrow (s, H_1, \phi'\{a/x\})$ for all $a \# \phi,H_1$, hence each such move can be matched by some $(s, H_2, \phi) \rightarrow (s, H_2, \phi'\{a/x\})$ and we have that $a \notin H_2 \subseteq H_1$ hence $((s, H_1, \phi'\{a/x\}),(s, H_2, \phi'\{a/x\})) \in \Rcal$. 
        \\
        Conversely, the possible moves in $\Gcal_\S$ are $(s, H_2, \phi) \rightarrow (s, H_2, \phi'\{a/x\})$ for $a \# \phi,H_2$. 
        If $a \notin H_1$, then as before this can be matched by $(s, H_1, \phi)\rightarrow (s, H_1, \phi'\{a/x\})$ with $((s, H_1, \phi'\{a/x\})$, $(s, H_2, \phi'\{a/x\})) \in \Rcal$.
        Otherwise, {and this is the more interesting case}, suppose $a \in H_1 \setminus H_2\cup\supp(\phi)$.
Then, we can select some $b \notin H_1$, meaning $b \notin H_2$ and $(a\ b)\cdot (s, H_2, \phi'\{a/x\}) = (s, H_2, \phi'\{b/x\})$. This can be matched by $(s, H_1, \phi) \rightarrow (s, H_1, \phi'\{b/x\})$ with $((s, H_1, \phi'\{b/x\}), (s, H_2, \phi'\{b/x\})) \in \Rcal$ as required. 
        \item $\bigvee_x \phi'$: 
        if $|H_2| \le \S$, then $(s, H_1, \bigvee_x \phi') \rightarrow (s, H_1, \phi'\{a/x\})$ can be matched by the move $(s, H_2, \bigvee_x \phi') \rightarrow (s, H_2, \phi'\{a/x\})$ with $((s, H_1, \phi'\{a/x\}), (s, H_2, \phi'\{a/x\})) \in \Rcal$ as required. And similarly for each $(s, H_2, \bigvee_x \phi') \rightarrow (s, H_2, \phi'\{a/x\})$.
        \\
        Suppose $|H_2| = \S+1$, and
        let $(s, H_1, \bigvee_x \phi') \rightarrow (s, H_1, \phi'\{a/x\})$ for some $a \in \Atoms$.
Note that $\supp(\phi', s) \cap H_1 \subseteq H_2 \subseteq H_1$. We do a case analysis on $a$.
        \begin{itemize}
            \item 
        If $a \in H_2$ or $a\notin H_1$, then this means $a$ is in both histories, or neither, so it suffices to again match with the move $(s, H_2, \bigvee_x \phi') \rightarrow (s, H_2, \phi'\{a/x\})$ with $((s, H_1, \phi'\{a/x\})$, $(s, H_2, \phi'\{a/x\})) \in \Rcal$. 
        \item
        If $a \in H_1 \setminus H_2$, then we no longer have 
        $((s, H_1, \phi'\{a/x\}),(s, H_2, \phi'\{a/x\})) \in \Rcal$. 
        Note that then $a\notin\supp(s,\phi')$ by definition of $\Rcal$. 
        We pick $b\in H_2\setminus\supp(s,\phi')$ and make the move
        $(s, H_2, \bigvee_x \phi') \rightarrow (s, H_2, \phi'\{b/x\})$. We have $(a\ b)\cdot(s, H_2, \phi'\{b/x\})=(s, H_{2(-b+a)}, \phi'\{a/x\})$ and
        $((s, H_1, \phi'\{a/x\}),(s, H_{2(-b+a)}, \phi'\{a/x\})) \in \Rcal$, as required.
                \end{itemize}
        Conversely, suppose $(s, H_2, \bigvee_x \phi')\to(s, H_2, \phi'\{a/x\})$. We do case analysis on $a$.
        \begin{itemize}
        \item If $a \in H_2$ or $a \notin H_1$, then
          $(s, H_1, \bigvee_x \phi')$ $ \rightarrow (s, H_1, \phi'\{a/x\})\Rcal(s, H_2, \phi'\{a/x\})$.
            \item If $a \in H_1 \setminus H_2$, then let us choose the transition $(s, H_1, \bigvee_x \phi') \to (s, H_1, \phi'\{b/x\})$ for some $b \notin H_1 \cup \supp(\phi', s)$.
            Then, we have that $(s, H_1, \phi'\{b/x\}) \nomeq (s, H_{1(-a+b)}, \phi'\{a/x\})$ and $((s, H_{1(-a+b)}, \phi'\{a/x\}),(s, H_2, \phi'\{a/x\})) \in \Rcal$ as required.
        \end{itemize}
        % Similar can be shown for $\bigwedge_x \phi'$.
        %
        \item $\diam{t, \vec a} \phi'$: in $\fgame$, the possible moves are $(s, H_1, \phi) \rightarrow (s', H_1', \phi')$ for all $(s', H_1')$ such that $(s, H_1) \xrightarrow{t, \vec a} (s', H_1')$.
        By definition, we know $\supp(s') \subseteq \supp(s) \cup \{\vec a\}$ and $H_1' = H_1 \cup \{\vec a\}$. 
        If $|H_1'| \le \S$, then this is matched by some move $(s,H_2,\phi) \rightarrow (s', H_2', \phi')$ with $H_2' = H_2 \cup \{\vec a\}$ and $((s', H_1', \phi'),(s',H_1',\phi')) \in \Rcal$. 
        Otherwise, we can see that $|H_2 \cup \{\vec a\}| > \S$. This can be matched by some move $(s,H_2,\phi) \rightarrow (s', H_2', \phi')$ with $H_1' \cap \supp(\phi', s') \subseteq H_2' \subseteq H_1'$ and $|H_2'| = \S + 1$. 
        As $|\supp(\phi',s')| \le \S$, we have that such a $H_2'$ exists and $((s', H_1', \phi'),(s',H_2',\phi')) \in \Rcal$ as required.
        \\
        Conversely, let $(s, H_2, \phi) \rightarrow (s', H_2', \phi')$ with $s\xrightarrow{t, \vec a} s'$. One of the following is the case:
        \begin{itemize}
            \item $|H_2 \cup \{\vec a\}| \le \S$, then $H_2' = H_1 \cup \{\vec a\}$. 
            This can be matched by the move $(s, H_1, \phi) \rightarrow (s', H_1', \phi')$ with $H_1' = H_1 \cup \{\vec a\}$ and $((s', H_1', \phi'),(s',H_2',\phi')) \in \Rcal$.
            \item $|H_2 \cup \{\vec a\}| > \S$ and $(H_2 \cup \{\vec a\}) \cap \supp(s', \phi') \subseteq H_2' \subseteq H_2 \cup \{\vec a\}$ with $|H_2'| = \S + 1$. 
            As $\supp(s') \subseteq \supp(s) \cup \{\vec a\}$ and $\supp(\phi) = \supp(\phi') \cup \{\vec a\}$, we have:
            \begin{align*}
                (H_1 \cup \{\vec a\}) \cap \supp(s', \phi') &= ((H_1\setminus\{\vec a\}) \cap \supp(s', \phi')) \cup (\{\vec a\} \cap \supp(s', \phi'))
                \\
                &\subseteq (H_2 \cap \supp(s', \phi')) \cup (\{\vec a\} \cap \supp(s', \phi'))\subseteq H_2' 
                % \\&= (H_2 \cup \{\vec a\}) \cap \supp(s', \phi') \subseteq H_2' 
            \end{align*}
            Thus, the move can be matched by $(s, H_1, \phi) \rightarrow (s', H_1', \phi')$ with $H_1' = H_1 \cup \{\vec a\}$ and $((s', H_1', \phi'),(s',H_2',\phi')) \in \Rcal$.
        \end{itemize}
% Similar can be shown for $\sq{t, \vec a} \phi'$.
        %
        \item $(\mu X(\vec x).\phi')(\vec a)$: in $\fgame$, the only move is $(s, H_1, \phi) \rightarrow (s, H_1, (\phi'\{\mu X(\vec x).\phi'/X\})\{\vec a/\vec x\})$. As $H_1$ does not change and $\supp(\phi, s) = \supp((\phi'\{\mu X(\vec x).\phi'/X\})\{\vec a/\vec x\})$, then this can be matched by the only move in $\Gcal_\S$, i.e.\ $(s, H_2, \phi) \rightarrow (s, H_2, (\phi'\{\mu X(\vec x).\phi'/X\})\{\vec a/\vec x\})$ with $((s, H_1, (\phi'\{\mu X(\vec x).\phi'/X\})\{\vec a/\vec x\}), (s, H_2, (\phi'\{\mu X(\vec x).\phi'/X\})\{\vec a/\vec x\})) \in \Rcal$. 
        % Similar can be shown for $(\nu X(\vec x).\phi')(\vec a)$.
    \end{itemize} 
The remaining cases can be shown similarly to the ones above.
  \end{proof}

\subsection{Proof of \cref{lem:nombisimequiv}}
\nombisimequiv*
\begin{proof}
  By symmetry, it suffices to show the left-to-right direction. Suppose $(P_1,P_2)\in R$ for some bisimulation $R$ and let $\Theta_1$ be a winning strategy for Defender from $P_1$. We assume the existence of a choice function $\iota$ such that, for each $(P_1',P_2')\in R$:
  \begin{itemize}
\item    if $P_1',P_2'$ are Attacker positions then, for each move $P_2'\to P_2''$, $\iota(P_1',P_2',P_2'')=P_1''$ such that $P_1'\to P_1''$ and $(P_1'',P_2'')\in R$; 
\item    if $P_1',P_2'$ are Defender positions then, for each move $P_1'\to P_1''$, $\iota(P_1',P_2',P_1'')=P_2''$ such that $P_2'\to P_2''$ and $(P_1'',P_2'')\in R$.
    \end{itemize}
While our strategies are positional (i.e.\ in $V_D\rightharpoonup E$), we next define a strategy $\Theta_2:V^*\times V_D\to E$, i.e.\ one that looks at the whole play. It is known that if there is such a winning strategy for Defender from a root position then there is also a positional one~\cite{EmersonJutla91}. Let us call a play $\Pi_2=P_{2,0},P_{2,1},\dots,P_{2,n}$ in $\Gcal_2$ with $P_{2,0}=P_2$ to be \emph{$(\Theta_1,R)$-compatible} if there is a unique play 
$\Pi_1=P_{1,0},P_{1,1},\dots,P_{1,n}$ in $\Gcal_1$, called \emph{the $\Gcal_1$-match of $\Pi_2$}, such that $P_{1,0}=P_1$ and, for each $i<n$, $(P_{1,i+1},P_{2,i+1})\in R$ and:
\begin{itemize}
\item if $P_{1,i},P_{2,i}$ are Attacker positions then $P_{1,i+1}=\iota(P_{1,i},P_{2,i},P_{2,i+1})$;
\item if $P_{1,i},P_{2,i}$ are Defender positions then $\Theta_1(P_{1,i})=P_{1,i}\to P_{1,i+1}$ and $P_{2,i+1}=\iota(P_{1,i},P_{2,i},P_{1,i+1})$.
\end{itemize}
We let the strategy $\Theta_2$ be defined on \emph{$(\Theta_1,R)$-compatible} plays ending in Defender positions, by setting
$\Theta_2(P_{2,0},P_{2,1},\dots,P_{2,n})=\iota(P_{1,n},P_{2,n},\Theta_1(P_{1,n}))$, where $P_{1,0},P_{1,1},\dots,P_{1,n}$ the $\Gcal_1$-match of $P_{2,0},P_{2,1},\dots,P_{2,n}$. Note that, for all $i$, $(P_{1,i},P_{2,i})\in R$.
\\
Let now $\Pi_2=P_{2,0},P_{2,1},\dots$ be a complete or infinite play in $\Gcal_2$ starting from $P_2$ and following $\Theta_2$ for Defender. We claim that every finite prefix $P_{2,0},P_{2,1},\dots,P_{2,n}$ is $(\Theta_1,R)$-compatible. Indeed, if $P_{2,0},P_{2,1},\dots,P_{2,i}$ is $(\Theta_1,R)$-compatible with $\Gcal_1$-match $P_{1,0},P_{1,1},\dots,P_{1,i}$ then, depending on the polarity of $P_{1,i},P_{2,i}$, we get a unique  $P_{1,0},P_{1,1},\dots,P_{1,i+1}$ such that $P_{2,0},P_{2,1},\dots,P_{2,i+1}$ is $(\Theta_1,R)$-compatible. Hence, if $\Pi_2$ is complete then Defender wins as the corresponding play $\Pi_1$ must be winning by hypothesis. If, on the other hand, $\Pi_2$ is infinite then we can define the infinite play $\Pi_1$ by setting, for each $n$, $P_{1,0},P_{1,1},\dots,P_{1,n}$ to be the $\Gcal_1$-match of $P_{2,0},P_{2,1},\dots,P_{2,n}$. As Defender must win $\Pi_1$, they also win $\Pi_2$.
\end{proof}

\subsection{Proof of \cref{lem:orbitGameSize}}\label{app:orbitGameSize}

In the following proofs, when discussing any sequence of names $\vec a$, it is assumed that all names in $\vec a$ are distinct.

% \subsection{Proof of \cref{lem:sizebound}}
\begin{lemma}\label{lem:sizebound}
	For every formula $\phi$ such that for every subformula $(\mu X(\vec x).\phi')(\vec u)$ of $\phi$ we have $X$ is free in $\phi'$, then $|\supp(\phi)|+2\|\phi\|+\zeta(\phi)\leq|\phi|$ where $\zeta(\phi)$ is the number of free name variables (of the form $x, y, \dots$) in $\phi$.
\end{lemma}
\begin{proof}
	We perform an induction on $\phi$ with the following base cases:
	\begin{itemize}
		\item $u=v$: we have that $|\supp(\phi)| + \zeta(\phi) = 2$, $2\|\phi\| = 0$  and $|\phi| = 2$. 
		\item $X(\vec u)$: we have that $|\supp(X(\vec u))| + \zeta(\phi) = |\vec u|$, $2\|X(\vec u)\| = 0$, $\zeta(\phi)=0$ and $|X(\vec u)| = 1 + |\vec u|$.
	\end{itemize}
	Next we perform inductive steps:
	\begin{itemize}
		\item $\phi_1 \lor \phi_2$: 
		% for all $i \in \{1,2\}$, 
		we have that 
		% $\supp(\phi_i) \subseteq |\supp(\phi)|$ and 
		$\|\phi\| = \max(\|\phi_1\|, \|\phi_2\|)$,
		and by the inductive hypothesis we have $|\supp(\phi_i)| + 2\|\phi_i\| + \zeta(\phi_i) \le |\phi_i|$.
		Additionally, we know $\supp(\phi_i) \subseteq \supp(\phi)$ and $\zeta(\phi) \le \zeta(\phi_1) + \zeta(\phi_2)$. 
		By definition, we know $|\phi| = 1 + |\phi_1| + |\phi_2|$, 
		thus, we have:
		\begin{align*}
			|\supp(\phi)| + 2\|\phi\| + \zeta(\phi)&\le (|\supp(\phi_1)| + |\supp(\phi_2)|) + 2(\|\phi_1\| + \|\phi_2\|) + \zeta(\phi_1) + \zeta(\phi_2) 
			\\&\le 1 + |\phi_1| + |\phi_2| = |\phi|
		\end{align*}
		\item $\lnot \phi'$: we have that $\|\lnot \phi'\| = \|\phi'\|$, and by the inductive hypothesis we have $|\supp(\phi')| + 2\|\phi'\| + \zeta(\phi') \le |\phi'|$. 
		Additionally, we have that $\supp(\phi') = \supp(\lnot \phi')$ and $\zeta(\phi) = \zeta(\phi')$.
		By definition, we have that $|\lnot \phi'| = 1 + |\phi'|$, thus we obtain:
		\[
		|\supp(\lnot \phi')| + 2\|\lnot \phi'\| + \zeta(\phi)= |\supp(\phi')| + 2\|\phi'\| + \zeta(\phi')\le 1 + |\phi'|
		\]
		\item $\diam{\ell} \phi'$: 
		we have that $\|\diam{\ell}\phi'\| = \|\phi'\|$, and by the inductive hypothesis we have $|\supp(\phi')| + 2\|\phi'\| + \zeta(\phi') \le |\phi'|$. 
		Additionally, as $\phi$ could have more names than $\phi'$ by up to $|\vec u|$, we have that $|\supp(\phi)| \le |\supp(\phi')| + |\vec u|$, and as no variables become free we have $\zeta(\phi) = \zeta(\phi')$.
		By definition, we have that $|\diam{t, \vec u} \phi'| = 1 + |\vec u| + |\phi'|$, thus we obtain:
		\begin{align*}
			|\supp(\diam{t, \vec u} \phi')| + 2\|\diam{t, \vec u} \phi'\| + \zeta(\phi) 
			&= |\supp(\phi')| + 2\|\phi'\| + \zeta(\phi') + |\vec u|\\
			&\le 1 + |\phi'| + |\vec u| = |\phi|
		\end{align*}
		Similar can be shown for the case of $\fresh{x}\phi'$.
		\item $\bigvee\nolimits_x \phi'$: we have that $\|\bigvee\nolimits_x \phi'\| = 1 + \|\phi'\|$, and by the inductive hypothesis we have $|\supp(\phi')| + 2\|\phi'\| + \zeta(\phi') \le |\phi'|$.
		Additionally, we have that $\supp(\bigvee\nolimits_x \phi') = \supp(\phi')$ and as $x$ is free in $\phi'$ (but not in $\phi$) we have $\zeta(\phi) = \zeta(\phi') - 1$.
		By definition, we have that $|\bigvee\nolimits_x \phi'| = 2 + |\phi'|$, thus we obtain:
		\begin{align*}
			|\supp(\bigvee\nolimits_x \phi')| + 2\|\bigvee\nolimits_x  \phi'\| + \zeta(\phi) &= |\supp(\phi')| + 2(\|\phi'\| + 1) + (\zeta(\phi') - 1)\\
			&= |\supp(\phi')| + 2\|\phi'\| + \zeta(\phi') + 1\\
			&\le 1 + |\phi'| = |\phi|  
		\end{align*}
		\item $(\mu X(\vec x).\phi')(\vec u)$: we have that $\|(\mu X(\vec x).\phi')(\vec u)\| = \|\phi'\| + n$ where $n = |\vec u| = |\vec x|$, and by the inductive hypothesis we have that $|\supp(\phi')| + 2\|\phi'\| + \zeta(\phi') \le |\phi'|$. 
		Additionally, $\phi$ will have at most $n$ many more names than $\phi'$, and so $|\supp(\phi)| \le |\supp(\phi')| + n$. 
		% we have that $\supp((\mu X(\vec x).\phi')(\vec u)) = \supp(\phi')$
		and as each $x \in \{\vec x\}$ is free in in $\phi'$ (but not in $\phi$), then $\zeta(\phi) = \zeta(\phi') - n$.
		By definition we have that $|(\mu X(\vec x).\phi')(\vec u)| = 1 + |\phi'| + 2n$, thus we obtain:
		\begin{align*}
			|\supp((\mu X(\vec x).\phi')(\vec u))| + 2\|(\mu X(\vec x).\phi')(\vec u)\| + \zeta(\phi)
			&\le (|\supp(\phi')| + n) + 2(\|\phi'\| + n) + \zeta(\phi') - n\\
			% &= |\supp(\phi')| + 2\|\phi'\| + \zeta(\phi') + |\vec a|\\
			&\le |\phi'| + 2n + 1 = |\phi|\ 
		\end{align*}
	\end{itemize}
\end{proof}

\begin{lemma}\label{lem:genPartialPermF}
	Let $\phi_0$ be a \fhml\ formula, $\vec a_0, \vec b_0 \in \Atoms^{n}$ where $|\vec a_0| = |\vec b_0| = n$ and let $(\phi, \gamma_1, \theta_1), (\psi, \gamma_2, \theta_2) \in \clos(\phi_0)$.
%	Let 
	Then, the time complexity of finding $\move{\vec a}{\vec b}$ 
	such that $\move{\vec a}{\vec b} \cdot \phi\{\theta_1\}\{\gamma_1\} = \psi\{\theta_2\}\{\gamma_2\}$, $\vec a_0$ is a subsequence of $\vec a$ and $\vec b_0$ is a subsequence of $\vec b$ is bounded by $O(|\phi_0)|)$.
%	\ntnote{I think that everywhere here you assume that sequences like $\vec a,\vec b$ have distinct names?}
\end{lemma}

\begin{proof}
	Let us define the function $\Gamma$ that takes some $(\phi, \theta_1, \gamma_1, \psi, \theta_2, \gamma_2, \vec a, \vec b)$ as parameters and returns some $(\vec a', \vec b')$ such that:
	\[
	\begin{aligned}
		\vec a' &= \vec a, \vec a'' &[\text{s.t.} \{\vec a\} \cap \{\vec a''\} = \emptyset]\\
		\vec b' &= \vec b, \vec b'' &[\text{s.t.} \{\vec b\} \cap \{\vec b''\} = \emptyset]\\
	\end{aligned}
	\]
	\vspace{-1.5ex}
	\[
	\move{\vec a'}{\vec b'} \cdot \phi'\{\theta_1\}\{\gamma_1\} = \psi\{\theta_2\}\{\gamma_2\}
	\]
	or returns $\NO$ if not possible. 
	Moreover, we define the function $\update$ that takes some $(\vec a, \vec b, \vec c, \vec d)$, such that $|\vec a| = |\vec b|$ and:
	\begin{itemize}
		\item if $|\vec c| \neq |\vec d|$ then return $\NO$,
		\item if $\vec c = \varepsilon$ and $\vec d = \varepsilon$, then return $(\vec a, \vec b)$,
		\item if $c_1 \in \{\vec a\}$ and $d_1 \in \{\vec b\}$, then 
		we
		return the result of $$\update(\vec a, \vec b, \vec{c}_{2:}, \vec{d}_{2:})$$ (where for any $\vec x$, $\vec{x}_{2:}$ is $\vec{x}$ excluding the first term) if $\move{\vec a}{\vec b}(c_1) = d_1$, and returns $\NO$ otherwise,
		\item if $c_1 \notin \{\vec a\}$ and $d_1 \notin \{\vec b\}$, then we 
		return the result of $$\update(\vec a(c_1), \vec b(d_1), \vec{c}_{2:}, \vec{d}_{2:})$$
		\item if $c_1 \in \{\vec a\}$ and $d_1 \notin \{\vec b\}$ (or vice versa), then we return $\NO$.
	\end{itemize}
	We construct $\Gamma(\phi, \theta_1, \gamma_1, \psi, \theta_2, \gamma_2, \vec a, \vec b)$ inductively on $\phi, \psi$ with the following base cases:
	\begin{itemize}
		\item $\phi = (u_1=u_2)$: 
		if $\psi$ is not of the form $(v_1 = v_2)$, then we return $\NO$. 
		Else, we return $\update(\vec a, \vec b, (u_1,u_2)\{\gamma_1\}, (v_1,v_2)\{\gamma_2\})$.
		\item $\phi = X(\vec u)$: there must be a unique $\theta \in \theta_1$  such that $X(\vec c)\in\dom(\theta)$, where $\vec c=\vec u\{\gamma_1\}$.
		There are two valid possibilities:
		\begin{enumerate}
			\item $\psi$ is some $X(\vec v)$.
			\item $\psi$ is some $(\sigma X(\vec x).\psi')(\vec v)$ and $\theta(X(\vec c)) = (\sigma X(\vec x).\phi')(\vec c)$ for $\sigma \in \{\mu, \nu\}$.
		\end{enumerate}
		In either case, 
		let $\vec y_1$ be an ordering of the free variables of $\phi'$, and $\vec y_2$ be an ordering of the free variables of $\psi'$. 
		We return $\update(\vec a, \vec b, (\vec y_1\{\gamma_1\}, \vec c), (\vec y_2, \vec v)\{\gamma_2\})$.
		This base case is similar for when $\psi$ is of the form $X(\vec v)$.
	\end{itemize}
	Next we perform inductive steps:
	\begin{itemize}
		\item if $\phi = \phi_1 \lor \phi_2$ and $\psi \neq \psi_1 \lor \psi_2$, then we return $\NO$. 
		Otherwise, by the inductive hypothesis, we check if $\Gamma(\phi_1, \theta_1, \gamma_1, \psi_2, \theta_2, \gamma_2, \vec a, \vec b)$ is some $(\vec a', \vec b')$. 
		If not, then we return $\NO$.
		If so, by the inductive hypothesis, we return $\Gamma(\phi_2, \theta_1, \gamma_1,\psi_2, \theta_2, \gamma_1, \vec a', \vec b')$.
		Similar can be shown for when $\phi' = \phi_1 \land \phi_2$.
		\item if $\phi = \diam{t, \vec c}\phi'$ and $\psi\neq \diam{t, \vec d}\psi'$ then we return $\NO$. 
		Else, we check the result of $\update(\vec a, \vec b, \vec c\{\gamma_1\}, \vec d\{\gamma_2\})$.
		If the result is $\NO$, then we return $\NO$. 
		Otherwise if the result is some $(\vec a', \vec b')$, then by the inductive hypothesis, we return $ \Gamma(\phi', \theta_1, \gamma_1, \psi', \theta_2,\gamma_2, \vec a', \vec b')$.
		The case where $\phi = \sq{t, \vec c}\phi'$ is similar.
		\item if $\phi = \bigvee_x \phi'$ and $\psi \neq \bigvee_x \psi'$, then we return $\NO$. 
		Else, we select some $c \notin \{\vec a\} \cup \{\vec b\} \cup \supp(\phi, \psi)$ and, by the inductive hypothesis,
		we examine the result of $$\Gamma(\phi', \theta_1, (\gamma_1, \{c/x\}), \psi', \theta_2, (\gamma_2, \{c/x\}), (\vec a, c), (\vec b, c))$$
		If this is $\NO$, then we return $\NO$. 
		Else if the result is some $(\vec a', \vec b')$, then we return $(\vec a'', \vec b'')$, where $\vec a''$ is $\vec a'$ excluding $c$, and similarly for $\vec b''$.
		The case where $\phi = \bigwedge_x \phi'$ and $\phi = \fresh{x}\phi'$ is similar.
		\item if $\phi = (\mu X(\vec x).\phi')(\vec u)$ and $\psi \neq (\mu X(\vec x).\psi')(\vec v)$, then we return $\NO$.
		Otherwise, we check the result of $\update(\vec a, \vec b, \vec u\{\gamma_1\}, \vec v\{\gamma_2\})$. If this is $\NO$, then we return $\NO$. 
		Otherwise if this is some $(\vec a', \vec b')$, then 
		by the inductive hypothesis, we return $$\Gamma(\phi', \theta_1, (\gamma_1, \{\vec u\{\gamma_1\} / \vec x\}), \psi, \theta_2, (\gamma_2, \{\vec v\{\gamma_2\} / \vec x\}), \vec a', \vec b')$$
		Note the case where $\psi$ is of the form $X(\vec v)$ is covered by the base case.
		The case where $\phi = (\nu X(\vec x).\phi')(\vec u)$ is similar.
	\end{itemize}
	For any $(\phi, \gamma_1, \theta_1), (\psi, \gamma_2, \theta_2) \in \clos(\phi_0)$, to determine if there exists a $\vec a, \vec b$ such that $\move{\vec a}{\vec b}\cdot \phi\{\theta_1\}\{\gamma_1\} = \psi\{\theta_2\}\{\gamma_2\}$, it suffices to check if $\Gamma = \Gamma(\phi, \theta_1, \gamma_1, \psi, \theta_2, \gamma_2, \vec a_0, \vec b_0)$ is some $(\vec a', \vec b')$.
	It remains to determine the time complexity of constructing $\Gamma$.  
	The number of inductive calls to $\Gamma$ is bounded by $\min(|\phi|, |\psi|)$, and each such inductive call will internally call $\update$ at most once.
	The complexity of any call $\update(\vec a', \vec b', \vec c, \vec d) = (\vec a'', \vec b'')$ will be $O(|\vec a''| - |\vec a'|)$.  
	Over the course of constructing $\Gamma$, this means the cumulative complexity of all calls to $\update$ will be in $O(\min(|\supp(\phi)| + \|\phi\|, |\supp(\psi)| + \|\psi\|))$
	Thus, the time complexity of constructing $\Gamma$ is:
	\begin{align*}
		&O(\min(|\phi|, |\psi|) + \min(|\supp(\phi)| + \|\phi\|, |\supp(\psi)| + \|\psi\|)) \\
		\le\ &O(|\phi_0| + \supp(\phi_0) + \|\phi_0\|) \\
		\overset{\text{Lemma~\ref{lem:sizebound}}}{\le}\ &O(|\phi_0| + |\phi_0|) \le O(|\phi_0|)
	\end{align*} 
\end{proof}

\subsection*{Proof of \cref{lem:orbitGameSize}(2)}\label{app:complexityboundedgame}
\begin{lemma}
  Given a setup $(\Lcal,\phi_0,s_0,H_0)$ of grade $N$ and $\Gcal=\Gcal_N(\Lcal, \phi_0, s_0, H_0)$,
 the time complexity of constructing $\gorb(\Gcal)$ is in
 	\begin{align*}
	&O(
 	(|\phi_0| + \Phi_\Lcal(|\supp(\phi_0)| + \|\phi_0\|) + \S)
 	\cdot
 	\max(N + 2, |\Scal|
 	) \\&\ \cdot 
 	(|\orb(\Scal)| \cdot|\phi_0|\cdot \frac{(|\supp(\phi_0)|+\|\phi_0\|)!}{|\supp(\phi_0)|!}\cdot (\S + 1)^{|\supp(\phi_0)| + \|\phi_0\| + 1})
 	^2)
 \end{align*}
  \end{lemma}
\begin{proof}
%	We begin by defining the function $T$, which takes as parameter $((\phi', \gamma', \theta'), s', H')$ and returns $\orbs(\phi'\{\theta'\}\{\gamma'\}, s', H')$, so $(\phi', \gamma', \theta') \in \clos(\phi_0)$.
	We begin by defining the set $T$ of orbit representatives of $\pos(\Gcal)$, that is, a triple $((\phi, \gamma, \theta), s, H)$ represents $\orbs(\phi\{\theta\}\{\gamma\}, s, H)$.
	We build the game inductively on $\pos(\Gcal)$, with the base case being that $\orbs(\phi_0, s_0, H_0') \in \pos(\Gcal)$ for some $H_0'\wellbound{N}(s_0,\phi_0,H_0)$, and initially $T = \{((\phi_0, \varepsilon, \varepsilon), s_0, H_0')\}$.
	
	Suppose we are examining a position $\orbs( \phi,  s,  H)$, 
	with representative $((\phi_i, \gamma_i, \theta_i), s_i, H_i)$, so $(\phi_i, \gamma_i, \theta_i), s_i, H_i) \in T$. 
	For all moves $(\phi_i\{\theta_i\}\{\gamma_i\}, s_i, H_i) \to (\phi'\{\theta'\}\{\gamma'\}, s', H')$ such that $(\phi', \gamma', \theta') \in \clos(\phi_0)$, we need to check if $\orbs(\phi'\{\theta'\}\{\gamma'\}, s', H') \in \pos(\Gcal)$. 
%	By \lemmapoint{lem:freshClosProps}{2}, there exists a triple $(\phi', \gamma', \theta') \in \clos(\phi_0)$ such that $\phi'\{\theta'\}\{\gamma'\} = \phi$, and we assume this triple can be found in constant time.
%	\ntnote{\scriptsize there is no guarantee that is can be found efficiently. But, instead of checking the transitions of $(\phi_i\{\theta_i\}\{\gamma_i\}, s_i, H_i) \to (\phi, s', H')$ we can look directly at transitions $(\phi_i\{\theta_i\}\{\gamma_i\}, s_i, H_i) \to (\phi'\{\theta'\}\{\gamma'\}, s', H')$ where $(\phi',\gamma',\theta')$ is a child of $(\phi_i,\gamma_i,\theta_i)$ in the clos construction. This can be done in constant time}
	Thus, we need to check if for any $((\phi_1, \gamma_1, \theta_1), s_1, H_1) \in T$:
	$$\pi \cdot ((\phi', \gamma', \theta'), s', H') = ((\phi_1, \gamma_1, \theta_1), s_1, H_1)$$ 
	for some permutation $\pi$ (i.e., they are nominally equivalent).
%	This means we must see if there is a permutation $\pi \in \Perm$ such that $\pi \cdot ((\phi', \gamma', \theta'), s', H') = ((\phi_1,\gamma_1, \theta_1), s_1, H_1)$.
	For such a $\pi$ to exist, we must satisfy the following conditions:
	\begin{enumerate}
		\item $|H'| = |H_1|$.
		\item There must exist a minimal $\move{\vec a}{\vec b}$ such that $\move{\vec a}{\vec b} \cdot (\phi', \gamma', \theta') = (\phi_1, \gamma_1, \theta_1)$; 
	 	by Lemma~\ref{lem:genPartialPermF}, this can be constructed in $O(|\phi_0|)$.
		\item From $\vec a, \vec b$ as in the previous condition, the oracle $\Psi_\Lcal(\vec a, \vec b, s_1, s')$ must find a permutation.
		Each such check is in $O(\Phi_\Lcal(|\vec a|)) \le O(\Phi_\Lcal(|\supp(\phi)| + \|\phi\|)) \le O(\Phi_\Lcal(|\supp(\phi_0)| + \|\phi_0\|))$. 
%		which by Lemma~\ref{lem:sizebound}, is in $O(\Phi_\Lcal(|\phi_0|))$
		%
		\item The number of redundant names in the histories must be the same, that is, $|H' \setminus \supp(\phi, s')| = |H_1 \setminus \supp(\phi_1\{\theta\}\{\gamma_1\}, s_1)|$. 
		Each such check is in $O(|\supp(\phi_0)| + \|\phi_0\| + \regindex(\Scal)) = O(N)$.
	\end{enumerate}
%	e require that $|H'| = |H_1|$, otherwise there is no such permutation $\pi$.
%	 We can first construct a minimal permutation $\move{\vec a}{\vec b}$ such that $\move{\vec a}{\vec b} \cdot (\phi', \gamma', \theta') = (\phi_1, \gamma_1, \theta_1)$; 
%	 by Lemma~\ref{lem:genPartialPermF}, this can be constructed in $O(|\phi_0|)$.
%	If there is no such permutation, then $((\phi', \gamma', \theta'), s', H') $ is not nominally equivalent to $((\phi_1,\gamma_1, \theta_1), s_1, H_1)$.
%	Otherwise, the oracle $\Psi_\Lcal(\vec a, \vec b, (q_1, \rho_1), s')$ can attempt to find a permutation $\pi'$. 
%	If $\Psi_\Lcal(\vec a, \vec b, (q_1, \rho_1), s')$ returns $\NO$, then there is no such $\pi$;
%	otherwise, it remains to examine $H', H_1$.
%	We have that $\move{\vec a}{\vec b}\cdot\supp(\phi, s') = \supp(\phi_1\{\theta^*\}\{\gamma_1\}, s_1)$, and so it suffices to check if both $H'$ and $H_1$ have the same number of redundant names.
%	Thus, we check $|H' \setminus \supp(\phi, \rho')| = |H_1 \setminus \supp(\phi_1\{\theta^*\}\{\gamma_1\}, \rho_1)|$ where each such check is in $O(|\supp(\phi_0)| + \|\phi_0\| + \regindex(\Scal))$.
%	If this is not the case, then there is no such permutation.
	Each such nominal equivalence check is in:
	\[
	O(
	|\phi_0| + \Phi_\Lcal(|\supp(\phi_0)| + \|\phi_0\|) + N
	)
	\]
%		&O(|\phi_0| + |\supp(\phi_0)| + \|\phi_0\| + \regindex(\Scal) + \Phi_\Lcal(|\vec a|, |s'|, |(q_1, \rho_1)|)) \\\overset{\text{Lemma~\ref{lem:sizebound}}}{=}\ & O(2|\phi_0| + \regindex(\Scal) + \Phi_\Lcal(\min(|\supp(\phi')|, |\supp(\phi_1)|), \max(\Lcal), \max(\Lcal))
%		\end{align*}
	If there is no such $((\phi_1, \gamma_1, \theta_1), s_1, H_1) \in T$ that satisfies all above conditions for $((\phi', \gamma', \theta'), s', H')$, then we add $(\phi', \gamma', \theta'), s', H')$ to $T$.
%	\[
%		T((\phi', \gamma', \theta'), s', H')) = \orbs(\phi, s', H')
%	\]
		
	The moves that can be made from each $\orbs(\phi, s, H)$ would need to be checked at most once, thus the maximum number of positions that will need to be checked is bounded by the maximum number of moves that reach unique targets.
%	\ntnote{\tiny not sure about unique targets. Note that, somewhere in this calculation, the number of labels used in the game should appear, which are the number of tags times the names they can use ($N+2$) to the power of the tag's arity}
	This can be quantified as $\max(\supp(\phi, s) + 2, \ndbb(\Lcal))$, that is:
	\begin{enumerate}
		\item The maximum number of unique names in the support of $\phi, s$, plus one additional name that could be in the history, and one fresh name.
		\item The maximum number of unique moves at any given position.
	\end{enumerate}
	The number of elements that a nominal equivalence check required for each such target is bounded by $|\pos(\Gcal)|$.
	Finally, the above process is done for every position in $\pos(\Gcal)|$, 
	 thus the final complexity is in:
	\begin{align*}
 		&O(
 		(|\phi_0| + \Phi_\Lcal(|\supp(\phi_0)| + \|\phi_0\|) + \S) 			\cdot
 			\max(\supp(\phi, s) + 2, \ndbb(\Lcal)
 		)\cdot |\pos(\Gcal)|^2)
 		\\
 		\le\ &O(
 		(|\phi_0| + \Phi_\Lcal(|\supp(\phi_0)| + \|\phi_0\|) + \S)
 		\cdot
 		\max(N + 2, |\Scal|
 		) \\&\ \cdot 
 		(|\orb(\Scal)| \cdot|\phi_0|\cdot \frac{(|\supp(\phi_0)|+\|\phi_0\|)!}{|\supp(\phi_0)|!}\cdot (\S + 1)^{|\supp(\phi_0)| + \|\phi_0\| + 1}\cdot(1+o(1)))
 		^2)
% 		\\
% 		\overset{\text{Lemma~\ref{lem:sizebound}}}{\le}\ &O(
% 		(|\phi_0| + \Phi_\Lcal(|\phi_0|) + \S)
% 		\cdot
% 		\max(|\phi_0| + \regindex(\Scal) + 2, |\fstates|
% 		)
% 		\\
% 		=\ &O(
% 		(|\phi_0| + \Phi_\Lcal(|\phi_0|) + \S)
% 		\cdot
% 		\max(\S + 2, |\fstates|
% 		)
 	\end{align*}
 \end{proof}

\subsection{Proof of \cref{thm:MCinFHML}}
\MCinFHML*
  \begin{proof}
To verify whether $(s, H) \in \fstates$ satisfies $\phi_0$, we must check whether $(s, H) \in \sem[\xi_0]{\phi_0}$ where $\xi_0$ is an empty $\Ucal$-variable assignment. 
First, we set $\phi = !(\phi_0)$, which (by Lemma~\ref{lem:negfreesize}) will have size bounded by $|\phi_0|$. 
From this, we can construct the orbits parity game $\gorb(\Gcal_\S)$ of $\Gcal_\S = \Gcal_\S(\Lcal, \phi, s, H)$ where $\S = |\phi_0| + \|\phi_0\| + \regindex(\Scal)$, 
By Theorem~\ref{thm:FnomSatPar}, we have that if $(s, H) \in \sem[\xi_0]{\phi_0}$ then defender has a winning strategy from $\fgame = \fgame(\Lcal, \phi, s, H)$.
Using Lemma~\ref{lem:boundingisnombisim}, the relation $\Rcal_{(\Lcal, \phi, s, H)} \subseteq \pos(\fgame) \times \pos(\Gcal_\S)$ is a nom-bisimulation, 
and by Lemma~\ref{lem:nombisimequiv} then defender wins in $\Gcal$ from $(s, H, \xi_0, \phi)$ iff they win in $\Gcal_\S$ from $(s, H, \xi_0, \phi)$.
Similarly, using Lemmaa~\ref{lem:parityBisim} and Corollary~\ref{cor:winningNom} then defender wins from $\Gcal_\S$ from $(s, H, \xi_0, \phi)$ iff defender wins in $\gorb(\Gcal_\S)$ from $\orbs(s, H, \xi_0, \phi)$. 
%and by Lemma~\ref{lem:bisimfromnomeq} we know that $\Rcal' = \{(\orb(P_1), \orb(P_2))\mid (P_1,P_2) \in \Rcal\}$ is a bisimulation between $\gorb(\fgame)$ and $\gorb(\Gcal_\S)$.
%Using Lemma~\ref{lem:parityBisim} we have that $\fgame$ is bisimilar to $\gorb(\fgame)$, by Lemma~\ref{lem:parityBisimWinning} if defender has a winning strategy from $\fgame$ then defender must also have a winning strategy in $\gorb(\fgame)$, and using the same lemma, the defender would hence also have a winning strategy in $\gorb(\Gcal_\S)$. 
% Using Lemma~\ref{lem:parityBisim} and Lemma~\ref{lem:parityBisimWinning}, if defender has a winning strategy from $\Gcal$ then defender must also have a winning strategy in $\Gcal'$, hence it suffices to find the winning regions of $\Gcal'$. 
% Let us set 
% % $n_{\Ccal_\Acal} = \max\{(|\supp(s)|) \mid s \in \Ccal_\Acal\}$, 
% $n_{\phi_0} = |\supp(\phi_0)| + \|\phi_0\|$, and let $n_1 = \max(\kappa, n_\phi_0), n_2 = \min(\kappa, n_\phi_0)$.
Then, by Lemma~\ref{lem:fnomparsize}, $\gorb(\Gcal_\S)$ can be constructed in:
    \begin{align*}
	&O(
	(|\phi_0| + \Phi_\Lcal(|\supp(\phi_0)| + \|\phi_0\|) + \S)
	\cdot
	\max(N + 2, |\Scal|
	) \\&\ \cdot 
	(|\orb(\Scal)| \cdot|\phi_0|\cdot \frac{(|\supp(\phi_0)|+\|\phi_0\|)!}{|\supp(\phi_0)|!}\cdot (\S + 1)^{|\supp(\phi_0)| + \|\phi_0\| + 1}\cdot(1+o(1)))
	^2)
\end{align*}
%\begin{align*}
%	&O((2|\phi| + \Phi_\Lcal(|\phi|) + \regindex(\Scal))\cdot\max(|\phi| + \regindex(\Scal) + 2, |\fstates|))
%	\\
%	\le\ &O((2|\phi_0| + \Phi_\Lcal(|\phi_0|) + \regindex(\Scal))\cdot\max(|\phi_0| + \regindex(\Scal) + 2, |\fstates|))
%\end{align*}
and has size:
\[
n \cdot|\phi_0|\cdot \frac{(|\supp(\phi_0)|+\|\phi_0\|)!}{|\supp(\phi_0)|!}\cdot (\S + 1)^{|\supp(\phi_0)| + \|\phi_0\| + 1}\cdot(1+o(1))
\]
The winning regions of $\gorb(\Gcal_\S)$ can be calculated in time bounded by (cf.
~\cite{CaludeJKLS22}):
%\begin{align*}
\[
O((n \cdot|\phi_0|\cdot \frac{(|\supp(\phi_0)| + \|\phi_0\|)!}{|\supp(\phi_0)|!}\cdot (\S + 1)^{|\supp(\phi_0)| + \|\phi_0\| + 1})^{\log(d) + 6})
\]
% &O((n \cdot|\phi_0|\cdot \frac{(|\phi_0|)!}{|\supp(\phi_0)|!}\cdot (\S + 1)^{|\phi_0| + 1}(1 + o(1)))^{\log(d) + 6}\\
% &\ + (2|\phi_0| + \Phi_\Lcal(|\phi_0|) + \regindex(\Scal))\cdot\max(|\phi_0| + \regindex(\Scal) + 2, |\fstates|))
% \\
% \le\ &O((n \cdot|\phi_0|\cdot \frac{(|\phi_0|)!}{|\supp(\phi_0)|!}\cdot (\S + 1)^{|\phi_0| + 1}(1 + \epsilon))^{\log(d) + 6} 
% \\&\ + (2|\phi_0| + \Phi_\Lcal(|\phi_0|) + \regindex(\Scal))\cdot\max(|\phi_0| + \regindex(\Scal) + 2, |\fstates|))
%\end{align*}

And so, the final complexity is in:
\begin{align*}
	&O((n \cdot|\phi_0|\cdot \frac{(|\supp(\phi_0)| + \|\phi_0\|)!}{|\supp(\phi_0)|!}\cdot (\S + 1)^{|\supp(\phi_0)| + \|\phi_0\| + 1})^{\log(d) + 6}
	\\ &\ + (
	(|\phi_0| + \Phi_\Lcal(|\supp(\phi_0)| + \|\phi_0\|) + \S)
	\cdot
	\max(N + 2, |\Scal|
	) \\&\ \cdot 
	(n \cdot|\phi_0|\cdot \frac{(|\supp(\phi_0)|+\|\phi_0\|)!}{|\supp(\phi_0)|!}\cdot (\S + 1)^{|\supp(\phi_0)| + \|\phi_0\| + 1})
	^2)
	)
	\\
	=\ &O((n \cdot|\phi_0|\cdot \frac{(|\phi_0|)!}{|\supp(\phi_0)|!}\cdot (\S + 1)^{|\phi_0| + 1})^{\log(d) + 6}
	\\ &\ + (
	(|\phi_0| + \Phi_\Lcal(|\phi_0|) + \S)
	\cdot
	\max(N + 2, |\Scal|
	) \\&\ \cdot 
	(n \cdot|\phi_0|\cdot \frac{(|\phi_0|)!}{|\supp(\phi_0)|!}\cdot (\S + 1)^{|\phi_0| + 1})
	^2)
	)
\end{align*}
\end{proof}